\documentclass{IEEEtran}
\usepackage[T1]{fontenc}
\usepackage[latin9]{inputenc}
\usepackage{array}
\usepackage{verbatim}
\usepackage{float}
\usepackage{amsmath}
\usepackage{amssymb}
\usepackage{graphicx}
\usepackage{hyperref}
\usepackage{breakurl}

\makeatletter

\providecommand{\tabularnewline}{\\}
\floatstyle{ruled}
\newfloat{algorithm}{tbp}{loa}
\providecommand{\algorithmname}{Algorithm}
\floatname{algorithm}{\protect\algorithmname}

\@ifundefined{date}{}{\date{}}
\usepackage{cite}\usepackage{amsthm}\usepackage{dsfont}\usepackage{array}\usepackage{mathrsfs}\usepackage{comment}

\usepackage{breakurl}

\makeatother

\begin{document}
\theoremstyle{plain}\newtheorem{lemma}{\textbf{Lemma}}\newtheorem{theorem}{\textbf{Theorem}}\newtheorem{corollary}{\textbf{Corollary}}\newtheorem{assumption}{\textbf{Assumption}}\newtheorem{example}{\textbf{Example}}\newtheorem{definition}{\textbf{Definition}}

\theoremstyle{definition}

\theoremstyle{remark}\newtheorem{remark}{\textbf{Remark}}

\title{Robust Spectral Compressed Sensing via \\
 Structured Matrix Completion}

\author{Yuxin Chen, \emph{Student Member, IEEE}, and Yuejie Chi, \emph{Member,
IEEE}%
\thanks{Y. Chen is with the Department of Electrical Engineering, Stanford
University, Stanford, CA 94305, USA (email: yxchen@stanford.edu).}
\thanks{Y. Chi is with Department of Electrical and Computer Engineering and
Department of Biomedical Informatics, The Ohio State University, Columbus,
OH 43210, USA (email: chi.97@osu.edu).}
\thanks{Preliminary results of this
work have been presented at the 2013 International Conference on Machine
Learning (ICML) and the 2013 Signal Processing with Adaptive Sparse
Structured Representations Workshop (SPARS). Manuscript date: \today. %
} }
\maketitle
\begin{abstract}
The paper explores the problem of \emph{spectral compressed sensing},
which aims to recover a spectrally sparse signal from a small random
subset of its $n$ time domain samples. The signal of interest is
assumed to be a superposition of $r$ multi-dimensional complex sinusoids,
while the underlying frequencies can assume any \emph{continuous}
values in the normalized frequency domain. Conventional compressed
sensing paradigms suffer from the basis mismatch issue when imposing
a discrete dictionary on the Fourier representation. To address this
issue, we develop a novel algorithm, called \emph{Enhanced Matrix
Completion (EMaC)}, based on structured matrix completion that does
not require prior knowledge of the model order. The algorithm starts
by arranging the data into a low-rank enhanced form exhibiting multi-fold
Hankel structure, and then attempts recovery via nuclear norm minimization.
Under mild incoherence conditions, EMaC allows perfect recovery as
soon as the number of samples exceeds the order of $r\log^{4}n$,
and is stable against bounded noise. Even if a constant portion of
samples are corrupted with arbitrary magnitude, EMaC still allows
exact recovery, provided that the sample complexity exceeds the order
of $r^{2}\log^{3}n$. Along the way, our results demonstrate the power
of convex relaxation in completing a low-rank multi-fold Hankel or
Toeplitz matrix from minimal observed entries. The performance of
our algorithm and its applicability to super resolution are further
validated by numerical experiments. \end{abstract}

\section{Introduction}

\subsection{Motivation and Contributions}

A large class of practical applications features high-dimensional
signals that can be modeled or approximated by a superposition of
spikes in the spectral (resp. time) domain, and involves estimation
of the signal from its time (resp. frequency) domain samples. Examples
include acceleration of medical imaging \cite{Lustig2007sparse},
target localization in radar and sonar systems \cite{potter2010sparsity},
inverse scattering in seismic imaging \cite{Borcea2002imaging}, fluorescence
microscopy \cite{schermelleh2010guide}, channel estimation in wireless
communications \cite{chi2013compressive}, analog-to-digital conversion
\cite{tropp2010beyond}, etc. The data acquisition devices, however,
are often limited by hardware and physical constraints, precluding
sampling with the desired resolution. It is thus of paramount interest
to reduce sensing complexity while retaining recovery accuracy.

In this paper, we investigate the \emph{spectral compressed sensing}
problem, which aims to recover a spectrally sparse signal from a small
number of randomly observed time domain samples. The signal of interest
$x\left(\boldsymbol{t}\right)$ with ambient dimension $n$ is assumed
to be a weighted sum of multi-dimensional complex sinusoids at $r$
distinct frequencies $\{\boldsymbol{f}_{i}\in[0,1)^{K}:1\leq i\leq r\}$,
where the underlying frequencies can assume any continuous values
on the unit interval.

Spectral compressed sensing is closely related to the problem of \emph{harmonic
retrieval}, which seeks to extract the underlying frequencies of a
signal from a collection of its time domain samples. %This spans many signal processing applications including radar
%systems \cite{NionSid2010}, array imaging systems \cite{Borcea2002imaging},
%channel sensing in wireless communication \cite{SayeedAazhang1999},
%etc. %In fact, if the time domain representation of an signal can be estimated accurately, then its underlying frequencies can be identified using harmonic super-resolution methods like ESPRIT \cite{RoyKailathESPIRIT1989}.
Conventional methods for harmonic retrieval include Prony's method
\cite{prony1795essai}, ESPRIT \cite{RoyKailathESPIRIT1989}, the
matrix pencil method \cite{HuaSarkar1990}, the Tufts and Kumaresan
approach \cite{TuftsKumaresan1982}, the finite rate of innovation
approach \cite{vetterli2002sampling,gedalyahu2011multichannel}, etc.
These methods routinely exploit the \textit{shift invariance} of the
harmonic structure, namely, a consecutive segment of time domain samples
lies in the same subspace irrespective of the starting point of the
segment. However, one weakness of these techniques if that they require
prior knowledge of the model order, that is, the number of underlying
frequency spikes of the signal or at least an estimate of it. Besides,
these techniques heavily rely on the knowledge of the noise spectra,
and are often sensitive against noise and outliers \cite{DragVetterliBlu2007}.

Another line of work is concerned with Compressed Sensing (CS) \cite{CandRomTao06,Don2006}
over a discrete domain, which suggests that it is possible to recover
a signal even when the number of samples is far below its ambient
dimension, provided that the signal enjoys a sparse representation
in the transform domain. In particular, tractable algorithms based
on convex surrogates become popular due to their computational efficiency
and robustness against noise and outliers \cite{CandesRombergTao2006Stable,li2011compressed}.
Furthermore, they do not require prior information on the model order.
Nevertheless, the success of CS relies on sparse representation or
approximation of the signal of interest in a finite discrete dictionary,
while the true parameters in many applications are actually specified
in a \emph{continuous} dictionary. The \emph{basis mismatch} between
the true frequencies and the discretized grid \cite{Chi2011sensitivity}
results in loss of sparsity due to spectral leakage along the Dirichlet
kernel, and hence degeneration in the performance of conventional
CS paradigms. %Cautions need to be exercised when imposing a discrete dictionary on continuous frequencies, since nature never poses the frequencies on the pre-determined grid, no matter how fine the grid is \cite{Chi2011sensitivity,duarte2012spectral}. 

%where $r$ and $n$ denote respectively the spectral sparsity and the ambient dimension
In this paper, we develop an algorithm, called \emph{Enhanced Matrix
Completion (EMaC)}, that simultaneously exploits the shift invariance
property of harmonic structures and the spectral sparsity of signals.
Inspired by the conventional matrix pencil form \cite{Hua1992}, EMaC
starts by arranging the data samples into an enhanced matrix exhibiting
$K$-fold Hankel structures, whose rank is bounded above by the spectral
sparsity $r$. This way we convert the spectral sparsity into the
low-rank structure without imposing any pre-determined grid. EMaC
then invokes a nuclear norm minimization program to complete the enhanced
matrix from partially observed samples. When a small constant proportion
of the observed samples are corrupted with arbitrary magnitudes, EMaC
solves a weighted nuclear norm minimization and $\ell_{1}$ norm minimization
to recover the signal as well as the sparse corruption component.

The performance of EMaC depends on an incoherence condition that depends
only on the frequency locations regardless of the amplitudes of their
respective coefficients. The incoherence measure is characterized
by the reciprocal of the smallest singular value of some Gram matrix,
which is defined by sampling the {\em Dirichlet kernel} at the
wrap-around differences of all frequency pairs. The signal of interest
is said to obey the incoherence condition if the Gram matrix is well
conditioned, which arises over a broad class of spectrally sparse
signals including but not restricted to signals with well-separated
frequencies. We demonstrate that, under this incoherence condition,
EMaC enables exact recovery from $\mathcal{O}(r\log^{4}n)$ random
samples%
\footnote{The standard notation $f(n)=\mathcal{O}\left(g(n)\right)$ means that
there exists a constant $c>0$ such that $f(n)\leq cg(n)$; $f(n)=\Theta\left(g(n)\right)$
indicates that there are numerical constants $c_{1},c_{2}>0$ such
that $c_{1}g(n)\leq f(n)\leq c_{2}g(n)$.%
}, and is stable against bounded noise. Moreover, EMaC admits perfect
signal recovery from $\mathcal{O}(r^{2}\log^{3}n)$ random samples
even when a constant proportion of the samples are corrupted with
arbitrary magnitudes. Finally, numerical experiments validate our
theoretical findings, and demonstrate the applicability of EMaC in
super resolution.

Along the way, we provide theoretical guarantees for low-rank matrix
completion of Hankel matrices and Toeplitz matrices, which is of great
importance in control, natural language processing, and computer vision.
To the best of our knowledge, our results provide the first theoretical
guarantees for Hankel matrix completion that are close to the information
theoretic limit.

\subsection{Connection and Comparison to Prior Work}

%Spectral compressed sensing is closely related to \emph{harmonic retrieval}, which seeks to extract the underlying frequencies of a signal from a collection of its time domain samples. %This spans many signal processing applications including radar
%systems \cite{NionSid2010}, array imaging systems \cite{Borcea2002imaging},
%channel sensing in wireless communication \cite{SayeedAazhang1999},
%etc. %In fact, if the time domain representation of an signal can be estimated accurately, then its underlying frequencies can be identified using harmonic super-resolution methods like ESPRIT \cite{RoyKailathESPIRIT1989}.
%Conventional methods for harmonic retrieval include Prony's method
%\cite{prony1795essai}, ESPRIT \cite{RoyKailathESPIRIT1989}, the
%matrix pencil method \cite{HuaSarkar1990}, the Tufts and Kumaresan
%approach \cite{TuftsKumaresan1982}, etc. These methods are typically
%based on the eigenvalue decomposition of covariance matrices constructed
%from \emph{equi-spaced} samples, which can accommodate infinite frequency
%precision in the absence of noise.

The $K$-fold Hankel structure, which plays a central role in the
EMaC algorithm, roots from the traditional spectral estimation technique
named Matrix Enhancement Matrix Pencil (MEMP) \cite{Hua1992} for
multi-dimensional harmonic retrieval. The conventional MEMP algorithm
assumes fully observed equi-spaced time domain samples for estimation,
and require prior knowledge on the model order. Cadzow's denoising
method \cite{cadzow1982spectral} also exploits the low-rank structure
of the matrix pencil form for denoising line spectrum, but the method
is non-convex and lacks performance guarantees.

When the frequencies of the signal indeed fall on a grid, CS algorithms
based on $\ell_{1}$ minimization \cite{CandRomTao06,Don2006} assert
that it is possible to recover the spectrally sparse signal from $\mathcal{O}(r\log n)$
random time domain samples. These algorithms admit faithful recovery
even when the samples are contaminated by bounded noise \cite{CandesRombergTao2006Stable,CandesPlan2011RIPless}
or arbitrary sparse outliers \cite{li2011compressed}. When the inevitable
\textit{basis mismatch} issue \cite{Chi2011sensitivity} is present,
several remedies of CS algorithms have been proposed to mitigate the
effect \cite{duarte2012spectral,fannjiang2012coherence} under random
linear projection measurements, although theoretical guarantees are
in general lacking.

More recently, Cand�s and Fernandez-Granda \cite{CandesFernandez2012SR}
proposed a total-variation norm minimization algorithm to super-resolve
a sparse signal from frequency samples at the \emph{low end} of the
spectrum. This algorithm allows accurate super-resolution when the
point sources are sufficiently separated, and is stable against noise
\cite{CandesFernandez2012SRNoisy}. Inspired by this approach, Tang
et. al. \cite{TangBhaskarShahRecht2012} then developed an atomic
norm minimization algorithm for line spectral estimation from $\mathcal{O}(r\log r\log n)$
random time domain samples, which enables exact recovery when the
frequencies are separated by at least $4/n$ with random amplitude
phases. Similar performance guarantees are later established in \cite{chi2013asilomar}
for multi-dimensional frequencies. However, these results are established
under a random signal model, i.e. the complex signs of the frequency
spikes are assumed to be i.i.d. drawn from a uniform distribution.
The robustness of the method against noise and outliers is not established
either. In contrast, our approach yields deterministic conditions
for multi-dimensional frequency models that guarantee perfect recovery
with noiseless samples and are provably robust against noise and sparse
corruptions. We will provide detailed comparison with the approach
of Tang et. al. after we formally present our results. Numerical comparison
will also be provided in Section~\ref{sec:num:comparison} for the
line spectrum model.

Our algorithm is inspired by recent advances of Matrix Completion
(MC) \cite{ExactMC09,keshavan2010matrix}, which aims at recovering
a low-rank matrix from partial entries. It has been shown \cite{CanTao10,Gross2011recovering,chen2013incoherence}
that exact recovery is possible via nuclear norm minimization, as
soon as the number of observed entries exceeds the order of the information
theoretic limit. This line of algorithms is also robust against noise
and outliers \cite{CanLiMaWri09,Negahban2012restricted}, and allows
exact recovery even in the presence of a constant portion of adversarially
corrupted entries \cite{chandrasekaran2011rank,chen2011robust,chen2011low},
which have found numerous applications in collaborative filtering
\cite{wu2007collaborative}, medical imaging \cite{zhang2014accelerating,zhang2013fast},
etc. Nevertheless, the theoretical guarantees of these algorithms
do not apply to the more structured observation models associated
with the proposed multi-fold Hankel structure. Consequently, direct
application of existing MC results delivers pessimistic sample complexity,
which far exceeds the degrees of freedom underlying the signal.

Preliminary results of this work have been presented in \cite{chen2013spectral},
where an additional strong incoherence condition was introduced that
bore a similar role as the traditional strong incoherence parameter
in MC \cite{CanTao10} but lacked physical interpretations. This paper
removes this condition and further improves the sample complexity.

\subsection{Organization}

The rest of the paper is organized as follows. The signal and sampling
models are described in Section \ref{sec:Model-and-Algorithm}. By
restricting our attention to two-dimensional (2-D) frequency models,
we present the enhanced matrix form and the associated structured
matrix completion algorithms. The extension to multi-dimensional frequency
models is discussed in Section \ref{sub:Extension-to-Higher-Dimension}.
The main theoretical guarantees are summarized in Section \ref{sec:Main-Results},
based on the incoherence condition introduced in Section \ref{sub:IncoherenceMeasures}.
We then discuss the extension to low-rank Hankel and Toeplitz matrix
completion in Section \ref{sub:Extension-to-Hankel}. Section \ref{sec:Numerical-Experiments}
presents the numerical validation of our algorithms. The proofs of
Theorems \ref{theorem-EMaC-noiseless} and \ref{theorem-EMaC-Robust}
are based on duality analysis followed by a golfing scheme, which
are supplied in Section \ref{sec:Main-Proof-Exact} and Section \ref{sec:Main-Proof-Robust},
respectively. Section \ref{sec:Conclusions-and-Future} concludes
the paper with a short summary of our findings as well as a discussion
of potential extensions and improvements. Finally, the proofs of auxiliary
lemmas supporting our results are deferred to the appendices.

\section{Model and Algorithm\label{sec:Model-and-Algorithm}}

Assume that the signal of interest $x\left(\boldsymbol{t}\right)$
can be modeled as a weighted sum of $K$-dimensional complex sinusoids
at $r$ distinct frequencies $\boldsymbol{f}_{i}\in[0,1)^{K}$, $1\leq i\leq r$,
i.e. 
\begin{equation}
x(\boldsymbol{t})=\sum_{i=1}^{r}d_{i}e^{j2\pi\left\langle \boldsymbol{t},\boldsymbol{f}_{i}\right\rangle },\;\boldsymbol{t}\in\mathbb{Z}^{K}.\label{eq:xtform}
\end{equation}
It is assumed throughout that the frequencies $\boldsymbol{f}_{i}$'s
are normalized with respect to the Nyquist frequency of $x(\boldsymbol{t})$
and the time domain measurements are sampled at integer values. We
denote by $d_{i}$'s the complex amplitudes of the associated coefficients,
and $\left\langle \cdot,\cdot\right\rangle $ represents the inner
product. For concreteness, our discussion is mainly devoted to a 2-D
frequency model when $K=2$. This subsumes line spectral estimation
as a special case, and indicates how to address multi-dimensional
models. The algorithms for higher dimensional scenarios closely parallel
the 2-D case, which will be briefly discussed in Section \ref{sub:Extension-to-Higher-Dimension}.

\subsection{2-D Frequency Model\label{sub:2-D-Frequency-Model}}

Consider a data matrix $\boldsymbol{X}=[X_{k,l}]_{0\leq k<n_{1},0\leq l<n_{2}}$
of ambient dimension $n:=n_{1}n_{2}$, which is obtained by sampling
the signal \eqref{eq:xtform} on a uniform grid. From \eqref{eq:xtform}
each entry $X_{k,l}$ can be expressed as 
\begin{equation}
X_{k,l}=x(k,l)=\sum_{i=1}^{r}d_{i}y_{i}^{k}z_{i}^{l},\label{eq:Xform}
\end{equation}
where for any $i$ ($1\leq i\leq r$) we define 
\[
y_{i}:=\exp\left(j2\pi f_{1i}\right)\quad\text{and}\quad z_{i}:=\exp\left(j2\pi f_{2i}\right)
\]
for some frequency pairs $\left\{ \boldsymbol{f}_{i}=\left(f_{1i},f_{2i}\right)\mid1\leq i\leq r\right\} $.
We can then express $\boldsymbol{X}$ in a matrix form as follows
\begin{equation}
\boldsymbol{X}=\boldsymbol{Y}\boldsymbol{D}\boldsymbol{Z}^{\top},\label{eq:X_MatrixForm}
\end{equation}
where the above matrices are defined as 
\begin{equation}
\boldsymbol{Y}:=\left[\begin{array}{cccc}
1 & 1 & \cdots & 1\\
y_{1} & y_{2} & \cdots & y_{r}\\
\vdots & \vdots & \vdots & \vdots\\
y_{1}^{n_{1}-1} & y_{2}^{n_{1}-1} & \cdots & y_{r}^{n_{1}-1}
\end{array}\right],\label{eq:Y_matrix_form}
\end{equation}
\begin{equation}
\boldsymbol{Z}:=\left[\begin{array}{cccc}
1 & 1 & \cdots & 1\\
z_{1} & z_{2} & \cdots & z_{r}\\
\vdots & \vdots & \vdots & \vdots\\
z_{1}^{n_{2}-1} & z_{2}^{n_{2}-1} & \cdots & z_{r}^{n_{2}-1}
\end{array}\right],\label{eq:Z_matrix_form}
\end{equation}
and 
\begin{equation}
\boldsymbol{D}:=\text{diag}\left[d_{1},d_{2},\cdots,d_{r}\right].\label{eq:D_matrix_form}
\end{equation}
The above form (\ref{eq:X_MatrixForm}) is sometimes referred to as
the Vandemonde decomposition of $\boldsymbol{X}$.

Suppose that there exists a location set $\Omega$ of size $m$ such
that the $X_{k,l}$ is observed if and only if $\left(k,l\right)\in\Omega$.
It is assumed that $\Omega$ is sampled uniformly at random. Define
$\mathcal{P}_{\Omega}(\boldsymbol{X})$ as the orthogonal projection
of $\boldsymbol{X}$ onto the subspace of matrices that vanish outside
$\Omega$. We aim at recovering $\boldsymbol{X}$ from $\mathcal{P}_{\Omega}(\boldsymbol{X})$.

\subsection{Matrix Enhancement\label{sub:Matrix-Enhancement}}

One might naturally attempt recovery by applying the low-rank MC algorithms
\cite{ExactMC09}, arguing that when $r$ is small, perfect recovery
of $\boldsymbol{X}$ is possible from partial measurements since $\boldsymbol{X}$
is low rank if $r\ll\min\{n_{1},n_{2}\}$. Specifically, this corresponds
to the following algorithm: 
\begin{align}
\underset{\boldsymbol{M}\in\mathbb{C}^{n_{1}\times n_{2}}}{\text{minimize}}\quad & \|\boldsymbol{M}\|_{*}\\
\text{subject to}\quad & \mathcal{P}_{\Omega}\left(\boldsymbol{M}\right)=\mathcal{P}_{\Omega}\left(\boldsymbol{X}\right),\nonumber 
\end{align}
where $\left\Vert \boldsymbol{M}\right\Vert _{*}$ denotes the nuclear
norm (or sum of all singular values) of a matrix $\boldsymbol{M}=[M_{k,l}]$.
This is a convex relaxation paradigm with respect to rank minimization.
However, naive MC algorithms \cite{Gross2011recovering} require at
least the order of $r\max\left(n_{1},n_{2}\right)\log\left(n_{1}n_{2}\right)$
samples in order to allow perfect recovery, which far exceeds the
degrees of freedom (which is $\Theta\left(r\right)$) in our problem.
What is worse, since the number $r$ of spectral spikes can be as
large as $n_{1}n_{2}$, $\boldsymbol{X}$ might become full-rank once
$r>\min\left(n_{1},n_{2}\right)$. This motivates us to seek other
forms that better capture the harmonic structure.

In this paper, we adopt one effective enhanced form of $\boldsymbol{X}$
based on the following two-fold Hankel structure. The enhanced matrix
$\boldsymbol{X}_{\text{e}}$ with respect to $\boldsymbol{X}$ is
defined as a $k_{1}\times\left(n_{1}-k_{1}+1\right)$ block Hankel
matrix 
\begin{equation}
\boldsymbol{X}_{\text{e}}:=\left[\begin{array}{cccc}
\boldsymbol{X}_{0} & \boldsymbol{X}_{1} & \cdots & \boldsymbol{X}_{n_{1}-k_{1}}\\
\boldsymbol{X}_{1} & \boldsymbol{X}_{2} & \cdots & \boldsymbol{X}_{n_{1}-k_{1}+1}\\
\vdots & \vdots & \vdots & \vdots\\
\boldsymbol{X}_{k_{1}-1} & \boldsymbol{X}_{k_{1}} & \cdots & \boldsymbol{X}_{n_{1}-1}
\end{array}\right],\label{eq:XeEnhancedForm}
\end{equation}
where $k_{1}$ $(1\leq k_{1}\leq n_{1})$ is called a pencil parameter.
Each block is a $k_{2}\times\left(n_{2}-k_{2}+1\right)$ Hankel matrix
defined such that for every $\ell$ ($0\leq\ell<n_{1}$): 
\begin{equation}
\boldsymbol{X}_{\ell}:=\left[\begin{array}{cccc}
X_{\ell,0} & X_{\ell,1} & \cdots & X_{\ell,n_{2}-k_{2}}\\
X_{\ell,1} & X_{\ell,2} & \cdots & X_{\ell,n_{2}-k_{2}+1}\\
\vdots & \vdots & \vdots & \vdots\\
X_{\ell,k_{2}-1} & X_{\ell,k_{2}} & \cdots & X_{\ell,n_{2}-1}
\end{array}\right],
\end{equation}
where $1\leq k_{2}\leq n_{2}$ is another pencil parameter. This enhanced
form allows us to express each block as%
\footnote{Note that the $l$th ($0\leq l<n_{1}$) row $\boldsymbol{X}_{l*}$
of $\boldsymbol{X}$ can be expressed as 
\[
\boldsymbol{X}_{l*}=\left[y_{1}^{l},\cdots,y_{r}^{l}\right]\boldsymbol{D}\boldsymbol{Z}^{\top}=\left[y_{1}^{l}d_{1},\cdots,y_{r}^{l}d_{r}\right]\boldsymbol{Z}^{\top},
\]
and hence we only need to find the Vandemonde decomposition for $\boldsymbol{X}_{0}$
and then replace $d_{i}$ by $y_{i}^{l}d_{i}$.%
} 
\begin{equation}
\boldsymbol{X}_{\ell}=\boldsymbol{Z}_{\text{L}}\boldsymbol{Y}_{\text{d}}^{\ell}\boldsymbol{D}\boldsymbol{Z}{}_{\text{R}},\label{eq:XEachBlock_YZ}
\end{equation}
where $\boldsymbol{Z}_{\text{L}}$, $\boldsymbol{Z}_{\text{R}}$ and
$\boldsymbol{Y}_{\text{d}}$ are defined respectively as 
\[
\boldsymbol{Z}_{\text{L}}:=\left[\begin{array}{cccc}
1 & 1 & \cdots & 1\\
z_{1} & z_{2} & \cdots & z_{r}\\
\vdots & \vdots & \vdots & \vdots\\
z_{1}^{k_{2}-1} & z_{2}^{k_{2}-1} & \cdots & z_{r}^{k_{2}-1}
\end{array}\right],
\]
\[
\boldsymbol{Z}_{\text{R}}:=\left[\begin{array}{cccc}
1 & z_{1} & \cdots & z_{1}^{n_{2}-k_{2}}\\
1 & z_{2} & \cdots & z_{2}^{n_{2}-k_{2}}\\
\vdots & \vdots & \vdots & \vdots\\
1 & z_{r} & \cdots & z_{r}^{n_{2}-k_{2}}
\end{array}\right],
\]
and 
\[
\boldsymbol{Y}_{\text{d}}:=\text{diag}\left[y_{1},y_{2},\cdots,y_{r}\right].
\]
Substituting (\ref{eq:XEachBlock_YZ}) into (\ref{eq:XeEnhancedForm})
yields the following: 
\begin{equation}
\boldsymbol{X}_{\text{e}}=\underset{\sqrt{k_{1}k_{2}}\boldsymbol{E}_{\text{L}}}{\underbrace{\begin{bmatrix}\boldsymbol{Z}_{\text{L}}\\
\boldsymbol{Z}_{\text{L}}\boldsymbol{Y}_{\text{d}}\\
\vdots\\
\boldsymbol{Z}_{\text{L}}\boldsymbol{Y}_{\text{d}}^{k_{1}-1}
\end{bmatrix}}}\boldsymbol{D}\underset{\sqrt{\left(n_{1}-k_{1}+1\right)\left(n_{2}-k_{2}+1\right)}\boldsymbol{E}_{\text{R}}}{\underbrace{\begin{bmatrix}\boldsymbol{Z}_{\text{R}}, & \boldsymbol{Y}_{\text{d}}\boldsymbol{Z}_{\text{R}}, & \cdots, & \boldsymbol{Y}_{\text{d}}^{n_{1}-k_{1}}\boldsymbol{Z}_{\text{R}}\end{bmatrix}}},\label{eq:Xe_as_ELER}
\end{equation}
where $\boldsymbol{E}_{\text{L}}$ and $\boldsymbol{E}_{\text{R}}$
span the column and row space of $\boldsymbol{X}_{\text{e}}$, respectively.
This immediately implies that $\boldsymbol{X}_{\text{e}}$ is \textit{low-rank},
i.e. 
\begin{equation}
\text{rank}\left(\boldsymbol{X}_{\text{e}}\right)\leq r.
\end{equation}
This form is inspired by the traditional matrix pencil approach proposed
in \cite{HuaSarkar1990,Hua1992} to estimate harmonic frequencies
if \emph{all} entries of $\boldsymbol{X}$ are available. Thus, one
can extract all underlying frequencies of $\boldsymbol{X}$ using
methods proposed in \cite{Hua1992}, as long as $\boldsymbol{X}$
can be faithfully recovered.

\subsection{The EMaC Algorithm in the Absence of Noise \label{sub:Algorithm}}

We then attempt recovery through the following \emph{Enhancement Matrix
Completion (EMaC)} algorithm: 
\begin{align}
\text{(EMaC)}\quad\underset{\boldsymbol{M}\in\mathbb{C}^{n_{1}\times n_{2}}}{\text{minimize}}\quad & \left\Vert \boldsymbol{M}_{\text{e}}\right\Vert _{*}\label{eq:EMaC}\\
\text{subject to}\quad & \mathcal{P}_{\Omega}\left(\boldsymbol{M}\right)=\mathcal{P}_{\Omega}\left(\boldsymbol{X}\right),\nonumber 
\end{align}
where $\boldsymbol{M}_{\text{e}}$ denotes the enhanced form of $\boldsymbol{M}$.
In other words, EMaC minimizes the nuclear norm of the enhanced form
over all matrices compatible with the samples. This convex program
can be rewritten into a semidefinite program (SDP) \cite{RecFazPar07}
\begin{align*}
\underset{\boldsymbol{M}\in\mathbb{C}^{n_{1}\times n_{2}}}{\text{minimize}}\quad & \frac{1}{2}\text{Tr}\left(\boldsymbol{Q}_{1}\right)+\frac{1}{2}\text{Tr}\left(\boldsymbol{Q}_{2}\right)\\
\text{subject to}\quad & \mathcal{P}_{\Omega}\left(\boldsymbol{M}\right)=\mathcal{P}_{\Omega}\left(\boldsymbol{X}\right),\\
 & \left[\begin{array}{cc}
\boldsymbol{Q}_{1} & \boldsymbol{M}_{\text{e}}^{*}\\
\boldsymbol{M}_{\text{e}} & \boldsymbol{Q}_{2}
\end{array}\right]\succeq0,
\end{align*}
which can be solved using off-the-shelf solvers in a tractable manner
(see, e.g., \cite{RecFazPar07}). It is worth mentioning that EMaC
has a similar computational complexity as the atomic norm minimization
method \cite{TangBhaskarShahRecht2012} when restricted to the 1-D
frequency model.

Careful readers will remark that the performance of EMaC must depend
on the choices of the pencil parameters $k_{1}$ and $k_{2}$. In
fact, if we define a quantity 
\begin{equation}
c_{\text{s}}:=\max\left\{ \frac{n_{1}n_{2}}{k_{1}k_{2}},\frac{n_{1}n_{2}}{\left(n_{1}-k_{1}+1\right)\left(n_{2}-k_{2}+1\right)}\right\} \label{eq:DefnCsq}
\end{equation}
that measures how close $\boldsymbol{X}_{\mathrm{e}}$ is to a square
matrix, then it will be shown later that the required sample complexity
for faithful recovery is an increasing function of $c_{\text{s}}$.
In fact, both our theory and empirical experiments are in favor of
a small $c_{\text{s}}$, corresponding to the choices $k_{1}=\Theta\left(n_{1}\right)$,
$n_{1}-k_{1}+1=\Theta\left(n_{1}\right)$, $k_{2}=\Theta\left(n_{2}\right)$,
and $n_{2}-k_{2}+1=\Theta\left(n_{2}\right)$.

\begin{comment}
On the other hand, the computational complexity of EMaC is comparable
to other algorithms based on SDP. For instance, when restricted to
1-D data model (i.e. $n_{2}=1$), the atomic norm minimization method
\cite{TangBhaskarShahRecht2012} aims at solving an SDP with $n_{1}\times n_{1}$
matrices. In comparison, EMaC can be equivalently converted to \cite{RecFazPar07}
\begin{align*}
\underset{\boldsymbol{M}\in\mathbb{C}^{n_{1}\times n_{2}}}{\text{minimize}}\quad & \frac{1}{2}\text{Tr}\left(\boldsymbol{Q}_{1}\right)+\frac{1}{2}\text{Tr}\left(\boldsymbol{Q}_{2}\right)\\
\text{subject to}\quad & \mathcal{P}_{\Omega}\left(\boldsymbol{M}\right)=\mathcal{P}_{\Omega}\left(\boldsymbol{X}\right),\\
 & \left[\begin{array}{cc}
\boldsymbol{Q}_{1} & \boldsymbol{M}_{\text{e}}^{*}\\
\boldsymbol{M}_{\text{e}} & \boldsymbol{Q}_{2}
\end{array}\right]\succeq0,
\end{align*}
which forms an SDP with $n_{1}\times n_{1}$ matrices with the same
number of linear constraints, regardless of the choices of $k_{1}$
and $k_{2}$. Consequently, EMaC shares the same order of computational
complexity with atomic norm minimization if they are to be solved
by generic SDP solvers. 
\end{comment}

\subsection{The Noisy-EMaC Algorithm with Bounded Noise}

In practice, measurements are often contaminated by a certain amount
of noise. To make our model and algorithm more practically applicable,
we replace our measurements by $\boldsymbol{X}^{\text{o}}=[{X}_{k,l}^{\text{o}}]_{0\leq k<n_{1},0\leq l<n_{2}}$
through the following noisy model 
\begin{equation}
{X}_{k,l}^{\text{o}}={X}_{k,l}+{N}_{k,l},\quad\forall(k,l)\in\Omega,\label{eq:NoisyDataModel}
\end{equation}
where ${X}_{k,l}^{\text{o}}$ is the observed $(k,l)$-th entry, and
$\boldsymbol{N}=[N_{k,l}]_{0\leq k<n_{1},0\leq l<n_{2}}$ denotes
some unknown noise. We assume that the noise magnitude is bounded
by a known amount $\left\Vert \mathcal{P}_{\Omega}\left(\boldsymbol{N}\right)\right\Vert _{\text{F}}\leq\delta$,
where $\left\Vert \cdot\right\Vert _{\text{F}}$ denotes the Frobenius
norm. In order to adapt our algorithm to such noisy measurements,
one wishes that small perturbation in the measurements should result
in small variation in the estimate. Our algorithm is then modified
as follows 
\begin{align}
\text{(Noisy-EMaC)}:\quad\underset{\boldsymbol{M}\in\mathbb{C}^{n_{1}\times n_{2}}}{\text{minimize}}\quad & \left\Vert \boldsymbol{M}_{\text{e}}\right\Vert _{*}\label{eq:EMaCNoisy}\\
\text{subject to}\quad & \left\Vert \mathcal{\mathcal{P}}_{\Omega}\left(\boldsymbol{M}-\boldsymbol{X}^{\text{o}}\right)\right\Vert _{\text{F}}\leq\delta.\nonumber 
\end{align}
That said, the algorithm searches for a candidate with minimum nuclear
norm among all signals close to the measurements.

\subsection{The Robust-EMaC Algorithm with Sparse Outliers}

An outlier is a data sample that can deviate arbitrarily from the
true data point. Practical data samples one collects may contain a
certain portion of outliers due to abnormal behavior of data acquisition
devices such as amplifier saturation, sensor failures, and malicious
attacks. A desired recovery algorithm should be able to automatically
prune all outliers even when they corrupt up to a constant portion
of all data samples.

Specifically, suppose that our measurements $\boldsymbol{X}^{\text{o}}$
are given by 
\begin{equation}
{X}_{k,l}^{\text{o}}={X}_{k,l}+{S}_{k,l},\quad\forall(k,l)\in\Omega,\label{eq:SparseOutlierModel}
\end{equation}
where ${X}_{k,l}^{\text{o}}$ is the observed $(k,l)$-th entry, and
$\boldsymbol{S}=[S_{k,l}]_{0\leq k<n_{1},0\leq l<n_{2}}$ denotes
the outliers, which is assumed to be a sparse matrix supported on
some location set $\Omega^{\text{dirty}}\subseteq\Omega$. The sampling
model is formally described as follows. 
\begin{enumerate}
\item Suppose that $\Omega$ is obtained by sampling $m$ entries uniformly
at random, and define $\rho:=\frac{m}{n_{1}n_{2}}$. 
\item Conditioning on $(k,l)\in\Omega$, the events $\left\{ (k,l)\in\Omega^{\text{dirty}}\right\} $
are independent with conditional probability 
\[
\mathbb{P}\left\{ \left(k,l\right)\in\Omega^{\text{dirty}}\mid\left(k,l\right)\in\Omega\right\} =\tau
\]
for some small constant corruption fraction $0<\tau<1$. 
\item Define $\Omega^{\text{clean}}:=\Omega\backslash\Omega^{\text{dirty}}$
as the location set of \emph{uncorrupted} measurements. 
\end{enumerate}
EMaC is then modified as follows to accommodate sparse outliers: 
\begin{align}
\text{(Robust-EMaC)}\underset{\boldsymbol{M},\hat{\boldsymbol{S}}\in\mathbb{C}^{n_{1}\times n_{2}}}{\text{minimize}}\quad & \left\Vert \boldsymbol{M}_{\mathrm{e}}\right\Vert _{*}+\lambda\|\hat{\boldsymbol{S}}_{\mathrm{e}}\|_{1}\label{eq:RobustEMaC}\\
\text{subject to}\quad & \mathcal{P}_{\Omega}\left(\boldsymbol{M}+\hat{\boldsymbol{S}}\right)=\mathcal{P}_{\Omega}\left(\boldsymbol{X}+\boldsymbol{S}\right),\nonumber 
\end{align}
where $\lambda>0$ is a regularization parameter that will be specified
later. As will be shown later, $\lambda$ can be selected in a parameter-free
fashion. We denote by $\boldsymbol{M}_{\text{e}}$ and $\hat{\boldsymbol{S}}_{\text{e}}$
the enhanced form of $\boldsymbol{M}$ and $\hat{\boldsymbol{S}}$,
respectively. Here, $\Vert\hat{\boldsymbol{S}}_{\text{e}}\Vert_{1}:=\|\mbox{vec}(\hat{\boldsymbol{S}}_{\text{e}})\|_{1}$
represents the elementwise $\ell_{1}$-norm of $\hat{\boldsymbol{S}}_{\text{e}}$.
Robust-EMaC promotes the low-rank structure of the enhanced data matrix
as well as the sparsity of the outliers via convex relaxation with
respective structures.

\subsection{Notations}

Before continuing, we introduce a few notations that will be used
throughout. Let the singular value decomposition (SVD) of $\boldsymbol{X}_{\text{e}}$
be $\boldsymbol{X}_{\text{e}}=\boldsymbol{U}\boldsymbol{\Lambda}\boldsymbol{V}^{*}$.
Denote by 
\begin{align}
T: & =\left\{ \boldsymbol{U}\boldsymbol{M}^{*}+\tilde{\boldsymbol{M}}\boldsymbol{V}^{*}:\boldsymbol{M}\in\mathbb{C}^{\left(n_{1}-k_{1}+1\right)\left(n_{2}-k_{1}+1\right)\times r},\right.\nonumber \\
 & \quad\quad\quad\quad\quad\left.\tilde{\boldsymbol{M}}\in\mathbb{C}^{k_{1}k_{2}\times r}\right\} \label{eq:TangentSpace}
\end{align}
the tangent space with respect to $\boldsymbol{X}_{\text{e}}$, and
$T^{\perp}$ the orthogonal complement of $T$. Denote by $\mathcal{P}_{U}$
(resp. $\mathcal{P}_{V}$, $\mathcal{P}_{T}$) the orthogonal projection
onto the column (resp. row, tangent) space of $\boldsymbol{X}_{\text{e}}$,
i.e. for any $\boldsymbol{M}$, 
\[
\mathcal{P}_{U}\left(\boldsymbol{M}\right)=\boldsymbol{U}\boldsymbol{U}^{*}\boldsymbol{M},\quad\mathcal{P}_{V}\left(\boldsymbol{M}\right)=\boldsymbol{M}\boldsymbol{V}\boldsymbol{V}^{*},
\]
\[
\text{and}\quad\mathcal{P}_{T}=\mathcal{P}_{U}+\mathcal{P}_{V}-\mathcal{P}_{U}\mathcal{P}_{V}.
\]
We let $\mathcal{P}_{T^{\perp}}=\mathcal{I}-\mathcal{P}_{T}$ be the
orthogonal complement of $\mathcal{P}_{T}$, where $\mathcal{I}$
denotes the identity operator.

Denote by $\left\Vert \boldsymbol{M}\right\Vert $, $\left\Vert \boldsymbol{M}\right\Vert _{\mathrm{F}}$
and $\left\Vert \boldsymbol{M}\right\Vert _{*}$ the spectral norm
(operator norm), Frobenius norm, and nuclear norm of $\boldsymbol{M}$,
respectively. Also, $\left\Vert \boldsymbol{M}\right\Vert _{1}$ and
$\left\Vert \boldsymbol{M}\right\Vert _{\infty}$ are defined to be
the \emph{elementwise} $\ell_{1}$ and $\ell_{\infty}$ norm of $\boldsymbol{M}$.
Denote by $\boldsymbol{e}_{i}$ the $i^{\text{th}}$ standard basis
vector. Additionally, we use $\text{sgn}\left(\boldsymbol{M}\right)$
to denote the elementwise complex sign of $\boldsymbol{M}$.

On the other hand, we denote by $\Omega_{\text{e}}(k,l)$ the set
of locations of the enhanced matrix $\boldsymbol{X}_{\text{e}}$ containing
copies of $X_{k,l}$. Due to the Hankel or multi-fold Hankel structures,
one can easily verify the following: each location set $\Omega_{\text{e}}(k,l)$
contains at most one index in any given row of the enhanced form,
and at most one index in any given column. For each $\left(k,l\right)\in[n_{1}]\times[n_{2}]$,
we use $\boldsymbol{A}_{(k,l)}$ to denote a basis matrix that extracts
the average of all entries in $\Omega_{\text{e}}\left(k,l\right)$.
Specifically, 
\begin{equation}
\left(\boldsymbol{A}_{(k,l)}\right)_{\alpha,\beta}:=\begin{cases}
\frac{1}{\sqrt{\left|\Omega_{\text{e}}\left(k,l\right)\right|}},\quad & \text{if }\left(\alpha,\beta\right)\in\Omega_{\text{e}}\left(k,l\right),\\
0, & \text{else}.
\end{cases}\label{eq:DefnMeanBasis}
\end{equation}
We will use 
\begin{equation}
\omega_{k,l}:=\left|\Omega_{\text{e}}\left(k,l\right)\right|\label{eq:littleOmega}
\end{equation}
throughout as a short-hand notation.

\section{Main Results\label{sec:Main-Results}}

This section delivers the following encouraging news: under mild incoherence
conditions, EMaC enables faithful signal recovery from a minimal number
of time-domain samples, even when the samples are contaminated by
bounded noise or a constant portion of arbitrary outliers.

\subsection{Incoherence Measure\label{sub:IncoherenceMeasures}}

In general, matrix completion from a few entries is hopeless unless
the underlying structure is sufficiently uncorrelated with the observation
basis. This inspires us to introduce certain incoherence measures.
To this end, we define the 2-D Dirichlet kernel as 
\begin{equation}
\mathcal{D}(k_{1},k_{2},\boldsymbol{f}):=\frac{1}{k_{1}k_{2}}\left(\frac{1-e^{-j2\pi k_{1}f_{1}}}{1-e^{-j2\pi f_{1}}}\right)\left(\frac{1-e^{-j2\pi k_{2}f_{2}}}{1-e^{-j2\pi f_{2}}}\right),\label{dirichlet}
\end{equation}
where $\boldsymbol{f}=(f_{1},f_{2})\in[0,1)^{2}$. Fig.~\ref{fig:dirichlet}
(a) illustrates the amplitude of $\mathcal{D}(k_{1},k_{2},\boldsymbol{f})$
when $k_{1}=k_{2}=6$. The value of $|\mathcal{D}(k_{1},k_{2},\boldsymbol{f})|$
decays inverse proportionally with respect to the frequency $\boldsymbol{f}$.
Set $\boldsymbol{G}_{\text{L}}$ and $\boldsymbol{G}_{\text{R}}$
to be two $r\times r$ Gram matrices such that their entries are specified
respectively by 
\begin{align*}
(\boldsymbol{G}_{\text{L}})_{i,l} & =\mathcal{D}(k_{1},k_{2},\boldsymbol{f}_{i}-\boldsymbol{f}_{l}),\\
(\boldsymbol{G}_{\text{R}})_{i,l} & =\mathcal{D}(n_{1}-k_{1}+1,n_{2}-k_{2}+1,\boldsymbol{f}_{i}-\boldsymbol{f}_{l}),
\end{align*}
where the difference $\boldsymbol{f}_{i}-\boldsymbol{f}_{l}$ is understood
as the wrap-around distance in the interval $[-1/2,1/2)^{2}$. Simple
manipulation reveals that 
\[
\boldsymbol{G}_{\mathrm{L}}=\boldsymbol{E}_{\mathrm{L}}^{*}\boldsymbol{E}_{\mathrm{L}},\quad\boldsymbol{G}_{\mathrm{R}}=\left(\boldsymbol{E}_{\mathrm{R}}\boldsymbol{E}_{\mathrm{R}}^{*}\right)^{\top},
\]
where $\boldsymbol{E}_{\mathrm{L}}$ and $\boldsymbol{E}_{\mathrm{R}}$
are defined in \eqref{eq:Xe_as_ELER}. %\[
%\left(\boldsymbol{G}_{\text{L}}\right){}_{i_{1},i_{2}}:=\begin{cases}
%\frac{1}{k_{1}k_{2}}\frac{1-\left(y_{i_{1}}^{*}y_{i_{2}}\right)^{k_{1}}}{1-y_{i_{1}}^{*}y_{i_{2}}}\frac{1-\left(z_{i_{1}}^{*}z_{i_{2}}\right)^{k_{2}}}{1-z_{i_{1}}^{*}z_{i_{2}}},\quad & \text{if }i_{1}\neq i_{2},\\
%1, & \text{if }i_{1}=i_{2},
%\end{cases}
%\]
%\[
%\left(\boldsymbol{G}_{\text{R}}\right){}_{i_{1},i_{2}}:=\begin{cases}
%\frac{1}{\left(n_{1}-k_{1}+1\right)\left(n_{2}-k_{2}+1\right)}\frac{1-\left(y_{i_{1}}^{*}y_{i_{2}}\right)^{n_{1}-k_{1}+1}}{1-y_{i_{1}}^{*}y_{i_{2}}}\frac{1-\left(z_{i_{1}}^{*}z_{i_{2}}\right)^{n_{2}-k_{2}+1}}{1-z_{i_{1}}^{*}z_{i_{2}}},\quad & \text{if }i_{1}\neq i_{2},\\
%1, & \text{if }i_{1}=i_{2}.
%\end{cases}
%\]
%Note that $\boldsymbol{G}_{\text{L}}$ and $\boldsymbol{G}_{\text{R}}$ can be obtained by sampling the 2-D Dirichlet kernel, which is frequently considered in Fourier analysis. 

\begin{comment}
In this paper, we will always choose $\left(k_{1},k_{2}\right)$ such
that $k_{1}=\Theta\left(n_{1}\right)$, $k_{2}=\Theta(n_{2})$, $n-k_{1}=\Theta\left(n_{1}\right)$
and $n_{2}-k_{2}=\Theta(n_{2})$. Define $c_{\text{s}}:=\max\left(\frac{n_{1}n_{2}}{k_{1}k_{2}},\frac{n_{1}n_{2}}{\left(n_{1}-k_{1}+1\right)\left(n_{2}-k_{2}+1\right)}\right)$,
which is a constant. Then our incoherence measure is defined as follows. 
\end{comment}

Our incoherence measure is then defined as follows. \begin{definition}[\textbf{Incoherence}]A
matrix $\boldsymbol{X}$ is said to obey the incoherence property
with parameter $\mu_{1}$ if 
\begin{equation}
\sigma_{\min}\left(\boldsymbol{G}_{\mathrm{L}}\right)\geq\frac{1}{\mu_{1}}\quad\text{and}\quad\sigma_{\min}\left(\boldsymbol{G}_{\mathrm{R}}\right)\geq\frac{1}{\mu_{1}}.\label{eq:LeastSV_G}
\end{equation}
where $\sigma_{\min}\left(\boldsymbol{G}_{\mathrm{L}}\right)$ and
$\sigma_{\min}\left(\boldsymbol{G}_{\mathrm{R}}\right)$ represent
the least singular values of $\boldsymbol{G}_{\mathrm{L}}$ and $\boldsymbol{G}_{\mathrm{R}}$,
respectively. \end{definition}

%\begin{remark}Recall that $\Omega_{\text{e}}(i,l)$ denotes the set
%of locations of the enhanced form $\boldsymbol{X}_{\text{e}}$ containing
%copies of $x_{i,l}$, $k_{i,l}=\left|\Omega_{\text{e}}(i,l)\right|$,
%and $\boldsymbol{A}_{i,l}$ is the basis matrix defined in (\ref{eq:DefnMeanBasis})
%that extracts the mean of all entries in $\Omega_{\text{e}}(i,l)$.\end{remark}

The incoherence measure $\mu_{1}$ only depends on the locations of
the frequency spikes, irrespective of the amplitudes of their respective
coefficients. The signal is said to satisfy the incoherence condition
if $\mu_{1}$ scales as a small constant, which occurs when $\boldsymbol{G}_{\text{L}}$
and $\boldsymbol{G}_{\text{R}}$ are both well-conditioned. Our incoherence
condition naturally requires certain separation among all frequency
pairs, as when two frequency spikes are closely located, $\mu_{1}$
gets undesirably large. As shown in \cite[Theorem 2]{liao2014music},
a separation of about $2/n$ for line spectrum is sufficient to guarantee
the incoherence condition to hold. However, it is worth emphasizing
that such strict separation is not necessary as required in \cite{TangBhaskarShahRecht2012},
and thereby our incoherence condition is applicable to a broader class
of spectrally sparse signals.

%The frequencies satisfying (\ref{eq:LeastSV_G}) can be either spread out or minimally separated. 
%There is a large class of spectrally sparse $2$-D signals exhibiting the desired incoherence properties. 

\begin{figure}[htp]
\centering%
\begin{tabular}{cc}
\hspace{-0.1in}\includegraphics[width=0.25\textwidth]{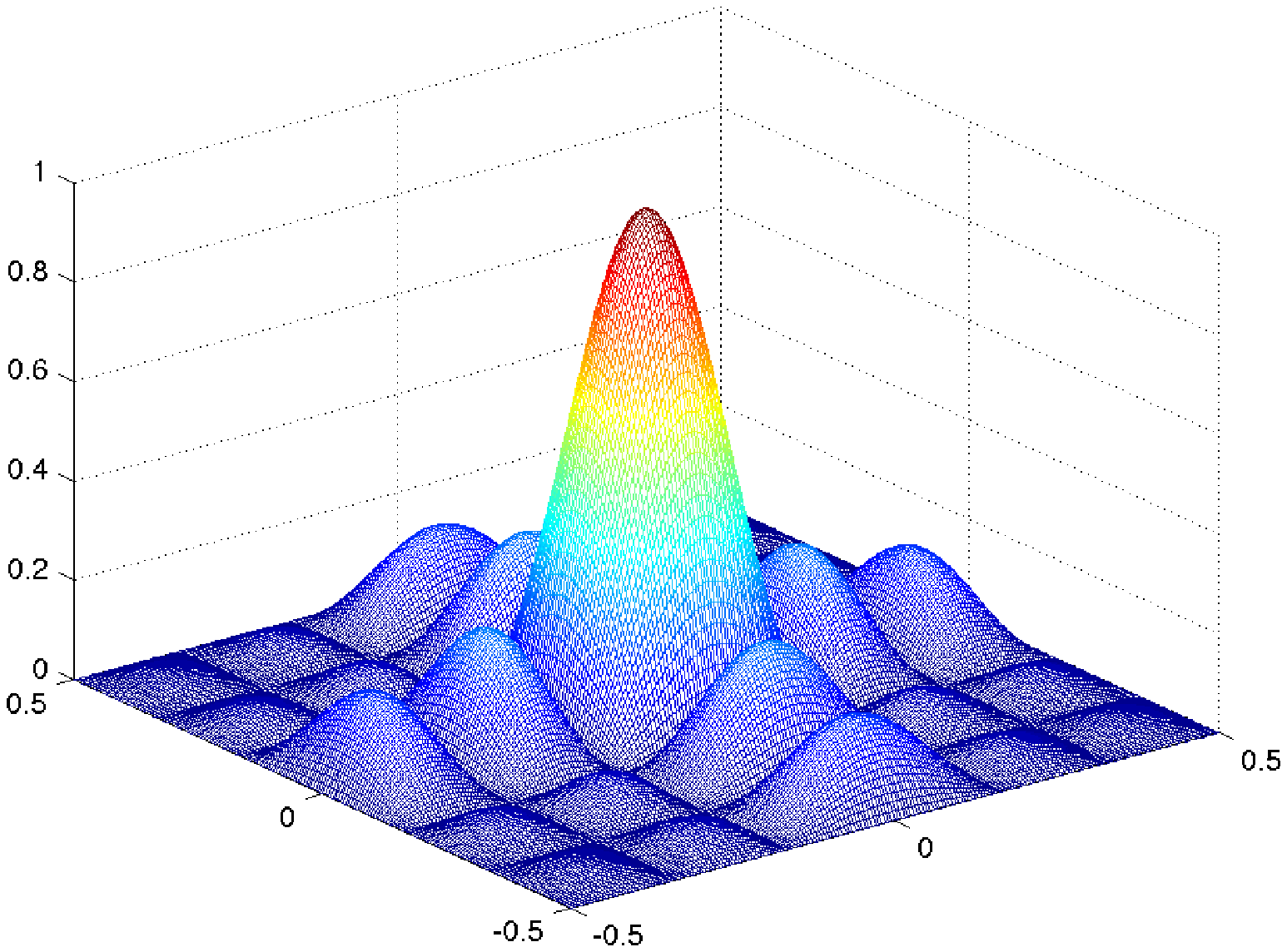}  & \hspace{-0.1in}\includegraphics[width=0.24\textwidth]{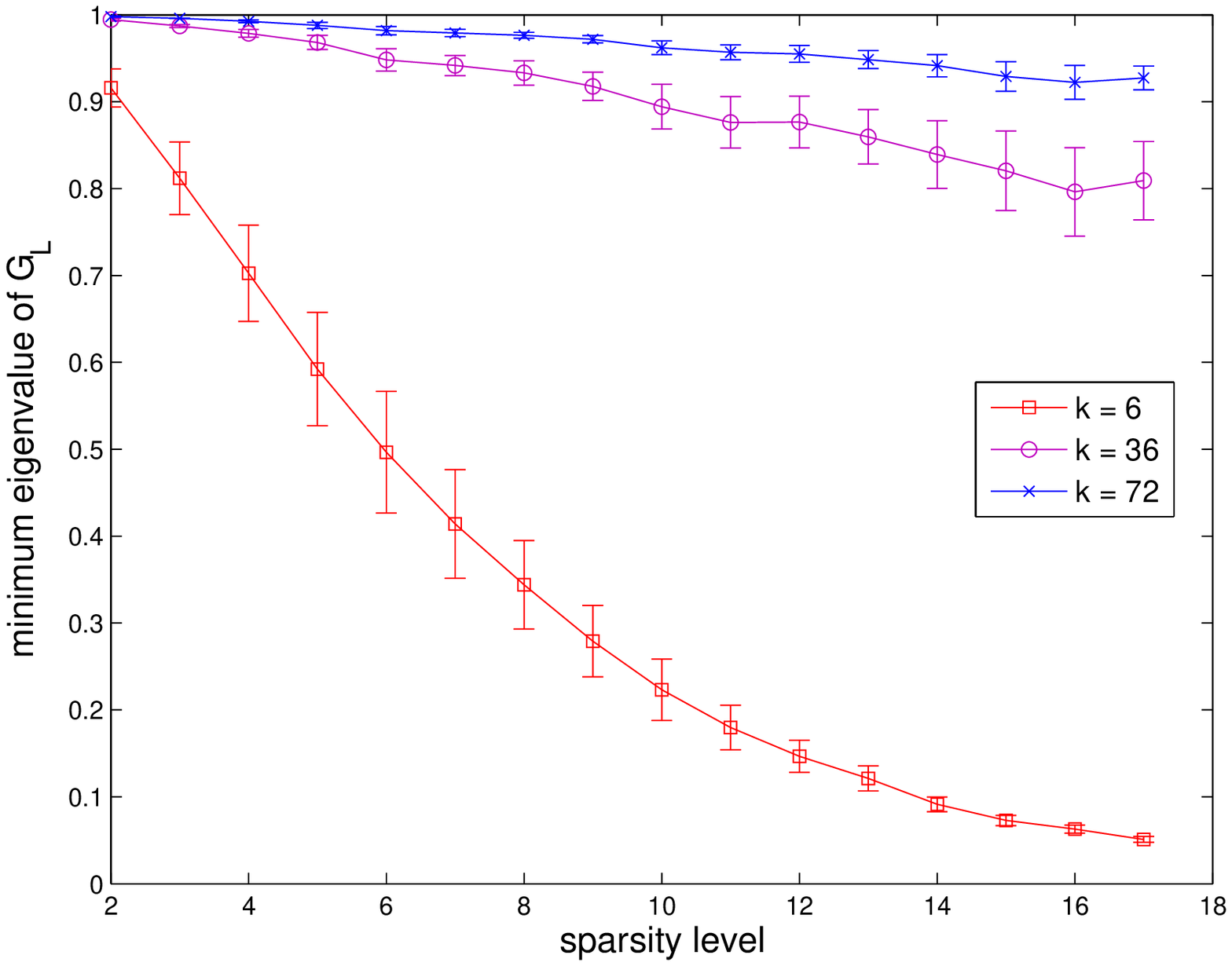}\tabularnewline
(a)  & (b)\tabularnewline
\end{tabular}\caption{\label{fig:dirichlet} (a) The 2-D Dirichlet kernel when $k=k_{1}=k_{2}=6$;
(b) The empirical distribution of the minimum eigenvalue $\sigma_{\min}(\boldsymbol{G}_{\text{L}})$
for various choices of $k$ with respect to the sparsity level. }
\end{figure}

To give the reader a flavor of the incoherence condition, we list
two examples below. For ease of presentation, we assume below 2-D
frequency models with $n_{1}=n_{2}$. Note, however, that the asymmetric
cases and general $K$-dimensional frequency models can be analyzed
in the same manner. 
\begin{itemize}
\item \emph{Random frequency locations}: suppose that the $r$ frequencies
are generated uniformly at random, then the minimum pairwise separation
can be crudely bounded by $\Theta\left(\frac{1}{r^{2}\log n_{1}}\right)$.
If $n_{1}\gg r^{2.5}\log n_{1}$, then a crude bound reveals that
$\forall i_{1}\neq i_{2},$ 
\[
\max\left\{ \frac{1}{k_{1}}\frac{1-\left(y_{i_{1}}^{*}y_{i_{2}}\right)^{k_{1}}}{1-y_{i_{1}}^{*}y_{i_{2}}},\frac{1}{k_{2}}\frac{1-\left(z_{i_{1}}^{*}z_{i_{2}}\right)^{k_{2}}}{1-z_{i_{1}}^{*}z_{i_{2}}}\right\} \ll\frac{1}{\sqrt{r}}
\]
holds with high probability, indicating that the off-diagonal entries
of $\boldsymbol{G}_{\text{L}}$ and $\boldsymbol{G}_{\text{R}}$ are
much smaller than $1/r$ in magnitude. Simple manipulation then allows
us to conclude that $\sigma_{\min}\left(\boldsymbol{G}_{\text{L}}\right)$
and $\sigma_{\min}\left(\boldsymbol{G}_{\text{R}}\right)$ are bounded
below by positive constants. Fig. \ref{fig:dirichlet} (b) shows the
minimum eigenvalue of $\boldsymbol{G}_{\text{L}}$ for different $k=k_{1}=k_{2}=6,36,72$
when the spikes are randomly generated and the number of spikes is
given as the sparsity level. The minimum eigenvalue of $\boldsymbol{G}_{\text{L}}$
gets closer to one as $k$ grows, confirming our argument. 
\item \emph{Small perturbation off the grid}: suppose that all frequencies
are within a distance at most $\frac{1}{n_{1}r^{1/4}}$ from some
grid points $\left(\frac{l_{1}}{k_{1}},\frac{l_{2}}{k_{2}}\right)$
$\left(0\leq l_{1}<k_{1},0\leq l_{2}<k_{2}\right)$. One can verify
that $\forall i_{1}\neq i_{2}$, 
\[
\max\left\{ \frac{1}{k_{1}}\frac{1-\left(y_{i_{1}}^{*}y_{i_{2}}\right)^{k_{1}}}{1-y_{i_{1}}^{*}y_{i_{2}}},\frac{1}{k_{2}}\frac{1-\left(z_{i_{1}}^{*}z_{i_{2}}\right)^{k_{2}}}{1-z_{i_{1}}^{*}z_{i_{2}}}\right\} <\frac{1}{2\sqrt{r}},
\]
and hence the magnitude of all off-diagonal entries of $\boldsymbol{G}_{\text{L}}$
and $\boldsymbol{G}_{\text{R}}$ are no larger than $1/(4r)$. This
immediately suggests that $\sigma_{\min}\left(\boldsymbol{G}_{\text{L}}\right)$
and $\sigma_{\min}\left(\boldsymbol{G}_{\text{R}}\right)$ are lower
bounded by $3/4$.
\end{itemize}
Note, however, that the class of incoherent signals are far beyond
the ones discussed above.

\subsection{Theoretical Guarantees\label{sub:Theoretical-Guarantee}}

With the above incoherence measure, the main theoretical guarantees
are provided in the following three theorems each accounting for a
distinct data model: 1) noiseless measurements, 2) measurements contaminated
by bounded noise, and 3) measurements corrupted by a constant proportion
of arbitrary outliers.

\subsubsection{Exact Recovery from Noiseless Measurements}

Exact recovery is possible from a minimal number of noise-free samples,
as asserted in the following theorem. \begin{theorem}\label{theorem-EMaC-noiseless}Let
$\boldsymbol{X}$ be a data matrix of form (\ref{eq:X_MatrixForm}),
and $\Omega$ the random location set of size $m$. Suppose that the
incoherence property (\ref{eq:LeastSV_G}) holds and that all measurements
are noiseless. Then there exists a universal constant $c_{1}>0$ such
that $\boldsymbol{X}$ is the unique solution to EMaC with probability
exceeding $1-\left(n_{1}n_{2}\right)^{-2}$, provided that 
\begin{equation}
m>c_{1}\mu_{1}c_{\mathrm{s}}r\log^{4}(n_{1}n_{2}).\label{eq:ExactRecovery_Weak}
\end{equation}
\end{theorem}

Theorem~\ref{theorem-EMaC-noiseless} asserts that under some mild\emph{
deterministic} incoherence condition such that $\mu_{1}$ scales as
a small constant, EMaC admits prefect recovery as soon as the number
of measurements exceeds $\mathcal{O}(r\log^{4}\left(n_{1}n_{2}\right))$.
Since there are $\Theta(r)$ degrees of freedom in total, the lower
bound should be no smaller than $\Theta(r)$. This demonstrates the
orderwise optimality of EMaC except for a logarithmic gap. We note,
however, that the polylog factor might be further refined via finer
tuning of concentration of measure inequalities.

It is worth emphasizing that while we assume random observation models,
the data model is assumed deterministic. This differs significantly
from \cite{TangBhaskarShahRecht2012}, which relies on randomness
in both the observation model and the data model. In particular, our
theoretical performance guarantees rely solely on the frequency locations
irrespective of the associated amplitudes. In contrast, the results
in \cite{TangBhaskarShahRecht2012} require the phases of all frequency
spikes to be i.i.d. drawn in a uniform manner in addition to a separation
condition.

\begin{remark}Theorem \ref{theorem-EMaC-noiseless} significantly
strengthens our prior results reported in \cite{chen2013spectral}
by improving the required sample complexity from $\mathcal{O}\left(\mu_{1}^{2}c_{\mathrm{s}}^{2}r^{2}\text{poly}\log(n_{1}n_{2})\right)$
to $\mathcal{O}\left(\mu_{1}c_{\mathrm{s}}r\text{poly}\log(n_{1}n_{2})\right)$.
\end{remark}

\subsubsection{Stable Recovery in the Presence of Bounded Noise\label{sub:Stable-Recovery}}

Our method enables stable recovery even when the time domain samples
are noisy copies of the true data. Here, we say the recovery is stable
if the solution of Noisy-EMaC is close to the ground truth in proportion
to the noise level. To this end, we provide the following theorem,
which is a counterpart of Theorem \ref{theorem-EMaC-noiseless} in
the noisy setting, whose proof is inspired by \cite{CanPla10}.

\begin{theorem}\label{theorem-EMaC-Noisy}Suppose $\boldsymbol{X}^{\mathrm{o}}$
is a noisy copy of $\boldsymbol{X}$ that satisfies $\|\mathcal{P}_{\Omega}(\boldsymbol{X}-\boldsymbol{X}^{\mathrm{o}})\|_{\mathrm{F}}\leq\delta$.
Under the conditions of Theorem \ref{theorem-EMaC-noiseless}, the
solution to Noisy-EMaC in (\ref{eq:EMaCNoisy}) satisfies 
\begin{equation}
\|\hat{\boldsymbol{X}}_{\mathrm{e}}-\boldsymbol{X}_{\mathrm{e}}\|_{\mathrm{F}}\leq5n_{1}^{3}n_{2}^{3}\delta\label{eq:AccuracyNoisyEMaC}
\end{equation}
with probability exceeding $1-(n_{1}n_{2})^{-2}$. \end{theorem}

Theorem \ref{theorem-EMaC-Noisy} reveals that the recovered enhanced
matrix (which contains $\Theta(n_{1}^{2}n_{2}^{2})$ entries) is close
to the true enhanced matrix at high SNR. In particular, the average
entry inaccuracy of the enhanced matrix is bounded above by $\mathcal{O}(n_{1}^{3}n_{2}^{3}\delta)$,
amplified by the subsampling factor. In practice, one is interested
in an estimate of $\boldsymbol{X}$, which can be obtained naively
by randomly selecting an entry in $\Omega_{\text{e}}(k,l)$ as $\hat{X}_{k,l}$,
then we have 
\[
\|\hat{\boldsymbol{X}}-\boldsymbol{X}\|_{\textrm{F}}\leq\|\hat{\boldsymbol{X}}_{\text{e}}-\boldsymbol{X}_{\text{e}}\|_{\textrm{F}}.
\]
This yields that the per-entry noise of $\hat{\boldsymbol{X}}$ is
about $\mathcal{O}(n_{1}^{2.5}n_{2}^{2.5}\delta)$, which is further
amplified due to enhancement by a factor of $n_{1}n_{2}$. However,
this factor arises from an analysis artifact due to our simple strategy
to deduce $\hat{\boldsymbol{X}}$ from $\hat{\boldsymbol{X}_{\text{e}}}$,
and may be elevated. We note that in numerical experiments, Noisy-EMaC
usually generates much better estimates, usually by a polynomial factor.
The practical applicability will be illustrated in Section \ref{sec:Numerical-Experiments}.

It is worth mentioning that to the best of our knowledge, our result
is the first stability result with partially observed data for spectral
compressed sensing off the grid. While the atomic norm approach is
near-minimax with full data \cite{tang2013minimax}, it is not clear
how it performs with partially observed data.

\subsubsection{Robust Recovery in the Presence of Sparse Outliers\label{sub:Robust-Recovery}}

Interestingly, Robust-EMaC can provably tolerate a constant portion
of arbitrary outliers. The theoretical performance is formally summarized
in the following theorem.

\begin{theorem}\label{theorem-EMaC-Robust}Let $\boldsymbol{X}$
be a data matrix with matrix form (\ref{eq:X_MatrixForm}), and $\Omega$
a random location set of size $m$. Set $\lambda=\frac{1}{\sqrt{m\log\left(n_{1}n_{2}\right)}}$,
and assume $\tau\leq0.1$ is some small positive constant. Then there
exist a numerical constant $c_{1}>0$ depending only on $\tau$ such
that if (\ref{eq:LeastSV_G}) holds and 
\begin{equation}
m>c_{1}\mu_{1}^{2}c_{\mathrm{s}}^{2}r^{2}\log^{3}(n_{1}n_{2}),
\end{equation}
then Robust-EMaC is exact, i.e. the minimizer $(\hat{\boldsymbol{M}},\hat{\boldsymbol{S}})$
satisfies $\hat{\boldsymbol{M}}=\boldsymbol{X}$, with probability
exceeding $1-(n_{1}n_{2})^{-2}$. \end{theorem}

\begin{remark}Note that $\tau\leq0.1$ is not a critical threshold.
In fact, one can prove the same theorem for a larger $\tau$ (e.g.
$\tau\leq0.25$) with a larger absolute constant $c_{1}$. However,
to allow even larger $\tau$ (e.g. in the regime where $\tau\geq50\%$),
we need the sparse components exhibit random sign patterns. \end{remark}

Theorem \ref{theorem-EMaC-Robust} specifies a candidate choice of
the regularization parameter $\lambda$ that allows recovery from
a few samples, which only depends on the size of $\Omega$ but is
otherwise parameter-free. In practice, however, $\lambda$ may better
be selected via cross validation. Furthermore, Theorem \ref{theorem-EMaC-Robust}
demonstrates the possibility of robust recovery under a constant proportion
of sparse corruptions. Under the same mild incoherence condition as
for Theorem~\ref{theorem-EMaC-noiseless}, robust recovery is possible
from $\mathcal{O}\left(r^{2}\log^{3}\left(n_{1}n_{2}\right)\right)$
samples, even when a constant proportion of the samples are arbitrarily
corrupted. As far as we know, this provides the first theoretical
guarantees for separating sparse measurement corruptions in the off-grid
compressed sensing setting.

\subsection{Extension to Higher-Dimensional and Damping Frequency Models \label{sub:Extension-to-Higher-Dimension}}

By letting $n_{2}=1$ the above 2-D frequency model reverts to the
line spectrum model. The EMaC algorithm and the main results immediately
extend to higher dimensional frequency models without difficulty.
In fact, for $K$-dimensional frequency models, one can arrange the
original data into a $K$-fold Hankel matrix of rank at most $r$.
For instance, consider a 3-D model such that 
\[
X_{l_{1},l_{2},l_{3}}=\sum_{i=1}^{r}d_{i}y_{i}^{l_{1}}z_{i}^{l_{2}}w_{i}^{l_{3}},\quad\forall\left(l_{1},l_{2},l_{3}\right)\in[n_{1}]\times[n_{2}]\times[n_{3}].
\]
An enhanced form can be defined as a 3-fold Hankel matrix such that
\[
\boldsymbol{X}_{\text{e}}:=\left[\begin{array}{cccc}
\boldsymbol{X}_{0,\text{e}} & \boldsymbol{X}_{1,\text{e}} & \cdots & \boldsymbol{X}_{n_{3}-k_{3},\text{e}}\\
\boldsymbol{X}_{1,\text{e}} & \boldsymbol{X}_{2,\text{e}} & \cdots & \boldsymbol{X}_{n_{3}-k_{3}+1,\text{e}}\\
\vdots & \vdots & \vdots & \vdots\\
\boldsymbol{X}_{k_{3}-1,\text{e}} & \boldsymbol{X}_{k_{1},\text{e}} & \cdots & \boldsymbol{X}_{n_{3}-1,\text{e}}
\end{array}\right],
\]
where $\boldsymbol{X}_{i,\text{e}}$ denotes the 2-D enhanced form
of the matrix consisting of all entries $X_{l_{1},l_{2},l_{3}}$ obeying
$l_{3}=i$. One can verify that $\boldsymbol{X}_{\text{e}}$ is of
rank at most $r$, and can thereby apply EMaC on the 3-D enhanced
form. To summarize, for $K$-dimensional frequency models, EMaC (resp.
Noisy-EMaC, Robust-EMaC) searches over all $K$-fold Hankel matrices
that are consistent with the measurements. The theoretical performance
guarantees can be similarly extended by defining the respective Dirichlet
kernel in 3-D and the coherence measure. In fact, all our analyses
can be extended to handle \textit{damping} modes, when the frequencies
are not of time-invariant amplitudes. We omit the details for conciseness.

\section{Structured Matrix Completion\label{sub:Extension-to-Hankel}}

One problem closely related to our method is completion of multi-fold
Hankel matrices from a small number of entries. While each spectrally
sparse signal can be mapped to a low-rank multi-fold Hankel matrix,
it is not clear whether all multi-fold Hankel matrices of rank $r$
can be written as the enhanced form of a signal with spectral sparsity
$r$. Therefore, one can think of recovery of multi-fold Hankel matrices
as a more general problem than the spectral compressed sensing problem.
Indeed, Hankel matrix completion has found numerous applications in
system identification \cite{Fazel2011hankel,Markovsky2008structured},
natural language processing \cite{Balle2012spectral}, computer vision
\cite{Sankaranarayanan2010compressive}, magnetic resonance imaging
\cite{LustigHankel2013}, etc.

There has been several work concerning algorithms and numerical experiments
for Hankel matrix completions \cite{Fazel2003Hankel,Fazel2011hankel,Markovsky2008structured}.
However, to the best of our knowledge, there has been little theoretical
guarantee that addresses directly Hankel matrix completion. Our analysis
framework can be straightforwardly adapted to the general $K$-fold
Hankel matrix completions. Below we present the performance guarantee
for the two-fold Hankel matrix completion without loss of generality.
Notice that we need to modify the definition of $\mu_{1}$ as stated
in the following theorem. %$\mu_{2}$ and $\mu_{3}$ are defined using the SVD of $\boldsymbol{X}_{\text{e}}$ in \eqref{eq:Inhomogenuity_UV} and \eqref{eq:APtAHomogenuity}, and we only need to

\begin{theorem}\label{theorem-EMaC-Hankel}Consider a two-fold Hankel
matrix $\boldsymbol{X}_{\mathrm{e}}$ of rank $r$. The bounds in
Theorems \ref{theorem-EMaC-noiseless}, \ref{theorem-EMaC-Noisy}
and \ref{theorem-EMaC-Robust} continue to hold, if the incoherence
$\mu_{1}$ is defined as the smallest number that satisfies 
\begin{equation}
\max_{\left(k,l\right)\in[n_{1}]\times[n_{2}]}\left\{ \Vert\boldsymbol{U}^{*}\boldsymbol{A}_{(k,l)}\Vert_{\mathrm{F}}^{2},\Vert\boldsymbol{A}_{(k,l)}\boldsymbol{V}\Vert_{\mathrm{F}}^{2}\right\} \leq\frac{\mu_{1}c_{\mathrm{s}}r}{n_{1}n_{2}}.\label{eq:IncohrenceUU_Hankel}
\end{equation}
\end{theorem} %\begin{IEEEproof}
%See Appendix \ref{proof-EMaC-Hankel}. 
%\end{IEEEproof}
Condition \eqref{eq:IncohrenceUU_Hankel} requires that the left and
right singular vectors are sufficiently uncorrelated with the observation
basis. In fact, condition \eqref{eq:IncohrenceUU_Hankel} is a weaker
assumption than \eqref{eq:LeastSV_G}.

It is worth mentioning that a low-rank Hankel matrix can often be
converted to its low-rank Toeplitz counterpart, by reversely ordering
all rows of the Hankel matrix. Both Hankel and Toeplitz matrices are
effective forms that capture the underlying harmonic structures. Our
results and analysis framework extend to low-rank Toeplitz matrix
completion problem without difficulty.

\section{Numerical Experiments \label{sec:Numerical-Experiments}}

In this section, we present numerical examples to evaluate the performance
of EMaC and its variants under different scenarios. We further examine
the application of EMaC in image super resolution. Finally, we propose
an extension of singular value thresholding (SVT) developed by Cai
et. al. \cite{cai2010singular} that exploits the multi-fold Hankel
structure to handle larger scale data sets.

\subsection{Phase Transition in the Noiseless Setting}

To evaluate the practical ability of the EMaC algorithm, we conducted
a series of numerical experiments to examine the phase transition
for exact recovery. Let $n_{1}=n_{2}$, and we take $k_{1}=k_{2}=\lceil(n_{1}+1)/2\rceil$
which corresponds to the smallest $c_{\text{s}}$. For each $(r,m)$
pair, 100 Monte Carlo trials were conducted. We generated a spectrally
sparse data matrix $\boldsymbol{X}$ by randomly generating $r$ frequency
spikes in $[0,1)\times[0,1)$, and sampled a subset $\Omega$ of size
$m$ entries uniformly at random. The EMaC algorithm was conducted
using the convex programming modeling software CVX with the interior-point
solver SDPT3 \cite{grant2008cvx}. Each trial is declared successful
if the normalized mean squared error (NMSE) satisfies $\|\hat{\boldsymbol{X}}-\boldsymbol{X}\|_{\text{F}}/\|\boldsymbol{X}\|_{\text{F}}\leq10^{-3}$,
where $\hat{\boldsymbol{X}}$ denotes the estimate returned by EMaC.
The empirical success rate is calculated by averaging over 100 Monte
Carlo trials.

Fig. \ref{fig:Phase-transition-plots-Hankel2D} illustrates the results
of these Monte Carlo experiments when the dimensions%
\footnote{We choose the dimension of $\boldsymbol{X}$ to be odd simply to yield
a squared matrix $\boldsymbol{X}_{\mathrm{e}}$. In fact, our results
do not rely on $n_{1}$ or $n_{2}$ being either odd or prime. We
note that when $n_{1}$ and $n_{2}$ are known to be prime numbers,
there might exist computationally cheaper methods to enable perfect
recovery (e.g. \cite{alexeev2012full})%
} of $\boldsymbol{X}$ are $11\times11$ and $15\times15$. The horizontal
axis corresponds to the number $m$ of samples revealed to the algorithm,
while the vertical axis corresponds to the spectral sparsity level
$r$. The empirical success rate is reflected by the color of each
cell. It can be seen from the plot that the number of samples $m$
grows approximately linearly with respect to the spectral sparsity
$r$, and that the slopes of the phase transition lines for two cases
are approximately the same. These observations are in line with our
theoretical guarantee in Theorem \ref{theorem-EMaC-noiseless}. This
phase transition diagrams justify the practical applicability of our
algorithm in the noiseless setting.

\begin{figure}[htp]
\centering%
\begin{tabular}{cc}
\hspace{-0.2in}\includegraphics[width=0.25\textwidth]{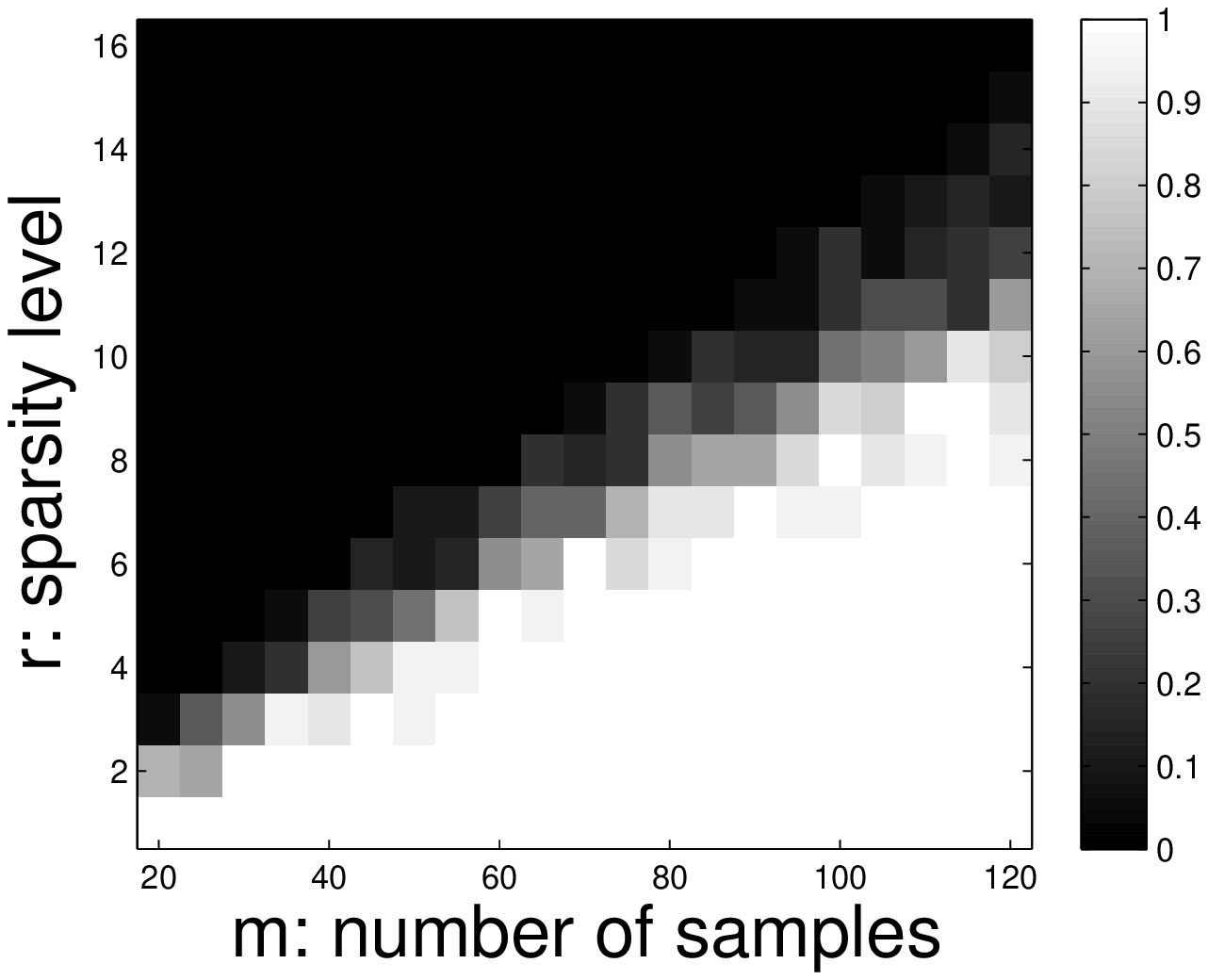}  & \hspace{-0.2in} \includegraphics[width=0.25\textwidth]{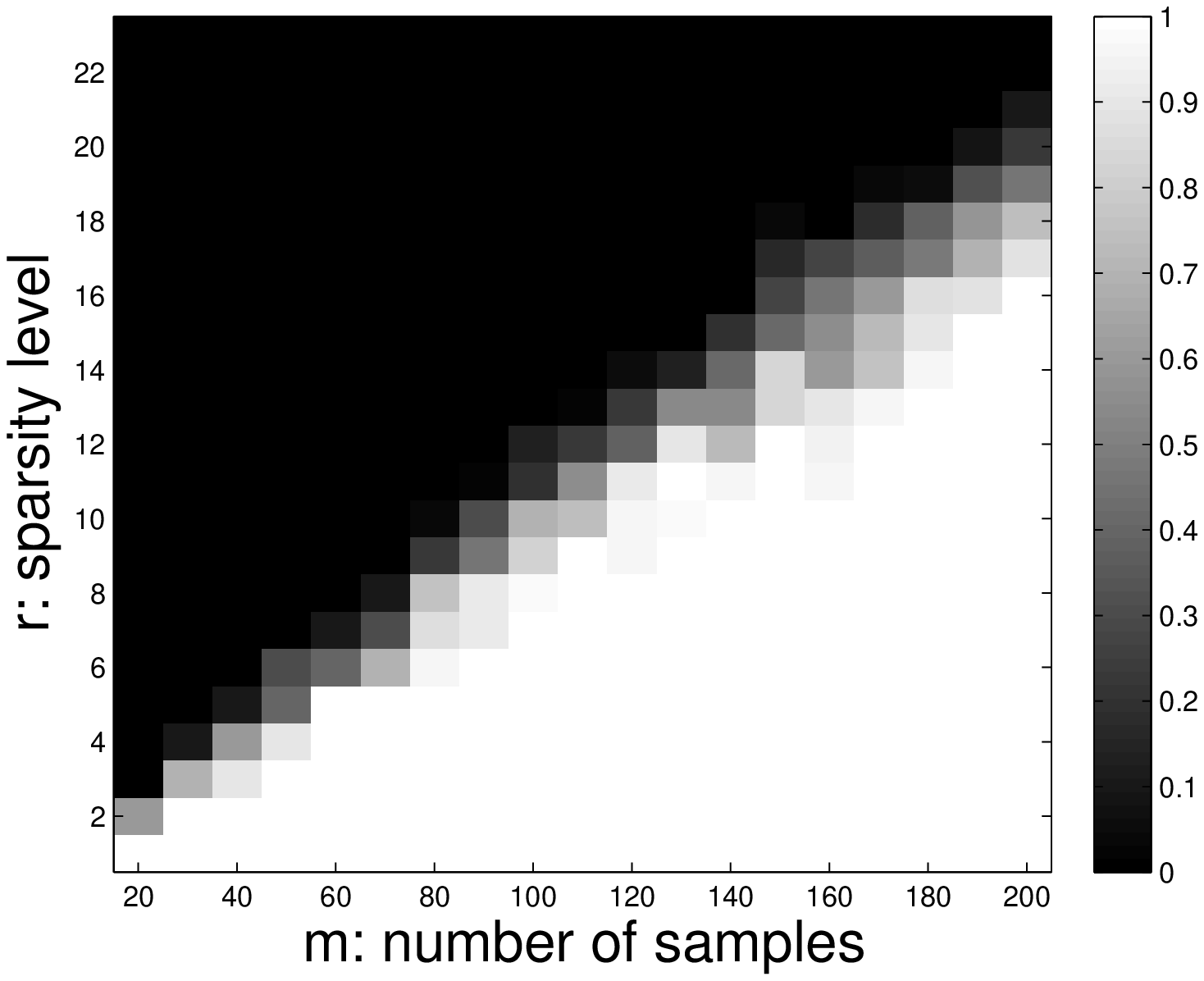}\tabularnewline
(a)  & (b)\tabularnewline
\end{tabular}\caption{\label{fig:Phase-transition-plots-Hankel2D}Phase transition plots
where frequency locations are randomly generated. The plot (a) concerns
the case where $n_{1}=n_{2}=11$, whereas the plot (b) corresponds
to the situation where $n_{1}=n_{2}=15$. The empirical success rate
is calculated by averaging over 100 Monte Carlo trials. }
\end{figure}

\subsection{Stable Recovery from Noisy Data}

Fig.~\ref{stability_plot} further examines the stability of the
proposed algorithm by performing Noisy-EMaC with respect to different
parameter $\delta$ on a noise-free dataset of $r=4$ complex sinusoids
with $n_{1}=n_{2}=11$. The number of random samples is $m=50$. The
reconstructed NMSE grows approximately linear with respect to $\delta$,
validating the stability of the proposed algorithm.

\begin{figure}[h]
\centering\includegraphics[width=0.3\textwidth]{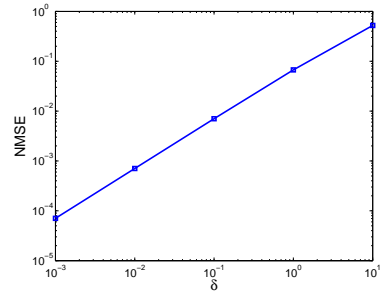} \caption{\label{stability_plot}The reconstruction NMSE with respect to $\delta$
for a dataset with $n_{1}=n_{2}=11$, $r=4$ and $m=50$.}
\end{figure}

\subsection{Comparison with Existing Approaches for Line Spectrum Estimation}

\label{sec:num:comparison}

Suppose that we randomly observe $64$ entries of an $n$-dimensional
vector ($n=127$) composed of $r=4$ modes. For such 1-D signals,
we compare EMaC with the atomic norm approach \cite{TangBhaskarShahRecht2012}
as well as basis pursuit \cite{CheDonSau01} assuming a grid of size
$2^{12}$. For the atomic norm and the EMaC algorithm, the modes are
recovered via linear prediction using the recovered data \cite{scharf1991statistical}.
Fig.~\ref{fig:mode_recovery} demonstrates the recovery of mode locations
for three cases, namely when (a) all the modes are on the DFT grid
along the unit circle; (b) all the modes are on the unit circle except
two closely located modes that are off the presumed grid; (c) all
the modes are on the unit circle except that one of the two closely
located modes is a damping mode with amplitude $0.99$. In all cases,
the EMaC algorithm successfully recovers the underlying modes, while
the atomic norm approach fails to recover damping modes, and basis
pursuit fails with both off-the-grid modes and damping modes. 
\begin{figure*}
\centering%
\begin{tabular}{ccc}
\hspace{-0.2in}\includegraphics[width=0.33\textwidth]{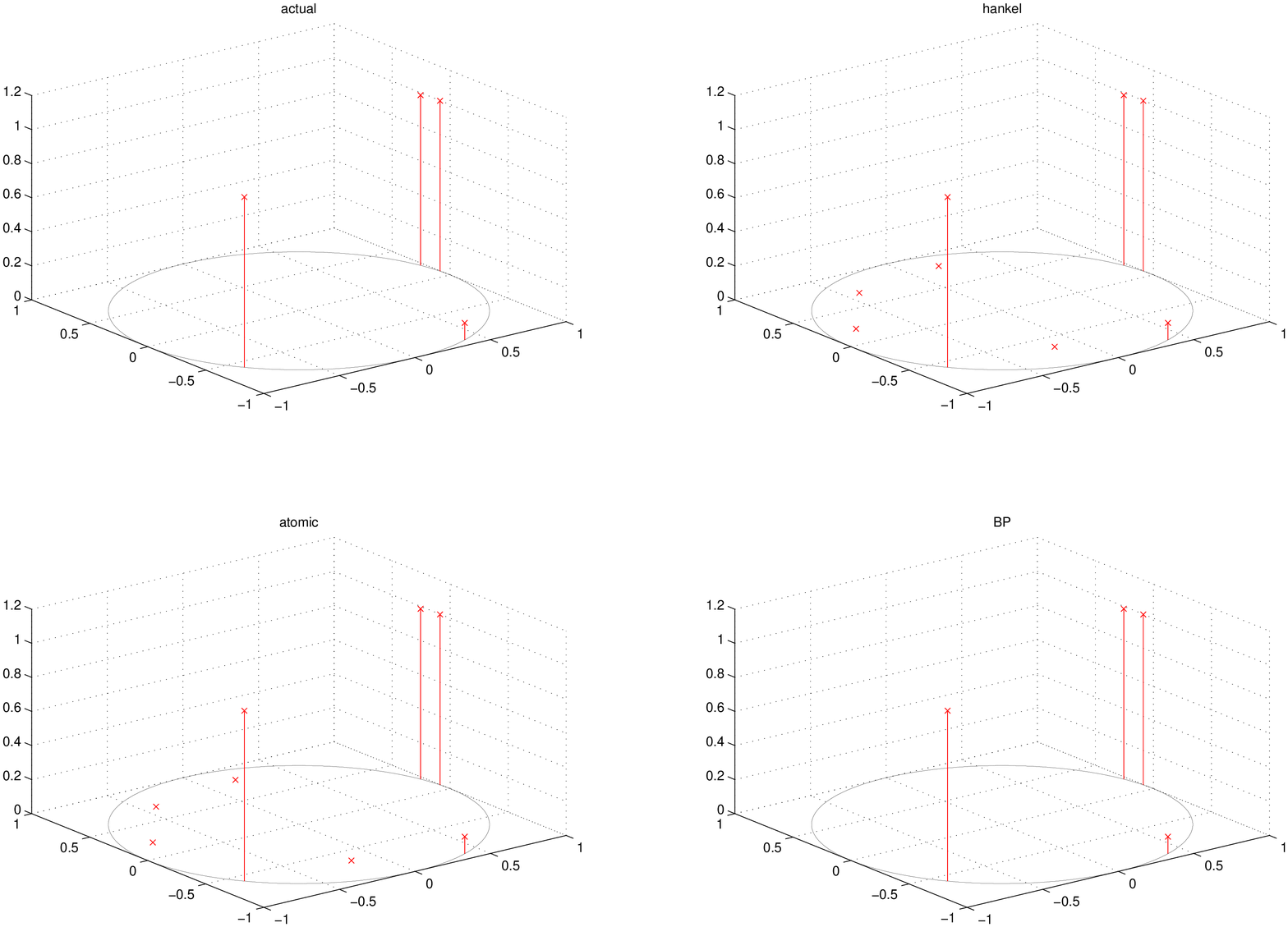}  & \hspace{-0.2in}\includegraphics[width=0.33\textwidth]{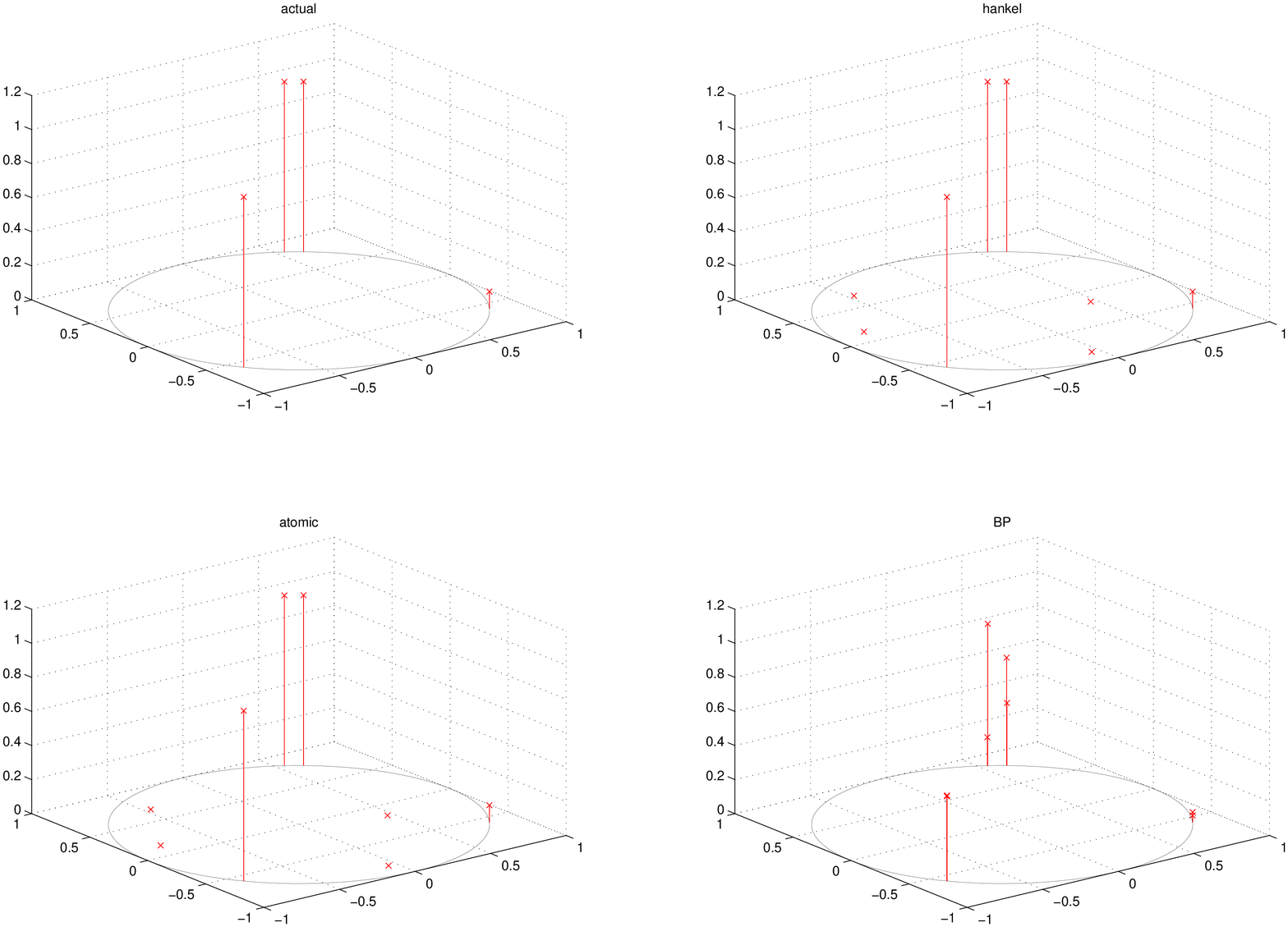} & \hspace{-0.2in} \includegraphics[width=0.33\textwidth]{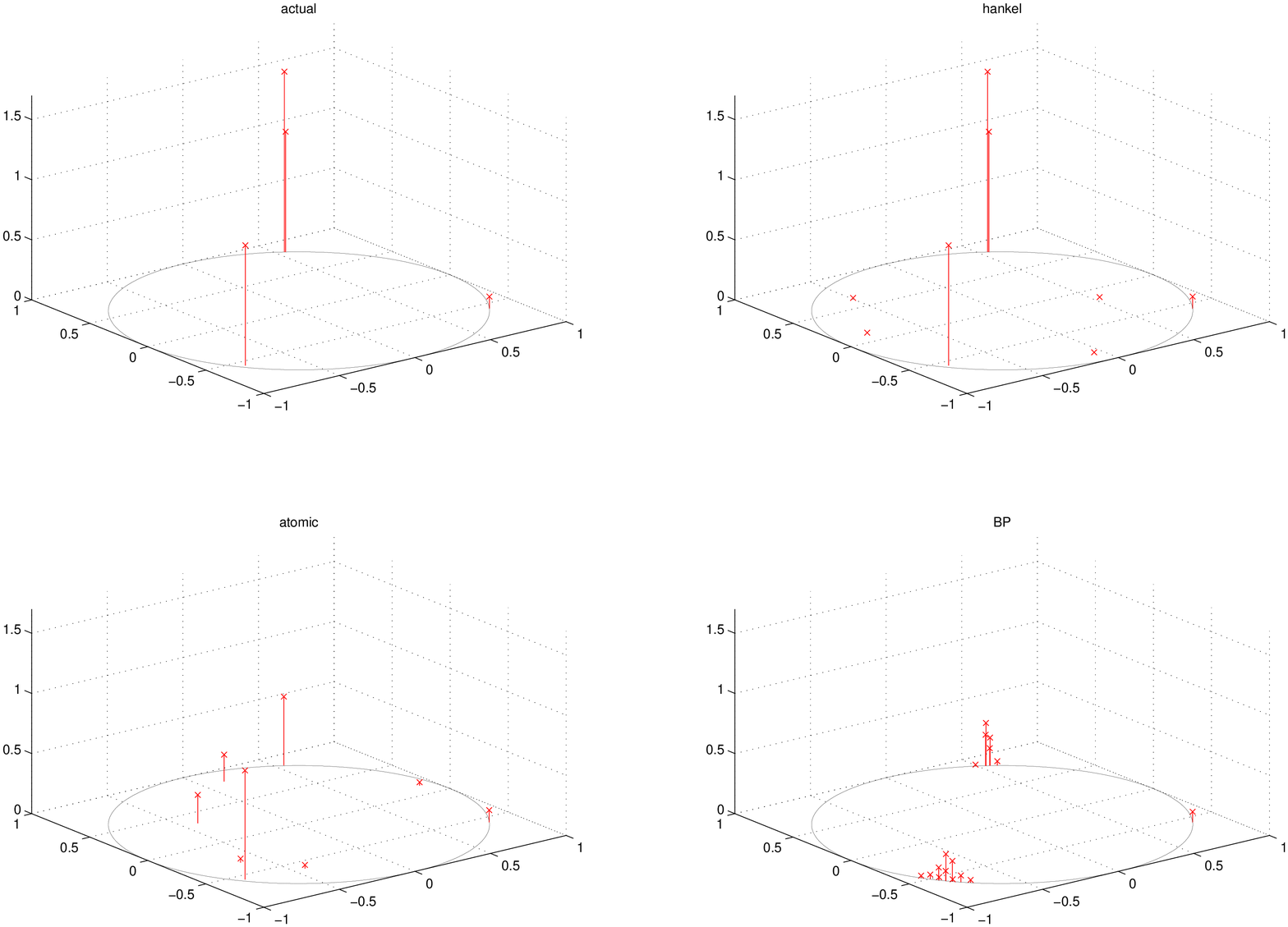} \tabularnewline
(a) On-the-grid  & (b) Frequency mismatch  & (c) Frequency and Damping mismatch \tabularnewline
\end{tabular}\caption{\label{fig:mode_recovery}Recovery of mode locations when (a) all
the modes are on the DFT grid along the unit circle; (b) all the modes
are on the unit circle except two closely located modes that are off
the DFT grid; (c) all the modes are on the unit circle except that
one of the two closely located modes is a damping mode. The panels
from the upper left, clockwise, are the ground truth, the EMaC algorithm,
the atomic norm approach \cite{TangBhaskarShahRecht2012}, and basis
pursuit \cite{CheDonSau01} assuming a grid of size $2^{12}$. }
\end{figure*}

We further compare the phase transition of the EMaC algorithm and
the atomic norm approach in \cite{TangBhaskarShahRecht2012} for line
spectrum estimation. We assume a 1-D signal of length $n=n_{1}=127$
and the pencil parameter $k_{1}$ of EMaC is chosen to be $64$. The
phase transition experiments are conducted in the same manner as Fig.~
\ref{fig:Phase-transition-plots-Hankel2D}. In the first case, the
spikes are generated randomly as Fig.~ \ref{fig:Phase-transition-plots-Hankel2D}
on a unit circle; in the second case, the spikes are generated until
a separation condition is satisfied $\Delta:=\min_{i_{1}\neq i_{2}}|f_{i_{1}}-f_{i_{2}}|\geq1.5/n$.
Fig.~\ref{fig:Phase-transition-plots-Hankel1D} (a) and (b) illustrate
the phase transition of EMaC and the atomic norm approach when the
frequencies are randomly generated without imposing the separation
condition. The performance of the atomic norm approach degenerates
severely when the separation condition is not met; on the other hand,
the EMaC gives a sharp phase transition similar to the 2D case. When
the separation condition is imposed, the phase transition of the atomic
norm approach greatly improves as shown in Fig.~\ref{fig:Phase-transition-plots-Hankel1D}
(c), while the phase transition of EMaC still gives similar performance
as in Fig.~\ref{fig:Phase-transition-plots-Hankel1D} (a) (We omit
the actual phase transition in this case.) However, it is worth mentioning
that when the sparsity level is relatively high, the required separation
condition is in general difficult to be satisfied in practice. In
comparison, EMaC is less sensitive to the separation requirement.

\begin{figure*}[htp]
\centering%
\begin{tabular}{ccc}
\hspace{-0.2in}\includegraphics[width=0.33\textwidth]{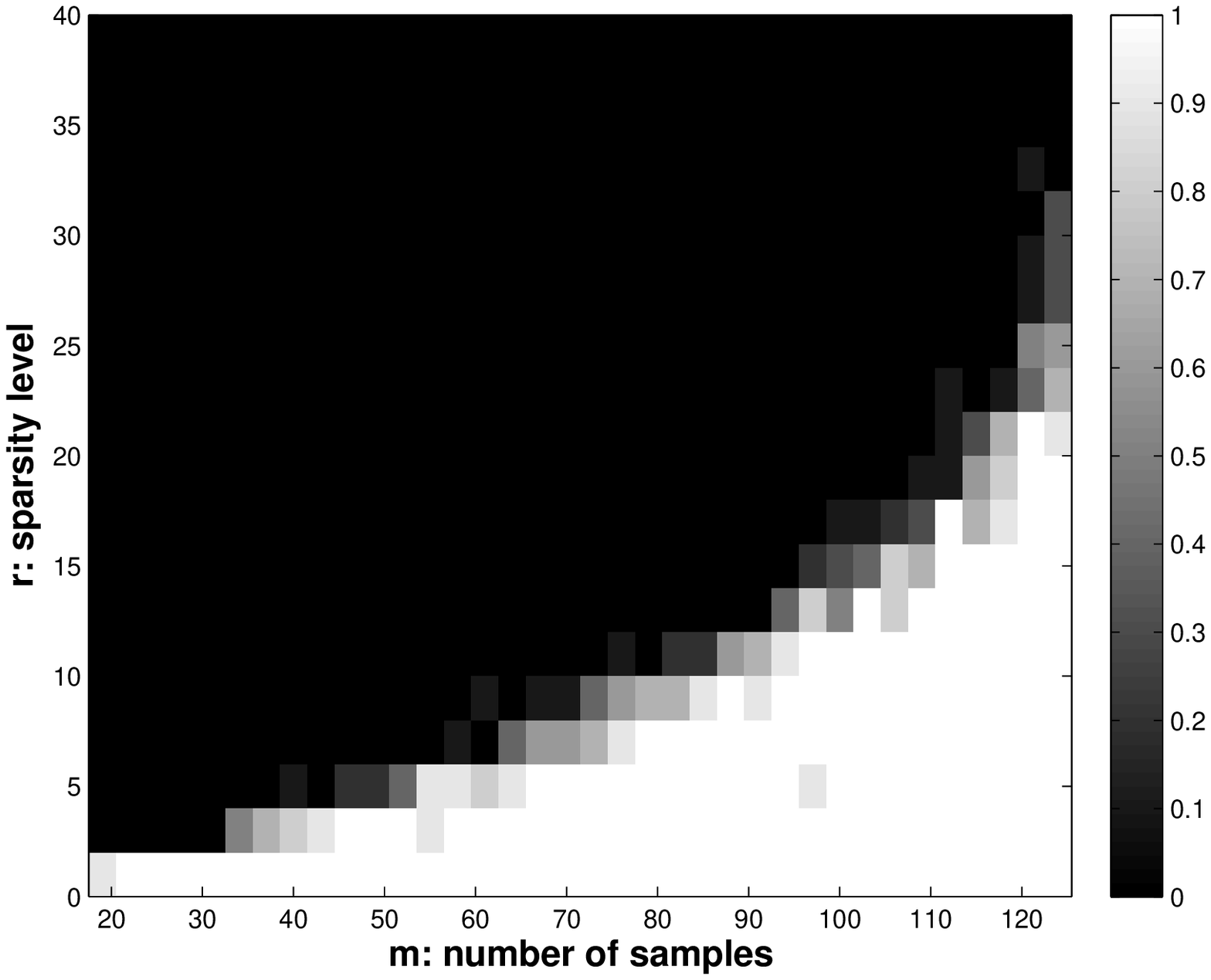}  & \hspace{-0.2in}\includegraphics[width=0.33\textwidth]{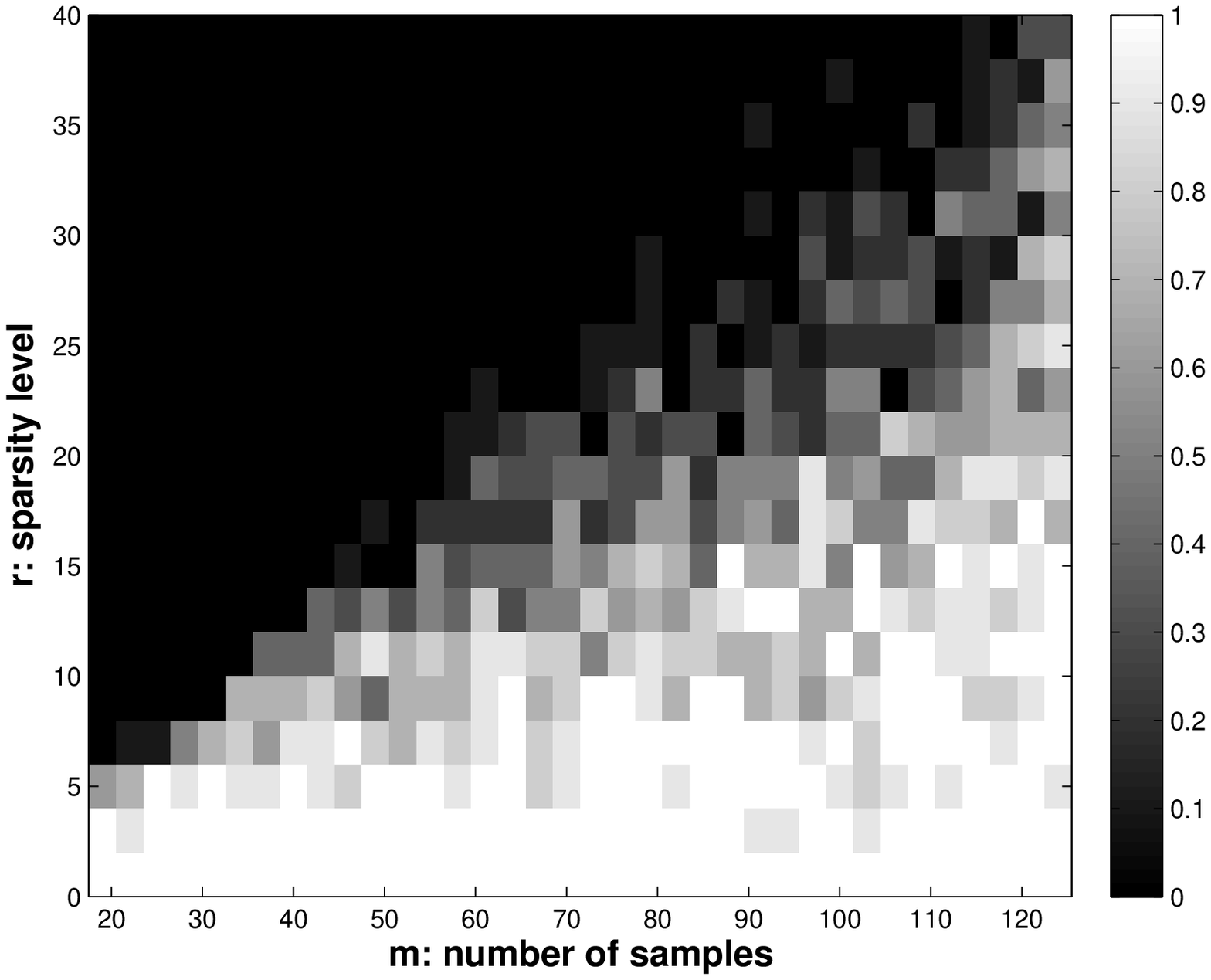}  & \hspace{-0.2in} \includegraphics[width=0.33\textwidth]{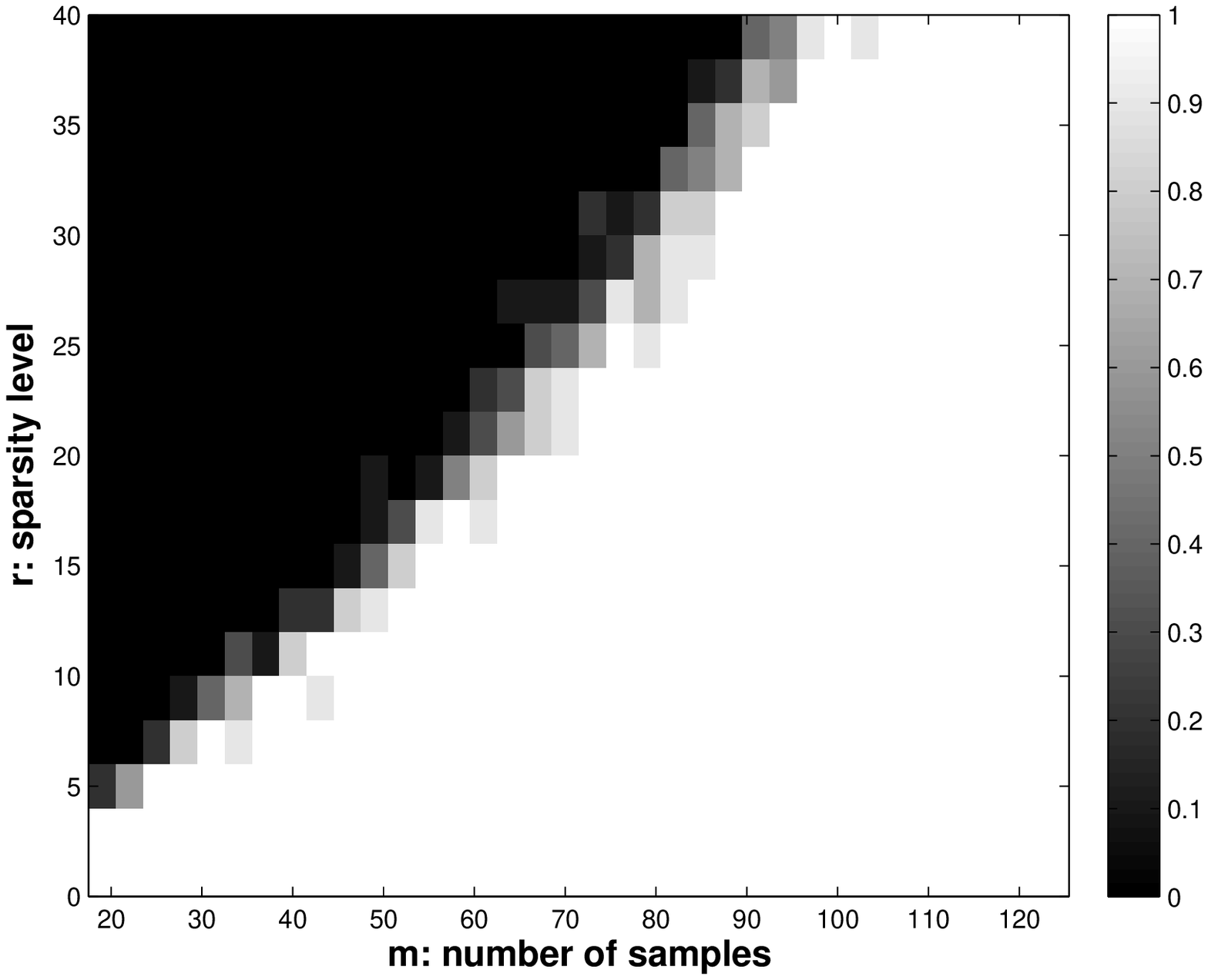} \tabularnewline
\hspace{-0.2in} (a) EMaC without separation  & \hspace{-0.2in} (b) Atomic Norm without separation  & \hspace{-0.2in} (c) Atomic Norm with separation \tabularnewline
\end{tabular}\caption{\label{fig:Phase-transition-plots-Hankel1D}Phase transition for line
spectrum estimation of EMaC and the atomic norm approach \cite{TangBhaskarShahRecht2012}.
(a) EMaC without imposing separation; (b) atomic norm approach without
imposing separation; (c) atomic norm approach with separation. }
\end{figure*}

\subsection{Robust Line Spectrum Estimation}

Consider the problem of line spectrum estimation, where the time domain
measurements are contaminated by a constant portion of outliers. We
conducted a series of Monte Carlo trials to illustrate the phase transition
for perfect recovery of the ground truth. The true data $\boldsymbol{X}$
is assumed to be a $125$-dimensional vector, where the locations
of the underlying frequencies are randomly generated. The simulations
were carried out again using CVX with SDPT3.

Fig. \ref{fig:robust}(a) illustrates the phase transition for robust
line spectrum estimation when $10\%$ of the entries are corrupted,
which showcases the tradeoff between the number $m$ of measurements
and the recoverable spectral sparsity level $r$. One can see from
the plot that $m$ is approximately linear in $r$ on the phase transition
curve even when 10\% of the measurements are corrupted, which validates
our finding in Theorem \ref{theorem-EMaC-Robust}. Fig. \ref{fig:robust}(b)
illustrates the success rate of exact recovery when we obtain samples
for all entry locations. This plot illustrates the tradeoff between
the spectral sparsity level and the number of outliers when all entries
of the corrupted $\boldsymbol{X}^{o}$ are observed. It can be seen
that there is a large region where exact recovery can be guaranteed,
demonstrating the power of our algorithms in the presence of sparse
outliers. 
\begin{figure}[h]
\centering %
\begin{tabular}{cc}
\hspace{-0.1in}\includegraphics[width=0.25\textwidth]{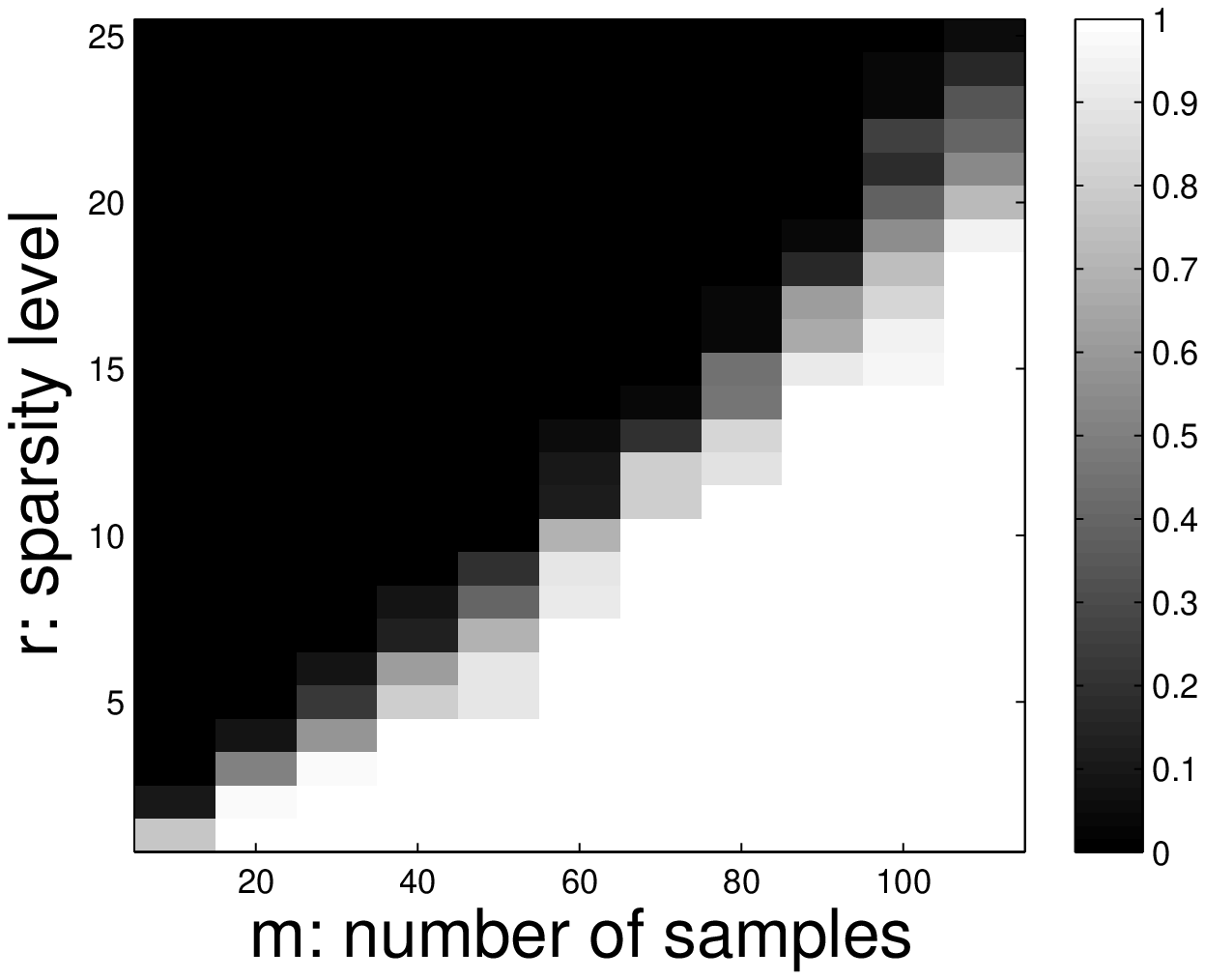}  & \hspace{-0.1in} \includegraphics[width=0.25\textwidth]{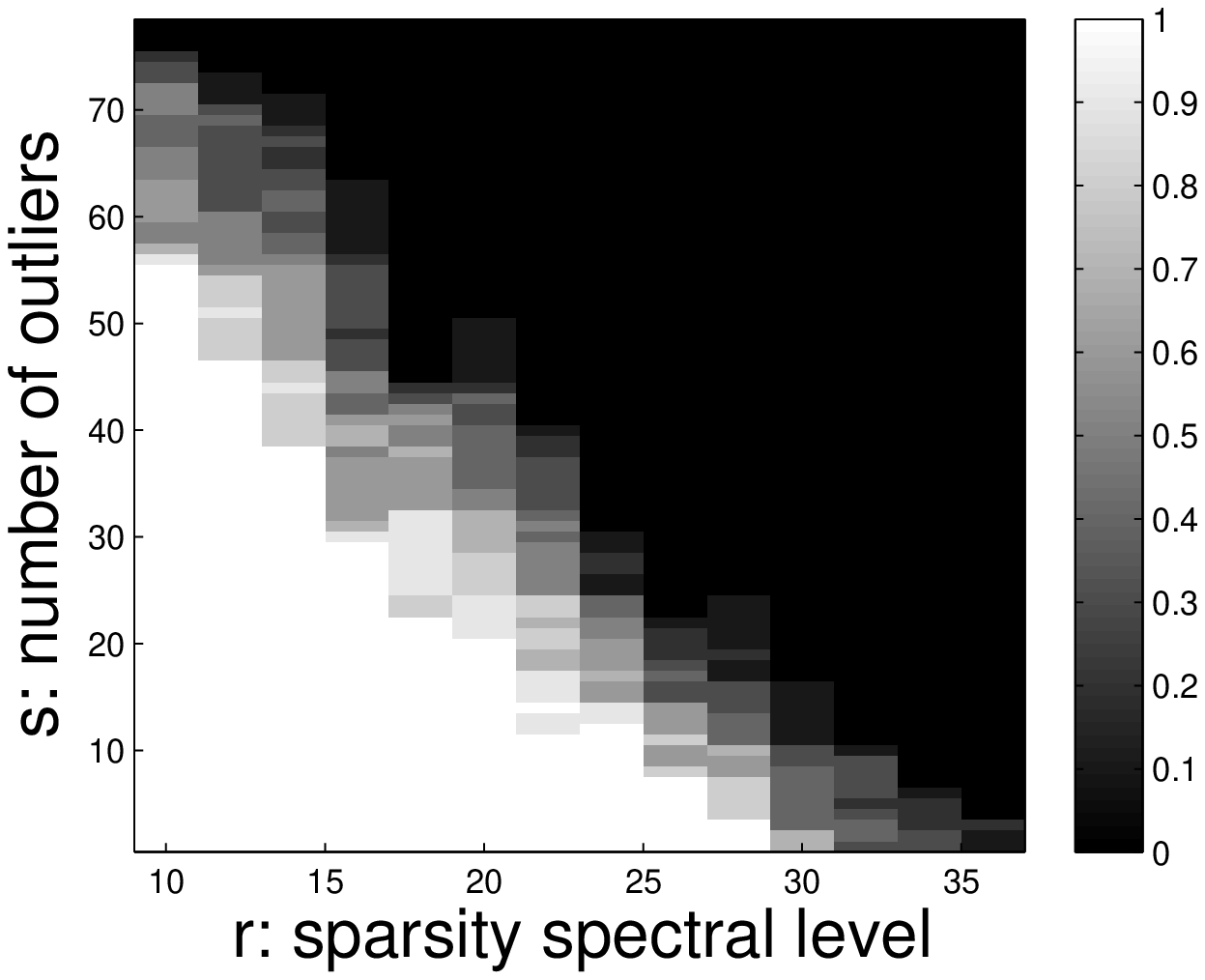} \tabularnewline
(a)  & (b) \tabularnewline
\end{tabular}\caption{\label{fig:robust}Robust line spectrum estimation where mode locations
are randomly generated: (a) Phase transition plots when $n=125$,
and $10\%$ of the entries are corrupted; the empirical success rate
is calculated by averaging over 100 Monte Carlo trials. (b) Phase
transition plots when $n=125$, and all the entries are observed;
the empirical success rate is calculated by averaging over 20 Monte
Carlo trials. }
\end{figure}

\subsection{Synthetic Super Resolution}

The proposed EMaC algorithm works beyond the random observation model
in Theorem~\ref{theorem-EMaC-noiseless}. Fig.~\ref{fig:super_resolution}
considers a synthetic super resolution example motivated by \cite{CandesFernandez2012SR},
where the ground truth in Fig.~\ref{fig:super_resolution}(a) contains
$6$ point sources with constant amplitude. The low-resolution observation
in Fig.~\ref{fig:super_resolution}(b) is obtained by measuring low-frequency
components $[-f_{\mathrm{lo}},f_{\mathrm{lo}}]$ of the ground truth.
Due to the large width of the associated point-spread function, both
the locations and amplitudes of the point sources are distorted in
the low-resolution image.

We apply EMaC to extrapolate high-frequency components up to $[-f_{\mathrm{hi}},f_{\mathrm{hi}}]$,
where $f_{\mathrm{hi}}/f_{\mathrm{lo}}=2$. The reconstruction in
Fig.~\ref{fig:super_resolution}(c) is obtained via applying directly
inverse Fourier transform of the spectrum to avoid parameter estimation
such as the number of modes. The resolution is greatly enhanced from
Fig.~\ref{fig:super_resolution}(b), suggesting that EMaC is a promising
approach for super resolution tasks. The theoretical performance is
left for future work. 
\begin{figure*}
\centering%
\begin{tabular}{ccc}
\hspace{-0.3in} \includegraphics[width=0.3\textwidth,height=2in]{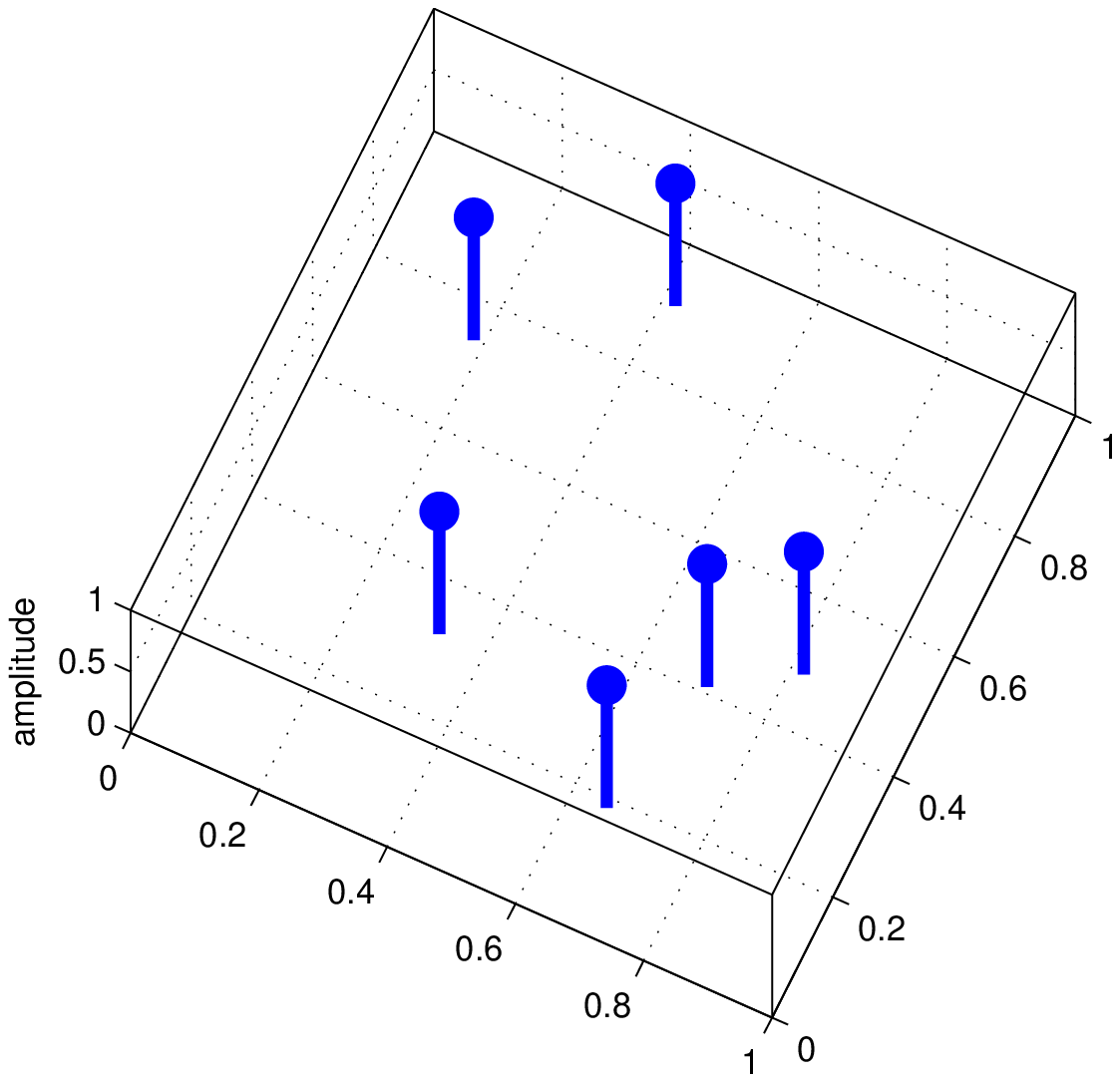}  & \hspace{-0.25in} \includegraphics[width=0.3\textwidth,height=1.8in]{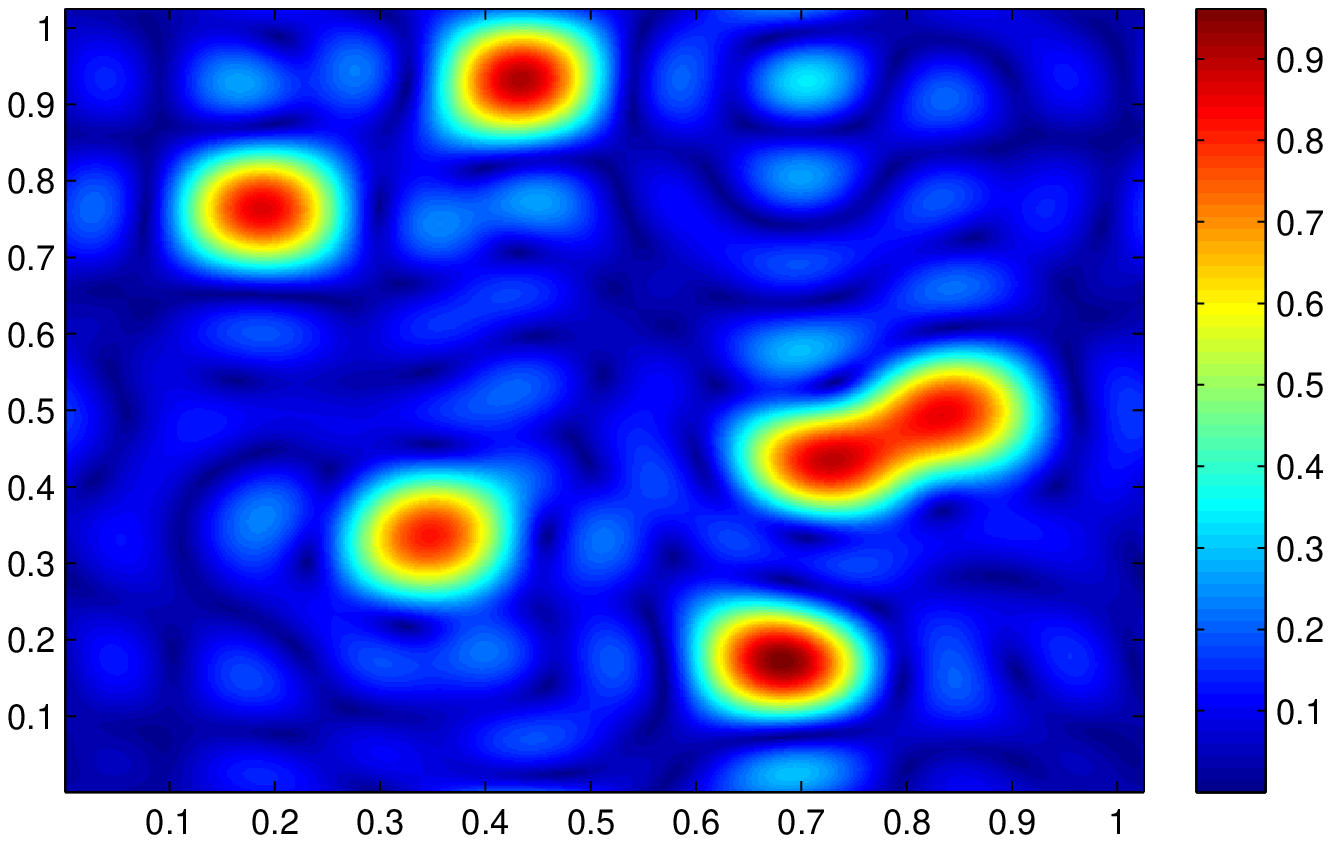}  & \hspace{-0.2in} \includegraphics[width=0.3\textwidth,height=1.8in]{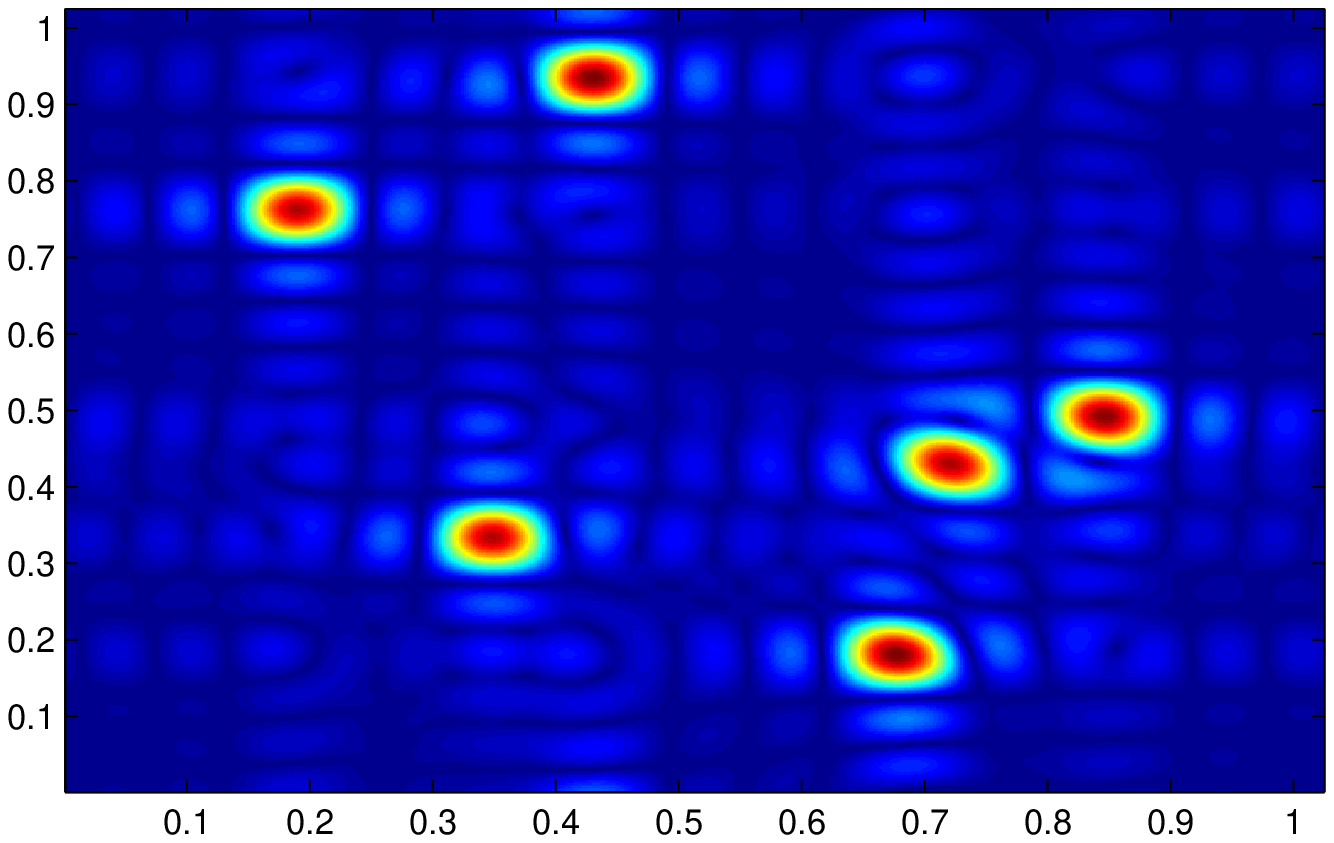} \tabularnewline
\hspace{-0.4in}(a) Ground truth  & \hspace{-0.25in} (b) Low-resolution observation  & (c) High-resolution reconstruction \tabularnewline
\end{tabular}\caption{\label{fig:super_resolution}A synthetic super resolution example,
where the observation (b) is taken from the low-frequency components
of the ground truth in (a), and the reconstruction (c) is done via
inverse Fourier transform of the extrapolated high-frequency components.}
\end{figure*}

\subsection{Singular Value Thresholding for EMaC}

The above Monte Carlo experiments were conducted using the advanced
SDP solver SDPT3. This solver and many other popular ones (e.g. SeDuMi)
are based on interior point methods, which are typically inapplicable
to large-scale data. In fact, SDPT3 fails to handle an $n\times n$
data matrix when $n$ exceeds 19, which corresponds to a $100\times100$
enhanced matrix.

One alternative for large-scale data is the first-order algorithms
tailored to matrix completion problems, e.g. the singular value thresholding
(SVT) algorithm \cite{cai2010singular}. We propose a modified SVT
algorithm in Algorithm \ref{alg:Singular-Value-Thresholding} to exploit
the Hankel structure.

\vspace{10pt}

\begin{algorithm}[h]
\begin{tabular}{l}
\textbf{Input}: The observed data matrix $\boldsymbol{X}^{\text{o}}$
on the location set $\Omega$.\tabularnewline
\textbf{initialize}: let $\boldsymbol{X}_{\text{e}}^{\text{o}}$ denote
the enhanced form of $\mathcal{P}_{\Omega}\left(\boldsymbol{X}^{\text{o}}\right)$; \tabularnewline
\quad{}set $\boldsymbol{M}_{0}=\boldsymbol{X}_{\text{e}}^{\text{0}}$
and $t=0$.\tabularnewline
\textbf{repeat}\tabularnewline
$\quad$1) \textbf{$\boldsymbol{Q}_{t}\leftarrow\mathcal{D}_{\tau_{t}}\left(\boldsymbol{M}_{t}\right)$}\tabularnewline
$\quad$2) $\boldsymbol{M}_{t}\leftarrow\mathcal{H}_{\boldsymbol{X}^{\text{0}}}\left(\boldsymbol{Q}_{t}\right)$\tabularnewline
$\quad$3) $t\leftarrow t+1$\tabularnewline
\textbf{until} convergence\tabularnewline
\textbf{output} $\hat{\boldsymbol{X}}$ as the data matrix with enhanced
form $\boldsymbol{M}_{t}$.\tabularnewline
\end{tabular}

\caption{Singular Value Thresholding for EMaC.\label{alg:Singular-Value-Thresholding}}
\end{algorithm}

\vspace{10pt}
 In particular, two operators are defined as follows: 
\begin{itemize}
\item $\mathcal{D}_{\tau_{t}}(\cdot)$ in Algorithm \ref{alg:Singular-Value-Thresholding}
denotes the singular value shrinkage operator. Specifically, if the
SVD of $\boldsymbol{X}$ is given by $\boldsymbol{X}=\boldsymbol{U}\boldsymbol{\Sigma}\boldsymbol{V}^{*}$
with $\boldsymbol{\Sigma}=\text{diag}\left(\left\{ \sigma_{i}\right\} \right)$,
then 
\[
\mathcal{D}_{\tau_{t}}\left(\boldsymbol{X}\right):=\boldsymbol{U}\text{diag}\left(\left\{ \left(\sigma_{i}-\tau_{t}\right)_{+}\right\} \right)\boldsymbol{V}^{*},
\]
where $\tau_{t}>0$ is the soft-thresholding level. 
\item In the $K$-dimensional frequency model, $\mathcal{H}_{\boldsymbol{X}^{\text{o}}}(\boldsymbol{Q}_{t})$
denotes the projection of $\boldsymbol{Q}_{t}$ onto the subspace
of enhanced matrices (i.e. $K$-fold Hankel matrices) that are consistent
with the observed entries. 
\end{itemize}
Consequently, at each iteration, a pair $\left(\boldsymbol{Q}_{t},\boldsymbol{M}_{t}\right)$
is produced by first performing singular value shrinkage and then
projecting the outcome onto the space of $K$-fold Hankel matrices
that are consistent with observed entries.

The key parameter that one needs to tune is the threshold $\tau_{t}$.
Unfortunately, there is no universal consensus regarding how to tweak
the threshold for SVT type of algorithms. One suggested choice is
$\tau_{t}=0.1\sigma_{\max}\left(\boldsymbol{M}_{t}\right)/\left\lceil \frac{t}{10}\right\rceil $,
which works well based on our empirical experiments.

Fig. \ref{fig:SVTNoisy} illustrates the performance of Algorithm
\ref{alg:Singular-Value-Thresholding}. We generated a true $101\times101$
data matrix $\boldsymbol{X}$ through a superposition of $30$ random
complex sinusoids, and revealed 5.8\% of the total entries (i.e. $m=600$)
uniformly at random. The noise was i.i.d. Gaussian giving a signal-to-noise
amplitude ratio of $10$. The reconstructed vectorized signal is superimposed
on the ground truth in Fig. \ref{fig:SVTNoisy}. The normalized reconstruction
error was $\|\hat{\boldsymbol{X}}-\boldsymbol{X}\|_{\text{F}}/\left\Vert \boldsymbol{X}\right\Vert _{\text{F}}=0.1098$,
validating the stability of our algorithm in the presence of noise.

\begin{figure}[h]
\includegraphics[width=0.52\textwidth]{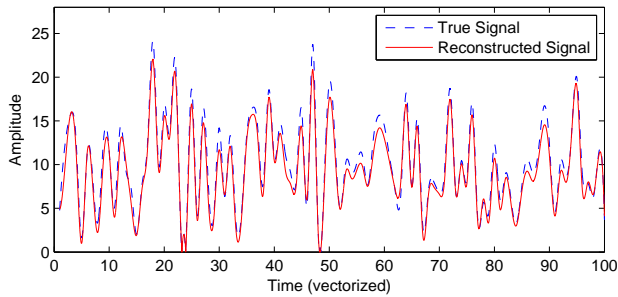}
\caption{\label{fig:SVTNoisy}The performance of SVT for Noisy-EMaC for a $101\times101$
data matrix that contains 30 random frequency spikes. 5.8\% of all
entries ($m=600$) are observed with signal-to-noise amplitude ratio
10. Here, $\tau_{t}=0.1\sigma_{\max}\left(\boldsymbol{M}_{t}\right)/\left\lceil \frac{t}{10}\right\rceil $
empirically. For concreteness, the reconstructed data against the
true data for the first 100 time instances (after vectorization) are
plotted.}
\end{figure}

\section{Proof of Theorems~\ref{theorem-EMaC-noiseless} and \ref{theorem-EMaC-Hankel}\label{sec:Main-Proof-Exact}}

EMaC has similar spirit as the well-known matrix completion algorithms
\cite{ExactMC09,Gross2011recovering}, except that we impose Hankel
and multi-fold Hankel structures on the matrices. While \cite{Gross2011recovering}
has presented a general sufficient condition for exact recovery (see
\cite[Theorem 3]{Gross2011recovering}), the basis in our case does
not exhibit desired coherence properties as required in \cite{Gross2011recovering},
and hence these results cannot deliver informative estimates when
applied to our problem. Nevertheless, the beautiful golfing scheme
introduced in \cite{Gross2011recovering} lays the foundation of our
analysis in the sequel. We also note that the analyses adopted in
\cite{ExactMC09,Gross2011recovering} rely on a desired joint incoherence
property on $\boldsymbol{U}\boldsymbol{V}^{*}$, which has been shown
to be unnecessary \cite{chen2013incoherence}.

For concreteness, the analyses in this paper focus on recovering harmonically
sparse signals as stated in Theorem \ref{theorem-EMaC-noiseless},
since proving Theorem \ref{theorem-EMaC-noiseless} is slightly more
involved than proving Theorem \ref{theorem-EMaC-Hankel}. We note,
however, that our analysis already entails all reasoning required
for establishing Theorem \ref{theorem-EMaC-Hankel}.

\begin{comment}
In this section, we restrict our attention to the real case (i.e.
$\boldsymbol{X}$ is real-valued) for simplicity%
\footnote{All of the proof arguments in this paper work for Hermitian structured
matrices (e.g. Hermitian Toeplitz matrices), and apply also to non-Hermitian
complex case with some minor adjustment in the constants (see, e.g.,
\cite[Section III.D]{Gross2011recovering}). %
}.

\begin{theorem}\label{theorem-EMaC-noiseless-V2}Suppose that $\boldsymbol{X}$
has incoherence measure $\left(\mu_{1},\mu_{2},\mu_{3},\mu_{4}\right)$.
If 
\[
m>c_{0}\max\left(\mu_{1}c_{\mathrm{s}},\mu_{2},\mu_{4}\right)r\log^{2}\left(n_{1}n_{2}\right),
\]
then $\boldsymbol{X}$ is the unique solution of EMaC with probability
exceeding $1-(n_{1}n_{2})^{-2}$. \end{theorem}

Note that Theorem \ref{theorem-EMaC-noiseless} can be delivered as
an immediate consequence of Theorem \ref{theorem-EMaC-noiseless-V2}
by exploiting the relations among $\left(\mu_{1},\mu_{2},\mu_{3},\mu_{4}\right)$
given in Lemma \ref{lemma-UpperBoundAbPtAa}. 
\end{comment}

\subsection{Dual Certification\label{sub:Duality-Noiseless}}

Denote by $\mathcal{A}_{\left(k,l\right)}\left(\boldsymbol{M}\right)$
the projection of $\boldsymbol{M}$ onto the subspace spanned by $\boldsymbol{A}_{(k,l)}$,
and define the projection operator onto the space spanned by all $\boldsymbol{A}_{(k,l)}$
and its orthogonal complement as 
\begin{equation}
\mathcal{A}:=\sum_{\left(k,l\right)\in[n_{1}]\times[n_{2}]}\mathcal{A}_{(k,l)},\quad\text{and}\quad\mathcal{A}^{\perp}=\mathcal{I}-\mathcal{A}.
\end{equation}

There are two common ways to describe the randomness of $\Omega$:
one corresponds to sampling \emph{without} replacement, and another
concerns sampling \emph{with} replacement (i.e. $\Omega$ contains
$m$ indices $\left\{ \boldsymbol{a}_{i}\in[n_{1}]\times[n_{2}]:1\leq i\leq m\right\} $
that are i.i.d. generated). As discussed in \cite[Section II.A]{Gross2011recovering},
while both situations result in the same order-wide bounds, the latter
situation admits simpler analysis due to independence. Therefore,
we will assume that $\Omega$ is a multi-set (possibly with repeated
elements) and $a_{i}$'s are independently and uniformly distributed
throughout the proofs of this paper, and define the associated operators
as 
\begin{equation}
\mathcal{A}_{\Omega}:=\sum_{i=1}^{m}\mathcal{A}_{\boldsymbol{a}_{i}}.\label{eq:DefinitionAOperator}
\end{equation}
We also define another projection operator $\mathcal{A}'_{\Omega}$
similar to (\ref{eq:DefinitionAOperator}), but with the sum extending
only over \emph{distinct} samples. Its complement operator is defined
as $\mathcal{A}'_{\Omega^{\perp}}:=\mathcal{A}-\mathcal{A}'_{\Omega}$.
Note that $\mathcal{A}_{\Omega}\left(\boldsymbol{M}\right)=0$ is
equivalent to $\mathcal{A}'_{\Omega}(\boldsymbol{M})=0$. With these
definitions, EMaC can be rewritten as the following general matrix
completion problem: 
\begin{align}
\underset{\boldsymbol{M}}{\text{minimize}}\quad & \left\Vert \boldsymbol{M}\right\Vert _{*}\label{eq:EMaCRewrite}\\
\text{subject to}\quad & \mathcal{A}'_{\Omega}\left(\boldsymbol{M}\right)=\mathcal{A}'_{\Omega}\left(\boldsymbol{X}_{\text{e}}\right),\nonumber \\
 & \mathcal{A}^{\perp}\left(\boldsymbol{M}\right)=\mathcal{A}^{\perp}\left(\boldsymbol{X}_{\text{e}}\right)=0.\nonumber 
\end{align}

To prove exact recovery of convex optimization, it suffices to produce
an appropriate dual certificate, as stated in the following lemma.

\begin{lemma}\label{lemma-Dual-Certificate}Consider a multi-set
$\Omega$ that contains $m$ random indices. Suppose that the sampling
operator $\mathcal{A}_{\Omega}$ obeys 
\begin{equation}
\left\Vert \mathcal{P}_{T}\mathcal{A}\mathcal{P}_{T}-\frac{n_{1}n_{2}}{m}\mathcal{P}_{T}\mathcal{A}_{\Omega}\mathcal{P}_{T}\right\Vert \leq\frac{1}{2}.\label{eq:WellConditionPtAomegaPt}
\end{equation}
If there exists a matrix $\boldsymbol{W}$ satisfying 
\begin{equation}
\mathcal{A}'_{\Omega^{\perp}}\left(\boldsymbol{W}\right)=0,\label{eq:UV_W_Contained_in_AOmega}
\end{equation}
\begin{equation}
\left\Vert \mathcal{P}_{T}\left(\boldsymbol{W}-\boldsymbol{U}\boldsymbol{V}^{*}\right)\right\Vert _{\mathrm{F}}\leq\frac{1}{2n_{1}^{2}n_{2}^{2}},\label{eq:W_component_T}
\end{equation}
and 
\begin{equation}
\left\Vert \mathcal{P}_{T^{\perp}}\left(\boldsymbol{W}\right)\right\Vert \leq\frac{1}{2},\label{eq:NormWTPerp}
\end{equation}
then $\boldsymbol{X}_{\mathrm{e}}$ is the unique solution to (\ref{eq:EMaCRewrite})
or, equivalently, $\boldsymbol{X}$ is the unique minimizer of EMaC.\end{lemma}

\begin{proof} See Appendix \ref{sec:Proof-of-Lemma-Dual-Certificate}.
\end{proof}

Condition (\ref{eq:WellConditionPtAomegaPt}) will be analyzed in
Section \ref{sub:Injectivity}, while a dual certificate $\boldsymbol{W}$
will be constructed in Section \ref{sub:Dual-Certificate-Noiseless}.
The validity of $\boldsymbol{W}$ as a dual certificate will be established
in Sections \ref{sub:Dual-Certificate-Noiseless} - \ref{sub:BoundW}.
These are the focus of the remaining section.

\subsection{Deviation of $\left\Vert \mathcal{P}_{T}\mathcal{A}\mathcal{P}_{T}-\frac{n_{1}n_{2}}{m}\mathcal{P}_{T}\mathcal{A}_{\Omega}\mathcal{P}_{T}\right\Vert $\label{sub:Injectivity}}

Lemma \ref{lemma-Dual-Certificate} requires that $\mathcal{A}_{\Omega}$
be sufficiently incoherent with respect to the tangent space $T$.
The following lemma quantifies the projection of each $\boldsymbol{A}_{(k,l)}$
onto the subspace $T$.

\begin{lemma}\label{lemma-IncoherenceT_W}Under the hypothesis \eqref{eq:LeastSV_G},
one has 
\begin{equation}
\left\Vert \boldsymbol{U}\boldsymbol{U}^{*}\boldsymbol{A}_{\left(k,l\right)}\right\Vert _{\mathrm{F}}^{2}\leq\frac{\mu_{1}c_{\mathrm{s}}r}{n_{1}n_{2}},\quad\left\Vert \boldsymbol{A}_{\left(k,l\right)}\boldsymbol{V}\boldsymbol{V}^{*}\right\Vert _{\mathrm{F}}^{2}\leq\frac{\mu_{1}c_{\mathrm{s}}r}{n_{1}n_{2}},\label{eq:IncoherenceUV_W}
\end{equation}
for all $\left(k,l\right)\in[n_{1}]\times[n_{2}]$. For any $\boldsymbol{a},\boldsymbol{b}\in[n_{1}]\times[n_{2}]$,
one has 
\begin{equation}
\left|\left\langle \boldsymbol{A}_{\boldsymbol{b}},\mathcal{P}_{T}\boldsymbol{A}_{\boldsymbol{a}}\right\rangle \right|\leq\sqrt{\frac{\omega_{\boldsymbol{b}}}{\omega_{\boldsymbol{a}}}}\frac{3\mu_{1}c_{\mathrm{s}}r}{n_{1}n_{2}}.\label{eq:UBAbPtAa}
\end{equation}
\end{lemma}

\begin{proof}See Appendix \ref{sec:Proof-of-Lemma-lemma-IncoherenceT_W}.\end{proof}

Recognizing that \eqref{eq:IncoherenceUV_W} is the same as \eqref{eq:IncohrenceUU_Hankel},
the following proof also establishes Theorem~\ref{theorem-EMaC-Hankel}.
Note that Lemma~\ref{lemma-IncoherenceT_W} immediately leads to
\begin{align}
\left\Vert \mathcal{P}_{T}\left(\boldsymbol{A}_{(k,l)}\right)\right\Vert _{\text{F}}^{2} & \leq\left\Vert \mathcal{P}_{U}\left(\boldsymbol{A}_{(k,l)}\right)\right\Vert _{\text{F}}^{2}+\left\Vert \mathcal{P}_{V}\left(\boldsymbol{A}_{(k,l)}\right)\right\Vert _{\text{F}}^{2}\nonumber \\
 & \leq\frac{2\mu_{1}c_{\mathrm{s}}r}{n_{1}n_{2}}.\label{eq:IncoherenceT_W}
\end{align}

%By definition, we have the identities 
%\begin{align*}
%\left\Vert \mathcal{P}_{T}\left(\boldsymbol{A}_{(k,l)}\right)\right\Vert _{\text{F}}^{2} & =\left\langle \mathcal{P}_{T}\left(\boldsymbol{A}_{(k,l)}\right),\boldsymbol{A}_{(k,l)}\right\rangle \\
% & =\left\langle \mathcal{P}_{U}\left(\boldsymbol{A}_{(k,l)}\right)+\mathcal{P}_{V}\left(\boldsymbol{A}_{(k,l)}\right)-\mathcal{P}_{U}\mathcal{P}_{V}\left(\boldsymbol{A}_{(k,l)}\right),\boldsymbol{A}_{(k,l)}\right\rangle \\
% & =\left\Vert \mathcal{P}_{U}\left(\boldsymbol{A}_{(k,l)}\right)\right\Vert _{\text{F}}^{2}+\left\Vert \mathcal{P}_{V}\left(\boldsymbol{A}_{(k,l)}\right)\right\Vert _{\text{F}}^{2}-\left\Vert \mathcal{P}_{U}\mathcal{P}_{V}\left(\boldsymbol{A}_{(k,l)}\right)\right\Vert _{\text{F}}^{2}.
%\end{align*}

As long as (\ref{eq:IncoherenceT_W}) holds, the fluctuation of $\mathcal{P}_{T}\mathcal{A}_{\Omega}\mathcal{P}_{T}$
can be controlled reasonably well, as stated in the following lemma.
This justifies Condition (\ref{eq:WellConditionPtAomegaPt}) as required
by Lemma \ref{lemma-Dual-Certificate}.

\begin{lemma}\label{lemma-Invertibility-PtWPt}Suppose that \eqref{eq:IncoherenceT_W}
holds. %\[
%\left\Vert \mathcal{P}_{T}\left(\boldsymbol{A}_{\left(k,l\right)}\right)\right\Vert _{\mathrm{F}}^{2}\leq\frac{2\mu_{1}c_{\mathrm{s}}r}{n_{1}n_{2}},
%\]
%for $\left(k,l\right)\in[n_{1}]\times[n_{2}]$. 
Then for any small constant $0<\epsilon\leq\frac{1}{2}$, one has
\begin{equation}
\left\Vert \frac{n_{1}n_{2}}{m}\mathcal{P}_{T}\mathcal{A}_{\Omega}\mathcal{P}_{T}-\mathcal{P}_{T}\mathcal{A}\mathcal{P}_{T}\right\Vert \leq\epsilon\label{eq:Invertibility-PtWPt}
\end{equation}
with probability exceeding $1-\left(n_{1}n_{2}\right)^{-4}$, provided
that $m>c_{1}\mu_{1}c_{\mathrm{s}}r\log\left(n_{1}n_{2}\right)$ for
some universal constant $c_{1}>0$. \end{lemma} %$1-2n_{1}n_{2}\exp\left(-\frac{\delta^{2}m}{16\mu_{1}c_{\mathrm{s}}r}\right)$.

\begin{proof}See Appendix \ref{sec:Proof-of-Lemma-lemma-Invertibility-PtWPt}.\end{proof}

\begin{comment}
Putting the above two lemmas together leads to the following fact:
for any given constant $\epsilon<e^{-1}<\frac{1}{2}$, $\left\Vert \frac{n_{1}n_{2}}{m}\mathcal{P}_{T}\mathcal{A}_{\Omega}\mathcal{P}_{T}-\mathcal{P}_{T}\mathcal{A}\mathcal{P}_{T}\right\Vert \leq\epsilon$
holds with probability exceeding $1-\left(n_{1}n_{2}\right)^{-4}$,
provided that $m>c_{1}\mu_{1}c_{\text{s}}r\log\left(n_{1}n_{2}\right)$
for some universal constant $c_{1}>0$. 
\end{comment}

\subsection{Construction of Dual Certificates\label{sub:Dual-Certificate-Noiseless}}

Now we are in a position to construct the dual certificate, for which
we will employ the golfing scheme introduced in \cite{Gross2011recovering}.
Suppose that we generate $j_{0}$ independent random location multi-sets
$\Omega_{i}$ ($1\leq i\leq j_{0}$), each containing $\frac{m}{j_{0}}$
i.i.d. samples. This way the distribution of $\Omega$ is the same
as $\Omega_{1}\cup\Omega_{2}\cup\cdots\cup\Omega_{j_{0}}$ . Note
that $\Omega_{i}$'s correspond to sampling \emph{with} replacement.
Let 
\begin{equation}
\rho:=\frac{m}{n_{1}n_{2}}\quad\text{and}\quad q:=\frac{\rho}{j_{0}}
\end{equation}
represent the undersampling factors of $\Omega$ and $\Omega_{i}$,
respectively.

Consider a small constant $\epsilon<\frac{1}{e}$, and pick $j_{0}:=3\log_{\frac{1}{\epsilon}}n_{1}n_{2}$.
The construction of the dual matrix $\boldsymbol{W}$ then proceeds
as follows:

\vspace{10pt}
\begin{tabular}{>{\raggedright}p{0.45\textwidth}}
\hline 
\textbf{Construction of a dual certificate $\boldsymbol{W}$ via the
golfing scheme.}\tabularnewline
\hline 
\noalign{\vskip\doublerulesep} $\quad\quad$1. Set $\boldsymbol{F}_{0}=\boldsymbol{U}\boldsymbol{V}^{*}$,
and $j_{0}:=5\log_{\frac{1}{\epsilon}}(n_{1}n_{2})$.\tabularnewline
$\quad\quad$2. For all $i$ ($1\leq i\leq j_{0}$), let $\boldsymbol{F}_{i}=\mathcal{P}_{T}\left(\mathcal{A}-\frac{1}{q}\mathcal{A}_{\Omega_{i}}\right)\mathcal{P}_{T}\left(\boldsymbol{F}_{i-1}\right).$\tabularnewline
$\quad\quad$3. Set $\boldsymbol{W}:=\sum_{j=1}^{j_{0}}\left(\frac{1}{q}\mathcal{A}_{\Omega_{i}}+\mathcal{A}^{\perp}\right)\left(\boldsymbol{F}_{i-1}\right)$.\tabularnewline
\hline 
\end{tabular}

\vspace{10pt}

We will establish that $\boldsymbol{W}$ is a valid dual certificate
by showing that $\boldsymbol{W}$ satisfies the conditions stated
in Lemma \ref{lemma-Dual-Certificate}, which we now proceed step
by step.

First, by construction, all summands 
\[
\left(\frac{1}{q}\mathcal{A}_{\Omega_{i}}+\mathcal{A}^{\perp}\right)\left(\boldsymbol{F}_{i-1}\right)
\]
lie within the subspace of matrices supported on $\Omega$ or the
subspace $\mathcal{A}^{\perp}$. This validates that $\mathcal{A}'_{\Omega^{\perp}}\left(\boldsymbol{W}\right)=0$,
as required in (\ref{eq:UV_W_Contained_in_AOmega}).

Secondly, the recursive construction procedure of $\boldsymbol{F}_{i}$
allows us to write 
\begin{align}
 & -\mathcal{P}_{T}\left(\boldsymbol{W}-\boldsymbol{F}_{0}\right)=\mathcal{P}_{T}\left(\boldsymbol{F}_{0}\right)-\sum_{j=1}^{j_{0}}\mathcal{P}_{T}\left(\frac{1}{q}\mathcal{A}_{\Omega_{i}}+\mathcal{A}^{\perp}\right)\left(\boldsymbol{F}_{i-1}\right)\nonumber \\
 & \quad=\mathcal{P}_{T}\left(\boldsymbol{F}_{0}\right)-\mathcal{P}_{T}\left(\frac{1}{q}\mathcal{A}_{\Omega_{i}}+\mathcal{A}^{\perp}\right)\mathcal{P}_{T}\left(\boldsymbol{F}_{0}\right)\nonumber \\
 & \quad\quad\quad\quad\quad\quad\quad-\sum_{j=2}^{j_{0}}\mathcal{P}_{T}\left(\frac{1}{q}\mathcal{A}_{\Omega_{i}}+\mathcal{A}^{\perp}\right)\left(\boldsymbol{F}_{i-1}\right)\nonumber \\
 & \quad=\mathcal{P}_{T}\left(\mathcal{A}-\frac{1}{q}\mathcal{A}_{\Omega_{i}}\right)\mathcal{P}_{T}\left(\boldsymbol{F}_{0}\right)-\sum_{j=2}^{j_{0}}\mathcal{P}_{T}\left(\frac{1}{q}\mathcal{A}_{\Omega_{i}}+\mathcal{A}^{\perp}\right)\boldsymbol{F}_{i-1}\nonumber \\
 & \quad=\mathcal{P}_{T}\left(\boldsymbol{F}_{1}\right)-\sum_{j=2}^{j_{0}}\mathcal{P}_{T}\left(\frac{1}{q}\mathcal{A}_{\Omega_{i}}+\mathcal{A}^{\perp}\right)\left(\boldsymbol{F}_{i-1}\right)\nonumber \\
 & \quad=\cdots=\mathcal{P}_{T}\left(\boldsymbol{F}_{j_{0}}\right).\label{eq:PtWUV_bound}
\end{align}
Lemma \ref{lemma-Invertibility-PtWPt} asserts the following: if $qn_{1}n_{2}\geq c_{1}\mu_{1}c_{\text{s}}r\log\left(n_{1}n_{2}\right)$
or, equivalently, $m\geq\tilde{c}_{1}\mu_{1}c_{\text{s}}r\log^{2}(n_{1}n_{2})$
for some constant $\tilde{c}_{1}>0$, then with overwhelming probability
one has 
\begin{align*}
\left\Vert \mathcal{P}_{T}-\mathcal{P}_{T}\left(\frac{1}{q}\mathcal{A}_{\Omega_{i}}+\mathcal{A}^{\perp}\right)\mathcal{P}_{T}\right\Vert  & =\left\Vert \mathcal{P}_{T}\mathcal{A}\mathcal{P}_{T}-\frac{1}{q}\mathcal{P}_{T}\mathcal{A}_{\Omega_{i}}\mathcal{P}_{T}\right\Vert \\
 & \leq\epsilon<\frac{1}{2}.
\end{align*}
This allows us to bound $\left\Vert \mathcal{P}_{T}\left(\boldsymbol{F}_{i}\right)\right\Vert _{\text{F}}$
as 
\[
\left\Vert \mathcal{P}_{T}\left(\boldsymbol{F}_{i}\right)\right\Vert _{\text{F}}\leq\epsilon^{i}\left\Vert \mathcal{P}_{T}\left(\boldsymbol{F}_{0}\right)\right\Vert _{\text{F}}\leq\epsilon^{i}\left\Vert \boldsymbol{U}\boldsymbol{V}^{*}\right\Vert _{\text{F}}=\epsilon^{i}\sqrt{r},
\]
which together with (\ref{eq:PtWUV_bound}) gives 
\begin{align}
\left\Vert \mathcal{P}_{T}\left(\boldsymbol{W}-\boldsymbol{U}\boldsymbol{V}^{*}\right)\right\Vert _{\text{F}} & =\left\Vert \mathcal{P}_{T}\left(\boldsymbol{W}-\boldsymbol{F}_{0}\right)\right\Vert _{\text{F}}=\left\Vert \mathcal{P}_{T}\left(\boldsymbol{F}_{j_{0}}\right)\right\Vert _{\text{F}}\nonumber \\
 & \leq\epsilon^{j_{0}}\sqrt{r}<\frac{1}{2n_{1}^{2}n_{2}^{2}}\label{eq:PtW_UB}
\end{align}
as required in Condition (\ref{eq:W_component_T}).

Finally, it remains to be shown that $\left\Vert \mathcal{P}_{T^{\perp}}\left(\boldsymbol{W}\right)\right\Vert \leq\frac{1}{2}$,
which we will establish in the next two subsections. In particular,
we first introduce two key metrics and characterize their relationships
in Section \ref{sub:Key-Lemmas}. These metrics are crucial in bounding
$\left\Vert \mathcal{P}_{T^{\perp}}\left(\boldsymbol{W}\right)\right\Vert $,
which will be the focus of Section \ref{sub:BoundW}.

\subsection{Two Metrics and Key Lemmas\label{sub:Key-Lemmas}}

In this subsection, we introduce the following two norms

\begin{equation}
\left\Vert \boldsymbol{M}\right\Vert _{\mathcal{A},\infty}:=\max_{(k,l)\in\left[n_{1}\right]\times\left[n_{2}\right]}\left|\frac{\left\langle \boldsymbol{A}_{(k,l)},\boldsymbol{M}\right\rangle }{\sqrt{\omega_{k,l}}}\right|,\label{eq:DefnAinf}
\end{equation}
\begin{equation}
\left\Vert \boldsymbol{M}\right\Vert _{\mathcal{A},2}:=\sqrt{\sum_{\left(k,l\right)\in[n_{1}]\times[n_{2}]}\frac{\left|\left\langle \boldsymbol{A}_{(k,l)},\boldsymbol{M}\right\rangle \right|^{2}}{\omega_{k,l}}}.\label{eq:DefnA2}
\end{equation}
Based on these two metrics, we can derive several technical lemmas
which, taken collectively, allow us to control $\left\Vert \mathcal{P}_{T^{\perp}}\left(\boldsymbol{W}\right)\right\Vert $.
% this is repetitive.
%We assume throughout all these lemmas that $\Omega$ is a multi-set obtained by sampling with replacement $m$ indices uniformly at random.
Specifically, these lemmas characterize the mutual dependence of three
norms $\left\Vert \cdot\right\Vert $, $\left\Vert \cdot\right\Vert _{\mathcal{A},2}$
and $\left\Vert \cdot\right\Vert _{\mathcal{A},\infty}$.

\begin{lemma}\label{lemma:OpNorm_Anorm}For any given matrix $\boldsymbol{M}$,
there exists some numerical constant $c_{2}>0$ such that 
\begin{align}
 & \left\Vert \left(\frac{n_{1}n_{2}}{m}\mathcal{A}_{\Omega}-\mathcal{A}\right)\left(\boldsymbol{M}\right)\right\Vert \leq c_{2}\sqrt{\frac{n_{1}n_{2}\log\left(n_{1}n_{2}\right)}{m}}\left\Vert \boldsymbol{M}\right\Vert _{\mathcal{A},2}\nonumber \\
 & \quad\quad\quad\quad\quad\quad\quad\quad+c_{2}\frac{n_{1}n_{2}\log\left(n_{1}n_{2}\right)}{m}\left\Vert \boldsymbol{M}\right\Vert _{\mathcal{A},\infty}\label{eq:NormOpt_bound}
\end{align}
with probability at least $1-\left(n_{1}n_{2}\right)^{-10}$.\end{lemma}

\begin{proof}See Appendix \ref{sec:proof-lemma:OpNorm_Anorm}.\end{proof}

\begin{lemma}\label{lemma:NormA2_bound}Assume that there exists
a quantity $\mu_{5}$ such that 
\begin{equation}
\omega_{\alpha,\beta}\left\Vert \mathcal{P}_{T}\left(\boldsymbol{A}_{(\alpha,\beta)}\right)\right\Vert _{\mathcal{A},2}^{2}\leq\frac{\mu_{5}r}{n_{1}n_{2}},\text{ }(\alpha,\beta)\in\left[n_{1}\right]\times\left[n_{2}\right].\label{eq:HypothesisMu5}
\end{equation}
For any given matrix $\boldsymbol{M}$, with probability exceeding
$1-\left(n_{1}n_{2}\right)^{-10}$, 
\begin{align}
 & \left\Vert \left(\frac{n_{1}n_{2}}{m}\mathcal{P}_{T}\mathcal{A}_{\Omega}-\mathcal{P}_{T}\mathcal{A}\right)\left(\boldsymbol{M}\right)\right\Vert _{\mathcal{A},2}\leq c_{3}\sqrt{\frac{\mu_{5}r\log\left(n_{1}n_{2}\right)}{m}}\cdot\nonumber \\
 & \quad\quad\quad\left(\left\Vert \boldsymbol{M}\right\Vert _{\mathcal{A},2}+\sqrt{\frac{n_{1}n_{2}\log\left(n_{1}n_{2}\right)}{m}}\left\Vert \boldsymbol{M}\right\Vert _{\mathcal{A},\infty}\right)\label{eq:NormA2_bound}
\end{align}
for some absolute constant $c_{3}>0$.\end{lemma}

\begin{proof}See Appendix \ref{sec:Proof-of-Lemma:NormA2_bound}.\end{proof}

\begin{lemma}\label{lemma:NormAinf_bound}For any given matrix $\boldsymbol{M}\in T$,
there is some absolute constant $c_{4}>0$ such that 
\begin{align}
 & \left\Vert \left(\frac{n_{1}n_{2}}{m}\mathcal{P}_{T}\mathcal{A}_{\Omega}-\mathcal{P}_{T}\mathcal{A}\right)\left(\boldsymbol{M}\right)\right\Vert _{\mathcal{A},\infty}\nonumber \\
 & \quad\leq c_{4}\sqrt{\frac{\mu_{1}c_{\mathrm{s}}r\log\left(n_{1}n_{2}\right)}{m}}\cdot\sqrt{\frac{\mu_{1}c_{\mathrm{s}}r}{n_{1}n_{2}}}\left\Vert \boldsymbol{M}\right\Vert _{\mathcal{A},2}\nonumber \\
 & \quad\quad+c_{4}\frac{\mu_{1}c_{\mathrm{s}}r\log\left(n_{1}n_{2}\right)}{m}\left\Vert \boldsymbol{M}\right\Vert _{\mathcal{A},\infty}\label{eq:NormAinf_bound}
\end{align}
with probability exceeding $1-\left(n_{1}n_{2}\right)^{-10}$.\end{lemma}

\begin{proof}See Appendix \ref{sec:Proof-of-Lemma:NormAinf_bound}.\end{proof}

Lemma \ref{lemma:NormA2_bound} combined with Lemma \ref{lemma:NormAinf_bound}
gives rise to the following inequality. Consider any given matrix
$\boldsymbol{M}\in T$. Applying the bounds (\ref{eq:NormA2_bound})
and (\ref{eq:NormAinf_bound}), one can derive 
\begin{align}
 & \left\Vert \left(\frac{n_{1}n_{2}}{m}\mathcal{P}_{T}\mathcal{A}_{\Omega}-\mathcal{P}_{T}\mathcal{A}\right)\left(\boldsymbol{M}\right)\right\Vert _{\mathcal{A},2}+\nonumber \\
 & \quad\quad\quad\sqrt{\frac{n_{1}n_{2}\log\left(n_{1}n_{2}\right)}{m}}\left\Vert \left(\frac{n_{1}n_{2}}{m}\mathcal{P}_{T}\mathcal{A}_{\Omega}-\mathcal{P}_{T}\mathcal{A}\right)\left(\boldsymbol{M}\right)\right\Vert _{\mathcal{A},\infty}\\
 & \leq\footnotesize c_{3}\sqrt{\frac{\mu_{5}r\log\left(n_{1}n_{2}\right)}{m}}\left(\left\Vert \boldsymbol{M}\right\Vert _{\mathcal{A},2}+\sqrt{\frac{n_{1}n_{2}\log\left(n_{1}n_{2}\right)}{m}}\left\Vert \boldsymbol{M}\right\Vert _{\mathcal{A},\infty}\right)\nonumber \\
 & \quad\quad\footnotesize+c_{4}\sqrt{\frac{n_{1}n_{2}\log\left(n_{1}n_{2}\right)}{m}}\left(\sqrt{\frac{\mu_{1}c_{\mathrm{s}}r\log\left(n_{1}n_{2}\right)}{m}}\cdot\sqrt{\frac{\mu_{1}c_{\mathrm{s}}r}{n_{1}n_{2}}}\left\Vert \boldsymbol{M}\right\Vert _{\mathcal{A},2}\right.\nonumber \\
 & \quad\quad\quad\quad\left.+\frac{\mu_{1}c_{\mathrm{s}}r\log\left(n_{1}n_{2}\right)}{m}\left\Vert \boldsymbol{M}\right\Vert _{\mathcal{A},\infty}\right)\nonumber \\
 & \leq c_{5}\left(\sqrt{\frac{\mu_{5}r\log\left(n_{1}n_{2}\right)}{m}}+\frac{\mu_{1}c_{\mathrm{s}}r\log\left(n_{1}n_{2}\right)}{m}\right)\cdot\nonumber \\
 & \quad\quad\quad\left\{ \left\Vert \boldsymbol{M}\right\Vert _{\mathcal{A},2}+\sqrt{\frac{n_{1}n_{2}\log\left(n_{1}n_{2}\right)}{m}}\left\Vert \boldsymbol{M}\right\Vert _{\mathcal{A},\infty}\right\} ,\label{eq:NormA2Ainf_bound}
\end{align}
with probability exceeding $1-\left(n_{1}n_{2}\right)^{-10}$, where
$c_{5}=\max\left\{ c_{3},c_{4}\right\} $. This holds under the hypothesis
(\ref{eq:HypothesisMu5}). %completing the proof.
%\begin{lemma}\label{lemma:NormAinf_A2_bound}Consider any given matrix
%$\boldsymbol{M}\in T$. With probability exceeding $1-\left(n_{1}n_{2}\right)^{-10}$,
%\begin{align}
% & \left\Vert \left(\frac{n_{1}n_{2}}{m}\mathcal{P}_{T}\mathcal{A}_{\Omega}-\mathcal{P}_{T}\mathcal{A}\right)\left(\boldsymbol{M}\right)\right\Vert _{\mathcal{A},2}+\sqrt{\frac{n_{1}n_{2}\log\left(n_{1}n_{2}\right)}{m}}\left\Vert \left(\frac{n_{1}n_{2}}{m}\mathcal{P}_{T}\mathcal{A}_{\Omega}-\mathcal{P}_{T}\mathcal{A}\right)\left(\boldsymbol{M}\right)\right\Vert _{\mathcal{A},\infty}\nonumber \\
%\leq & \text{ }c_{5}\left(\sqrt{\frac{\mu_{5}r\log\left(n_{1}n_{2}\right)}{m}}+\frac{\mu_{1}c_{\mathrm{s}}r\log\left(n_{1}n_{2}\right)}{m}\right)\left\{ \left\Vert \boldsymbol{M}\right\Vert _{\mathcal{A},2}+\sqrt{\frac{n_{1}n_{2}\log\left(n_{1}n_{2}\right)}{m}}\left\Vert \boldsymbol{M}\right\Vert _{\mathcal{A},\infty}\right\} \label{eq:NormA2Ainf_bound}
%\end{align}
% for some numerical constant $c_{5}>0$.\end{lemma}

%\begin{proof}See Appendix \ref{sec:Proof-of-Lemma:NormAinf_A2_bound}.\end{proof}

\subsection{An Upper Bound on $\left\Vert \mathcal{P}_{T^{\perp}}\left(\boldsymbol{W}\right)\right\Vert $
\label{sub:BoundW}}

Now we are ready to show how we may combine the above lemmas to develop
an upper bound on $\left\Vert \mathcal{P}_{T^{\perp}}\left(\boldsymbol{W}\right)\right\Vert $.
By construction, one has 
\[
\left\Vert \mathcal{P}_{T^{\perp}}\left(\boldsymbol{W}\right)\right\Vert \leq\sum_{l=1}^{j_{0}}\left\Vert \mathcal{P}_{T^{\perp}}\left(\frac{1}{q}\mathcal{A}_{\Omega_{l}}+\mathcal{A}^{\perp}\right)\mathcal{P}_{T}\left(\boldsymbol{F}_{l-1}\right)\right\Vert .
\]
Each summand can be bounded above as follows 
\begin{align}
 & \left\Vert \mathcal{P}_{T^{\perp}}\left(\frac{1}{q}\mathcal{A}_{\Omega_{l}}+\mathcal{A}^{\perp}\right)\mathcal{P}_{T}\left(\boldsymbol{F}_{l-1}\right)\right\Vert \nonumber \\
 & \quad=\left\Vert \mathcal{P}_{T^{\perp}}\left(\frac{1}{q}\mathcal{A}_{\Omega_{l}}-\mathcal{A}\right)\mathcal{P}_{T}\left(\boldsymbol{F}_{l-1}\right)\right\Vert \nonumber \\
 & \text{ }\text{ }\leq\left\Vert \left(\frac{1}{q}\mathcal{A}_{\Omega_{l}}-\mathcal{A}\right)\left(\boldsymbol{F}_{l-1}\right)\right\Vert \nonumber \\
 & \text{ }\text{ }\leq c_{2}\left(\sqrt{\frac{\log\left(n_{1}n_{2}\right)}{q}}\left\Vert \boldsymbol{F}_{l-1}\right\Vert _{\mathcal{A},2}+\frac{\log\left(n_{1}n_{2}\right)}{q}\left\Vert \boldsymbol{F}_{l-1}\right\Vert _{\mathcal{A},\infty}\right)\label{eq:2norm_Residual}\\
 & \text{ }\text{ }\leq\text{ }c_{2}c_{5}\left(\sqrt{\frac{\mu_{5}r\log\left(n_{1}n_{2}\right)}{qn_{1}n_{2}}}+\frac{\mu_{1}c_{\mathrm{s}}r\log\left(n_{1}n_{2}\right)}{qn_{1}n_{2}}\right)\nonumber \\
 & \quad\quad\quad\small\cdot\left\{ \sqrt{\frac{\log\left(n_{1}n_{2}\right)}{q}}\left\Vert \boldsymbol{F}_{l-2}\right\Vert _{\mathcal{A},2}+\frac{\log\left(n_{1}n_{2}\right)}{q}\left\Vert \boldsymbol{F}_{l-2}\right\Vert _{\mathcal{A},\infty}\right\} \label{eq:A2Ainf_residual}\\
 & \text{ }\text{ }\leq\small\left(\frac{1}{2}\right)^{l-1}\left(\sqrt{\frac{\log\left(n_{1}n_{2}\right)}{q}}\cdot\left\Vert \boldsymbol{F}_{0}\right\Vert _{\mathcal{A},2}+\frac{\log\left(n_{1}n_{2}\right)}{q}\left\Vert \boldsymbol{F}_{0}\right\Vert _{\mathcal{A},\infty}\right),\label{eq:PFi_bound}
\end{align}
where (\ref{eq:2norm_Residual}) follows from Lemma \ref{lemma:OpNorm_Anorm}
together with the fact that $\boldsymbol{F}_{i}\in T$, and (\ref{eq:A2Ainf_residual})
is a consequence of \eqref{eq:NormA2Ainf_bound}. The last inequality
holds under the hypothesis that $qn_{1}n_{2}\gg\max\left\{ \mu_{1}c_{\mathrm{s}},\mu_{5}\right\} r\log\left(n_{1}n_{2}\right)$
or, equivalently, $m\gg\max\left\{ \mu_{1}c_{\mathrm{s}},\mu_{5}\right\} r\log^{2}\left(n_{1}n_{2}\right)$.

Since $\boldsymbol{F}_{0}=\boldsymbol{U}\boldsymbol{V}^{*}$, it remains
to control $\left\Vert \boldsymbol{U}\boldsymbol{V}^{*}\right\Vert _{\mathcal{A},\infty}$
and $\left\Vert \boldsymbol{U}\boldsymbol{V}^{*}\right\Vert _{\mathcal{A},2}$.
We have the following lemma. %On the one hand, which follows from Lemma \ref{lemma-UpperBoundAbPtAa} supplied in Section \ref{sec:Main-Proof-Robust}. 

%On the other hand, $\left\Vert \boldsymbol{U}\boldsymbol{V}^{*}\right\Vert _{\mathcal{A},2}$ and $\mu_{5}$ can be controlled through the following lemma.

\begin{lemma}\label{lemma:mu6}With the incoherence measure $\mu_{1}$,
one can bound 
\begin{align}
\left\Vert \boldsymbol{U}\boldsymbol{V}^{*}\right\Vert _{\mathcal{A},\infty} & \leq\frac{\mu_{1}c_{\mathrm{s}}r}{n_{1}n_{2}},\label{eq:BoundMuInf_mu1}
\end{align}
\begin{equation}
\left\Vert \boldsymbol{U}\boldsymbol{V}^{*}\right\Vert _{\mathcal{A},2}^{2}\leq\frac{\mu_{1}c_{\mathrm{s}}r\log^{2}\left(n_{1}n_{2}\right)}{n_{1}n_{2}},\label{eq:BoundA2_mu1}
\end{equation}
and for any $(\alpha,\beta)\in\left[n_{1}\right]\times\left[n_{2}\right]$,
\begin{equation}
\left\Vert \mathcal{P}_{T}\left(\sqrt{\omega_{\alpha,\beta}}\boldsymbol{A}_{(\alpha,\beta)}\right)\right\Vert _{\mathcal{A},2}^{2}\leq\frac{c_{6}\mu_{1}c_{\mathrm{s}}\log^{2}\left(n_{1}n_{2}\right)r}{n_{1}n_{2}}\label{eq:BoundMu5Mu1}
\end{equation}
for some numerical constant $c_{6}>0$.\end{lemma}

\begin{proof}See Appendix \ref{sec:Proof-of-Lemma:mu6}.\end{proof}

In particular, the bound (\ref{eq:BoundMu5Mu1}) translates into 
\[
\mu_{5}\leq c_{6}\mu_{1}c_{\mathrm{s}}\log^{2}\left(n_{1}n_{2}\right).
\]
Substituting (\ref{eq:BoundMuInf_mu1}) and (\ref{eq:BoundA2_mu1})
into (\ref{eq:PFi_bound}) gives 
\begin{align*}
 & \left\Vert \mathcal{P}_{T^{\perp}}\left(\frac{n_{1}n_{2}}{m}\mathcal{A}_{\Omega_{l}}+\mathcal{A}^{\perp}\right)\mathcal{P}_{T}\left(\boldsymbol{F}_{l-1}\right)\right\Vert \\
 & \quad\leq\left(\frac{1}{2}\right)^{l-1}\left(\sqrt{\frac{\mu_{1}c_{\mathrm{s}}r\log^{2}\left(n_{1}n_{2}\right)}{qn_{1}n_{2}}}+\frac{\mu_{1}c_{\mathrm{s}}r\log\left(n_{1}n_{2}\right)}{qn_{1}n_{2}}\right)\\
 & \quad\ll\frac{1}{2}\cdot\left(\frac{1}{2}\right)^{l},
\end{align*}
as soon as $m>c_{7}\max\left\{ \mu_{1}c_{\mathrm{s}}\log^{2}\left(n_{1}n_{2}\right),\mu_{5}\log^{2}\left(n_{1}n_{2}\right)\right\} r$
or $m>\tilde{c}_{7}\mu_{1}c_{\mathrm{s}}\log^{4}\left(n_{1}n_{2}\right)$
for some sufficiently large constants $c_{7},\tilde{c}_{7}>0$, indicating
that 
\begin{align}
\left\Vert \mathcal{P}_{T^{\perp}}\left(\boldsymbol{W}\right)\right\Vert  & \leq\sum_{l=1}^{j_{0}}\left\Vert \mathcal{P}_{T^{\perp}}\left(\frac{1}{q}\mathcal{A}_{\Omega_{l}}+\mathcal{A}^{\perp}\right)\mathcal{P}_{T}\left(\boldsymbol{F}_{l-1}\right)\right\Vert \nonumber \\
 & \leq\frac{1}{2}\cdot\sum_{l=1}^{\infty}\left(\frac{1}{2}\right)^{l}\leq\frac{1}{2}
\end{align}
as required. So far, we have successfully verified that with high
probability, $\boldsymbol{W}$ is a valid dual certificate, and hence
by Lemma \ref{lemma-Dual-Certificate} the solution to EMaC is exact
and unique.

\section{Proof of Theorem~\ref{theorem-EMaC-Robust}\label{sec:Main-Proof-Robust}}

The algorithm Robust-EMaC is inspired by the well-known robust principal
component analysis \cite{CanLiMaWri09,li2011compressed} that seeks
a decomposition of low-rank plus sparse matrices, except that we impose
multi-fold Hankel structures on both the low-rank and sparse matrices.
Following similar spirit as to the proof of Theorem \ref{theorem-EMaC-noiseless},
the proof here is based on duality analysis, and relies on the golfing
scheme \cite{Gross2011recovering} to construct a valid dual certificate.

In this section, we prove the results for a slightly different sampling
model as follows. 
\begin{itemize}
\item The location multi-set $\Omega^{\text{clean}}$ of observed uncorrupted
entries is generated by sampling $\left(1-\tau\right)\rho n_{1}n_{2}$
i.i.d. entries uniformly at random. 
\item The location multi-set $\Omega$ of observed entries is generated
by sampling $\rho n_{1}n_{2}$ i.i.d. entries uniformly at random,
with the first $\left(1-\tau\right)\rho n_{1}n_{2}$ entries coming
from $\Omega^{\text{clean}}$. 
\item The location set $\Omega^{\text{dirty}}$ of observed corrupted entries
is given by $\Omega'\backslash\Omega^{\text{clean}'}$, where $\Omega'$
and $\Omega^{\text{clean}'}$ denote the sets of distinct entry locations
in $\Omega$ and $\Omega^{\text{clean}}$, respectively. 
\end{itemize}
As mentioned in the proof of Theorem \ref{theorem-EMaC-noiseless},
this slightly different sampling model, while resulting in the same
order-wise bounds, significantly simplifies the analysis due to the
independence assumptions.

\begin{comment}
We will prove our results for the real case under the incoherence
measures ($\mu_{1}$, $\mu_{2}$, $\mu_{3}$, $\mu_{4}$). Moreover, 
\end{comment}

We will prove Theorem~\ref{theorem-EMaC-Robust} under an additional
random sign condition, that is, the signs of all non-zero entries
of $\boldsymbol{S}$ are \emph{independent} \emph{zero-mean} random
variables. Specifically, we will prove the following theorem.

\begin{theorem}[Random Sign]\label{theorem-EMaC-Robust-V2}Suppose
that $\boldsymbol{X}$ obeys the incoherence condition with parameter
$\mu_{1}$, and let $\lambda=\frac{1}{\sqrt{m\log\left(n_{1}n_{2}\right)}}$.
Assume that $\tau\leq0.2$ is some small positive constant, and that
the signs of nonzero entries of $\boldsymbol{S}$ are independently
generated with zero mean. If 
\[
m>c_{0}\mu_{1}^{2}c_{\mathrm{s}}^{2}r^{2}\log^{3}\left(n_{1}n_{2}\right),
\]
then Robust-EMaC succeeds in recovering $\boldsymbol{X}$ with probability
exceeding $1-(n_{1}n_{2})^{-2}$. \end{theorem}

In fact, a simple derandomization argument introduced in \cite[Section 2.2]{CanLiMaWri09}
immediately suggests that the performance of Robust-EMaC under the
fixed-sign pattern is no worse than that under the random-sign pattern
with sparsity parameter $2\tau$, i.e. the condition on the signs
pattern of $\boldsymbol{S}$ is unnecessary and Theorem \ref{theorem-EMaC-Robust}
follows after we establish Theorem~\ref{theorem-EMaC-Robust-V2}.
As a result, the section will focus on Theorem \ref{theorem-EMaC-Robust}
with random sign patterns, which are much easier to analyze.

\begin{comment}
Theorem \ref{theorem-EMaC-Robust-V2} taken collectively with Lemma
\ref{lemma-UpperBoundAbPtAa} immediately establishes Theorem \ref{theorem-EMaC-Robust}. 
\end{comment}

%\end{enumerate}
%\end{lemma}\begin{proof}See Appendix \ref{sec:Proof-of-lemma-UpperBoundAbPtAa}.\end{proof}

\subsection{Dual Certification}

We adopt similar notations as in Section \ref{sub:Duality-Noiseless}.
That said, if we generate $\rho n_{1}n_{2}$ i.i.d. entry locations
$\boldsymbol{a}_{i}$'s uniformly at random, and let the multi-sets
$\Omega$ and $\Omega^{\text{clean}}$ contain respectively $\{\boldsymbol{a}_{i}|1\leq i\leq\rho n_{1}n_{2}\}$
and $\{\boldsymbol{a}_{i}|1\leq i\leq\rho(1-\tau)n_{1}n_{2}\}$),
then 
\[
\mathcal{A}_{\Omega}:=\sum_{i=1}^{\rho n_{1}n_{2}}\mathcal{A}_{\boldsymbol{a}_{i}},\quad\text{and}\quad\mathcal{A}_{\Omega^{\text{clean}}}:=\sum_{i=1}^{\rho\left(1-\tau\right)n_{1}n_{2}}\mathcal{A}_{\boldsymbol{a}_{i}},
\]
corresponding to sampling with replacement. Besides, $\mathcal{A}'_{\Omega}$
(resp. $\mathcal{A}'_{\Omega^{\text{clean}}}$) is defined similar
to $\mathcal{A}_{\Omega}$ (resp. $\mathcal{A}_{\Omega^{\text{clean}}}$),
but with the sum extending only over \emph{distinct} samples.

We will establish that exact recovery can be guaranteed, if we can
produce a valid dual certificate as follows.

\begin{lemma}\label{lemma-Dual-Robust}Suppose that $\tau$ is some
small positive constant. Suppose that the associated sampling operator
$\mathcal{A}_{\Omega^{\mathrm{clean}}}$ obeys 
\begin{equation}
\left\Vert \mathcal{P}_{T}\mathcal{A}\mathcal{P}_{T}-\frac{1}{\rho\left(1-\tau\right)}\mathcal{P}_{T}\mathcal{A}_{\Omega^{\mathrm{clean}}}\mathcal{P}_{T}\right\Vert \leq\frac{1}{2},\label{eq:InvertibilityPtAcleanPt}
\end{equation}
and 
\begin{equation}
\left\Vert \mathcal{A}_{\Omega^{\mathrm{clean}}}\left(\boldsymbol{M}\right)\right\Vert _{\mathrm{F}}\leq10\log\left(n_{1}n_{2}\right)\left\Vert \mathcal{A}'_{\Omega^{\mathrm{clean}}}\left(\boldsymbol{M}\right)\right\Vert _{\mathrm{F}},\label{eq:ConnectionSamplingReplacement}
\end{equation}
for any matrix $\boldsymbol{M}$. If there exist a regularization
parameter $\lambda$ $\left(0<\lambda<1\right)$ and a matrix $\boldsymbol{W}$
obeying 
\begin{equation}
\begin{cases}
\left\Vert \mathcal{P}_{T}\left(\boldsymbol{W}+\lambda\mathrm{sgn}\left(\boldsymbol{S}_{\mathrm{e}}\right)-\boldsymbol{U}\boldsymbol{V}^{*}\right)\right\Vert _{\mathrm{F}}\leq\frac{\lambda}{n_{1}^{2}n_{2}^{2}},\\
\left\Vert \mathcal{P}_{T^{\perp}}\left(\boldsymbol{W}+\lambda\mathrm{sgn}\left(\boldsymbol{S}_{\mathrm{e}}\right)\right)\right\Vert \leq\frac{1}{4},\\
\mathcal{A}'_{\left(\Omega^{\mathrm{clean}}\right)^{\perp}}\left(\boldsymbol{W}\right)=0,\\
\left\Vert \mathcal{A}'_{\Omega^{\mathrm{clean}}}\left(\boldsymbol{W}\right)\right\Vert _{\infty}\leq\frac{\lambda}{4},
\end{cases}\label{eq:DualProperties-Robust}
\end{equation}
then Robust-EMaC is exact, i.e. the minimizer $\left(\hat{\boldsymbol{M}},\hat{\boldsymbol{S}}\right)$
satisfies $\hat{\boldsymbol{M}}=\boldsymbol{X}$.\end{lemma}

\begin{proof} See Appendix \ref{sec:Proof-of-Lemma-Dual-Robust}.\end{proof}

We note that a reasonably tight bound on $\left\Vert \mathcal{P}_{T}\mathcal{A}\mathcal{P}_{T}-\frac{1}{\rho\left(1-\tau\right)}\mathcal{P}_{T}\mathcal{A}_{\Omega^{\text{clean}}}\mathcal{P}_{T}\right\Vert $
has been developed by Lemma \ref{lemma-Invertibility-PtWPt}. Specifically,
there exists some constant $c_{1}>0$ such that if $\rho\left(1-\tau\right)n_{1}n_{2}>c_{1}\mu_{1}c_{\text{s}}r\log\left(n_{1}n_{2}\right)$,
then one has 
\[
\left\Vert \mathcal{P}_{T}\mathcal{A}\mathcal{P}_{T}-\frac{1}{\rho\left(1-\tau\right)}\mathcal{P}_{T}\mathcal{A}_{\Omega^{\text{clean}}}\mathcal{P}_{T}\right\Vert \leq\frac{1}{2}
\]
with probability exceeding $1-\left(n_{1}n_{2}\right)^{-4}$. Besides,
Chernoff bound \cite{Alon2008} indicates that with probability exceeding
$1-\left(n_{1}n_{2}\right)^{-3}$, none of the entries is sampled
more than $10\log\left(n_{1}n_{2}\right)$ times. Equivalently, 
\begin{align*}
 & \mathbb{P}\left(\forall\boldsymbol{M}:\left\Vert \mathcal{A}_{\Omega^{\text{clean}}}\left(\boldsymbol{M}\right)\right\Vert _{\mathrm{F}}\leq10\log\left(n_{1}n_{2}\right)\left\Vert \mathcal{A}'_{\Omega^{\text{clean}}}\left(\boldsymbol{M}\right)\right\Vert _{\mathrm{F}}\right)\\
 & \quad\geq1-(n_{1}n_{2})^{-3}.
\end{align*}
Our objective in the remainder of this section is to produce a dual
matrix $\boldsymbol{W}$ satisfying Condition (\ref{eq:DualProperties-Robust}).

\subsection{Construction of Dual Certificate}

Suppose that we generate $j_{0}$ independent random location multi-sets
$\Omega_{j}^{\text{clean}}$, where $\Omega_{j}^{\text{clean}}$ contains
$qn_{1}n_{2}$ i.i.d. samples uniformly at random. Here, we set $q:=\frac{\left(1-\tau\right)\rho}{j_{0}}$
and $\epsilon<\frac{1}{e}$. This way the distribution of the multi-set
$\Omega$ is the same as $\Omega_{1}^{\text{clean}}\cup\Omega_{2}^{\text{clean}}\cup\cdots\cup\Omega_{j_{0}}^{\text{clean}}$.

We now propose constructing a dual certificate $\boldsymbol{W}$ as
follows:

\vspace{10pt}

\begin{tabular}{>{\raggedright}p{0.45\textwidth}}
\hline 
\textbf{Construction of a dual certificate $\boldsymbol{W}$ via the
golfing scheme.}\tabularnewline
\hline 
\noalign{\vskip\doublerulesep} $\quad\quad$1. Set $\boldsymbol{F}_{0}=\mathcal{P}_{T}\left(\boldsymbol{U}\boldsymbol{V}^{*}-\lambda\mbox{sgn}\left(\boldsymbol{S}_{\text{e}}\right)\right)$,
and $j_{0}:=5\log_{\frac{1}{\epsilon}}n_{1}n_{2}$.\tabularnewline
$\quad\quad$2. For every $i$ ($1\leq i\leq j_{0}$), let $\boldsymbol{F}_{i}:=\mathcal{P}_{T}\left(\mathcal{A}-\frac{1}{q}\mathcal{A}_{\Omega_{i}^{\text{clean}}}\right)\mathcal{P}_{T}\left(\boldsymbol{F}_{i-1}\right).$\tabularnewline
$\quad\quad$3. Set $\boldsymbol{W}:=\sum_{j=1}^{j_{0}}\left(\frac{1}{q}\mathcal{A}_{\Omega_{j}^{\text{clean}}}+\mathcal{A}^{\perp}\right)\left(\boldsymbol{F}_{j-1}\right)$.\tabularnewline
\hline 
\end{tabular}

\vspace{10pt}

Take $\lambda=\frac{1}{\sqrt{m\log\left(n_{1}n_{2}\right)}}$. Note
that the construction of $\boldsymbol{W}$ proceeds with a similar
procedure as in Section \ref{sub:Dual-Certificate-Noiseless}, except
that $\boldsymbol{F}_{0}$ and $\Omega_{i}$ are replaced by $\mathcal{P}_{T}\left(\boldsymbol{U}\boldsymbol{V}^{*}-\lambda\mbox{sgn}\left(\boldsymbol{S}_{\text{e}}\right)\right)$
and $\Omega_{i}^{\text{clean}}$, respectively.

We will justify that $\boldsymbol{W}$ is a valid dual certificate,
by examining the conditions in \eqref{eq:DualProperties-Robust} step
by step.

(1) The first condition requires the term $\left\Vert \mathcal{P}_{T}\left(\boldsymbol{W}+\lambda\mbox{sgn}\left(\boldsymbol{S}_{\text{e}}\right)-\boldsymbol{U}\boldsymbol{V}^{*}\right)\right\Vert _{\text{F}}=\left\Vert \mathcal{P}_{T}\left(\boldsymbol{W}-\boldsymbol{F}_{0}\right)\right\Vert _{\text{F}}$
to be reasonably small. Lemma \ref{lemma-Invertibility-PtWPt} asserts
that there exist some constants $c_{1},\tilde{c}_{1}>0$ such that
if $m=\rho n_{1}n_{2}>c_{1}\mu_{1}c_{\text{s}}r\log^{2}\left(n_{1}n_{2}\right)$
or, equivalently, $q_{i}n_{1}n_{2}>\tilde{c}_{1}\mu_{1}c_{\text{s}}r\log^{2}\left(n_{1}n_{2}\right)$,
then 
\begin{align}
\left\Vert \mathcal{P}_{T}\left(\boldsymbol{F}_{j_{0}}\right)\right\Vert _{\text{F}} & \leq\epsilon\left\Vert \mathcal{P}_{T}\left(\boldsymbol{F}_{j_{0}-1}\right)\right\Vert _{\text{F}}\leq\cdots\leq\epsilon^{j_{0}}\left\Vert \mathcal{P}_{T}\left(\boldsymbol{F}_{0}\right)\right\Vert _{\text{F}}\nonumber \\
 & \leq\frac{1}{n_{1}^{5}n_{2}^{5}}\left(\left\Vert \boldsymbol{U}\boldsymbol{V}^{*}\right\Vert _{\text{F}}+\lambda\left\Vert \mbox{sgn}\left(\boldsymbol{S}_{\text{e}}\right)\right\Vert _{\text{F}}\right)\nonumber \\
 & \leq\frac{1}{n_{1}^{5}n_{2}^{5}}\left(\sqrt{r}+\lambda n_{1}n_{2}\right)<\frac{1}{n_{1}^{5}n_{2}^{5}}\left(n_{1}n_{2}+\lambda n_{1}n_{2}\right)\nonumber \\
 & \leq\frac{\lambda}{n_{1}^{2}n_{2}^{2}}\label{eq:PtFbound-Robust}
\end{align}
with probability exceeding $1-(n_{1}n_{2})^{-3}$. Apply the same
argument as for (\ref{eq:PtWUV_bound}) to derive 
\begin{align*}
-\mathcal{P}_{T}\left(\boldsymbol{W}-\boldsymbol{F}_{0}\right) & =\mathcal{P}_{T}\left(\boldsymbol{F}_{j_{0}}\right).
\end{align*}
Plugging this into (\ref{eq:PtFbound-Robust}) establishes that 
\begin{align}
\left\Vert \mathcal{P}_{T}\left(\boldsymbol{W}+\lambda\mbox{sgn}\left(\boldsymbol{S}_{\text{e}}\right)-\boldsymbol{U}\boldsymbol{V}^{*}\right)\right\Vert _{\text{F}} & =\left\Vert \mathcal{P}_{T}\left(\boldsymbol{F}_{j_{0}}\right)\right\Vert _{\text{F}}\nonumber \\
 & \leq\frac{\lambda}{n_{1}^{2}n_{2}^{2}}.
\end{align}

(2) The second condition relies on an upper bound on $\left\Vert \mathcal{P}_{T^{\perp}}\left(\boldsymbol{W}+\lambda\mbox{sgn}\left(\boldsymbol{S}_{\text{e}}\right)\right)\right\Vert $.
To this end, we proceed by controlling $\left\Vert \mathcal{P}_{T^{\perp}}\left(\boldsymbol{W}\right)\right\Vert $
and $\left\Vert \mathcal{P}_{T^{\perp}}\left(\lambda\mbox{sgn}\left(\boldsymbol{S}_{\text{e}}\right)\right)\right\Vert $
separately. Applying the same argument as for (\ref{eq:PFi_bound})
suggests 
\begin{align}
 & \left\Vert \mathcal{P}_{T^{\perp}}\left(\frac{1}{q}\mathcal{A}_{\Omega_{l}}+\mathcal{A}^{\perp}\right)\mathcal{P}_{T}\left(\boldsymbol{F}_{l-1}\right)\right\Vert \nonumber \\
 & \quad\leq\small\left(\frac{1}{2}\right)^{l-1}\left(\sqrt{\frac{\log\left(n_{1}n_{2}\right)}{q}}\cdot\left\Vert \boldsymbol{F}_{0}\right\Vert _{\mathcal{A},2}+\frac{\log\left(n_{1}n_{2}\right)}{q}\left\Vert \boldsymbol{F}_{0}\right\Vert _{\mathcal{A},\infty}\right)\nonumber \\
 & \quad\leq\small\left(\frac{1}{2}\right)^{l-1}\left(\sqrt{\frac{n_{1}n_{2}\log\left(n_{1}n_{2}\right)}{q}}+\frac{\log\left(n_{1}n_{2}\right)}{q}\right)\cdot\left\Vert \boldsymbol{F}_{0}\right\Vert _{\mathcal{A},\infty}\nonumber \\
 & \quad\leq\left(\frac{1}{2}\right)^{l-2}\frac{n_{1}n_{2}\log\left(n_{1}n_{2}\right)}{\sqrt{m}}\left\Vert \boldsymbol{F}_{0}\right\Vert _{\mathcal{A},\infty},\label{eq:UBPtperp_F}
\end{align}
where the second inequality follows since $\left\Vert \boldsymbol{M}\right\Vert _{\mathcal{A},2}\leq\sqrt{n_{1}n_{2}}\left\Vert \boldsymbol{M}\right\Vert _{\mathcal{A},\infty}$,
and the last inequality arises from the fact that 
\[
\frac{\log\left(n_{1}n_{2}\right)}{q}\leq\sqrt{\frac{n_{1}n_{2}\log\left(n_{1}n_{2}\right)}{q}}=\frac{n_{1}n_{2}\log\left(n_{1}n_{2}\right)}{\sqrt{m}}
\]
when $m\gg\log^{2}\left(n_{1}n_{2}\right)$. Note that $\boldsymbol{F}_{0}=\boldsymbol{U}\boldsymbol{V}^{*}-\lambda\mathcal{P}_{T}\left(\text{sgn}\left(\boldsymbol{S}_{\mathrm{e}}\right)\right)$.
Since we have established an upper bound on $\left\Vert \boldsymbol{U}\boldsymbol{V}^{*}\right\Vert _{\mathcal{A},\infty}$
in (\ref{eq:BoundMuInf_mu1}), what remains to be controlled is $\left\Vert \mathcal{P}_{T}\left(\mbox{sgn}\left(\boldsymbol{S}_{\text{e}}\right)\right)\right\Vert _{\mathcal{A},\infty}$.
This is achieved by the following lemma.

\begin{lemma}\label{lemma-Psparse}Suppose that $s$ is a positive
constant. then one has %If $m>c_{8}\mu_{1}^{2}c_{\mathrm{s}}^{2}r^{2}\log\left(n_{1}n_{2}\right)$, one has 
\[
\left\Vert \mathcal{P}_{T}\left(\mathrm{sgn}\left(\boldsymbol{S}_{\mathrm{e}}\right)\right)\right\Vert _{\mathcal{A},\infty}\leq c_{9}\frac{\mu_{1}c_{\mathrm{s}}r}{n_{1}n_{2}}\sqrt{m\tau\log\left(n_{1}n_{2}\right)}
\]
for some constant $c_{9}>0$ with probability at least $1-(n_{1}n_{2})^{-4}$.
\end{lemma} \begin{proof}See Appendix \ref{sec:Proof-of-Lemma-Psparse}.\end{proof}
%By setting $\lambda:=\frac{1}{\sqrt{m\log\left(n_{1}n_{2}\right)}}$,
%From , we have
%\[
%\left\Vert \lambda\mathcal{P}_{T}\left(\mbox{sgn}\left(\boldsymbol{S}_{\text{e}}\right)\right)\right\Vert _{\mathcal{A},\infty}\leq c_{9}\sqrt{s}\frac{\mu_{1}c_{\mathrm{s}}r}{n_{1}n_{2}}.
%\]
From (\ref{eq:BoundMuInf_mu1}) and Lemma~\ref{lemma-Psparse}, we
have 
\begin{align}
\left\Vert \boldsymbol{F}_{0}\right\Vert _{\mathcal{A},\infty} & \leq\left\Vert \boldsymbol{U}\boldsymbol{V}^{*}\right\Vert _{\mathcal{A},\infty}+\lambda\left\Vert \mathcal{P}_{T}\left(\mbox{sgn}\left(\boldsymbol{S}_{\text{e}}\right)\right)\right\Vert _{\mathcal{A},\infty}\nonumber \\
 & \leq\frac{\mu_{1}c_{\mathrm{s}}r}{n_{1}n_{2}}+\frac{c_{9}\mu_{1}c_{\mathrm{s}}r\sqrt{\tau}}{n_{1}n_{2}}\nonumber \\
 & \leq\frac{\tilde{c}_{9}\mu_{1}c_{\mathrm{s}}r}{n_{1}n_{2}},\label{eq:F0_Ainf}
\end{align}
and substitute \eqref{eq:F0_Ainf} into (\ref{eq:UBPtperp_F}) we
have 
\begin{align*}
 & \left\Vert \mathcal{P}_{T^{\perp}}\left(\frac{1}{q}\mathcal{A}_{\Omega_{l}}+\mathcal{A}^{\perp}\right)\mathcal{P}_{T}\left(\boldsymbol{F}_{l-1}\right)\right\Vert \\
 & \quad\leq\left(\frac{1}{2}\right)^{l-2}\frac{\tilde{c}_{9}\mu_{1}c_{\mathrm{s}}r\log\left(n_{1}n_{2}\right)}{\sqrt{m}}.
\end{align*}
In particular, if $m>c_{8}\mu_{1}^{2}c_{\mathrm{s}}^{2}r^{2}\log^{2}\left(n_{1}n_{2}\right)$
for some large enough constant $c_{8}$, then one has 
\[
\left\Vert \mathcal{P}_{T^{\perp}}\left(\frac{1}{q}\mathcal{A}_{\Omega_{l}}+\mathcal{A}^{\perp}\right)\mathcal{P}_{T}\left(\boldsymbol{F}_{l-1}\right)\right\Vert \leq\left(\frac{1}{2}\right)^{l+4}.
\]
As a result, we can obtain 
\begin{align}
\left\Vert \mathcal{P}_{T^{\perp}}\left(\boldsymbol{W}\right)\right\Vert  & \leq\sum_{i=1}^{j_{0}}\left\Vert \mathcal{P}_{T^{\perp}}\left(\frac{1}{q}\mathcal{A}_{\Omega_{i}^{\text{clean}}}+\mathcal{A}^{\perp}\right)\mathcal{P}_{T}\left(\boldsymbol{F}_{i-1}\right)\right\Vert \nonumber \\
 & \leq\sum_{i=0}^{j_{0}}\left(\frac{1}{2}\right)^{i+4}<\frac{1}{8}\label{eq:UBPtperpW_robust}
\end{align}
with probability exceeding $1-(n_{1}n_{2})^{-4}$.

It remains to control the term $\left\Vert \mathcal{P}_{T^{\perp}}\left(\lambda\mbox{sgn}\left(\boldsymbol{S}_{\text{e}}\right)\right)\right\Vert $,
which is supplied in the following lemma.

\begin{lemma}\label{lemma-PTperp_sparse}%Consider $\lambda=\frac{1}{\sqrt{m\log\left(n_{1}n_{2}\right)}},$
Suppose that $\tau$ is a small positive constant, then one has 
\begin{equation}
\begin{cases}
\left\Vert \mathrm{sgn}\left(\boldsymbol{S}_{\mathrm{e}}\right)\right\Vert \leq\sqrt{c_{10}\rho\tau n_{1}n_{2}\log^{\frac{1}{2}}\left(n_{1}n_{2}\right)},\\
\left\Vert \mathcal{P}_{T^{\perp}}\left(\frac{1}{q}\mathcal{A}_{\Omega_{l}}+\mathcal{A}^{\perp}\right)\mathcal{P}_{T}\left(\boldsymbol{F}_{l-1}\right)\right\Vert \leq\frac{1}{8},
\end{cases}\label{eq:LemmaPtperp_SgnSe}
\end{equation}
with probability at least $1-(n_{1}n_{2})^{-5}$. \end{lemma}

\begin{proof}See Appendix \ref{sec:Proof-of-Lemma-Pperp_sparse}.\end{proof}Putting
(\ref{eq:UBPtperpW_robust}) and (\ref{eq:LemmaPtperp_SgnSe}) together
yields 
\begin{align*}
\small\left\Vert \mathcal{P}_{T^{\perp}}\left(\boldsymbol{W}+\lambda\mbox{sgn}\left(\boldsymbol{S}_{\text{e}}\right)\right)\right\Vert  & \leq\left\Vert \mathcal{P}_{T^{\perp}}\left(\boldsymbol{W}\right)\right\Vert +\left\Vert \mathcal{P}_{T^{\perp}}\left(\lambda\mbox{sgn}\left(\boldsymbol{S}_{\text{e}}\right)\right)\right\Vert \\
 & \leq\frac{1}{4}
\end{align*}
with high probability.

(3) By construction, one has $\mathcal{A}'_{\left(\Omega^{\text{clean}}\right)^{\perp}}\left(\boldsymbol{W}\right)=0$.

(4) The last step is to bound $\left\Vert \mathcal{A}'_{\Omega^{\text{clean}}}\left(\boldsymbol{W}\right)\right\Vert _{\infty}$,
which is apparently bounded above by $\left\Vert \mathcal{A}_{\Omega^{\text{clean}}}\left(\boldsymbol{W}\right)\right\Vert _{\infty}$.
The construction procedure together with Lemma \ref{lemma:NormAinf_bound}
allows us to bound 
\begin{align*}
 & \left\Vert \boldsymbol{F}_{i}\right\Vert _{\mathcal{A},\infty}\leq c_{4}\left(\sqrt{\frac{\mu_{1}c_{\mathrm{s}}r\log\left(n_{1}n_{2}\right)}{qn_{1}n_{2}}}\cdot\sqrt{\frac{\mu_{1}c_{\mathrm{s}}r}{n_{1}n_{2}}}\left\Vert \boldsymbol{F}_{i-1}\right\Vert _{\mathcal{A},2}\right.\\
 & \quad\quad\quad\quad\quad\quad\left.+c_{4}\frac{\mu_{1}c_{\mathrm{s}}r\log\left(n_{1}n_{2}\right)}{qn_{1}n_{2}}\left\Vert \boldsymbol{F}_{i-1}\right\Vert _{\mathcal{A},\infty}\right)\\
 & \text{ }\footnotesize\leq c_{4}\left(\sqrt{\frac{\mu_{1}c_{\mathrm{s}}r\log\left(n_{1}n_{2}\right)}{qn_{1}n_{2}}}\sqrt{\mu_{1}c_{\mathrm{s}}r}+\frac{\mu_{1}c_{\mathrm{s}}r\log\left(n_{1}n_{2}\right)}{qn_{1}n_{2}}\right)\left\Vert \boldsymbol{F}_{i-1}\right\Vert _{\mathcal{A},\infty}\\
 & \text{ }\leq2c_{4}\mu_{1}c_{\mathrm{s}}r\sqrt{\frac{\log\left(n_{1}n_{2}\right)}{qn_{1}n_{2}}}\left\Vert \boldsymbol{F}_{i-1}\right\Vert _{\mathcal{A},\infty},
\end{align*}
where the second inequality arises since $\left\Vert \boldsymbol{F}_{i}\right\Vert _{\mathcal{A},2}\leq\sqrt{n_{1}n_{2}}\left\Vert \boldsymbol{F}_{i}\right\Vert _{\mathcal{A},\infty}$,
and the last step follows since $\sqrt{\frac{\log\left(n_{1}n_{2}\right)}{qn_{1}n_{2}}}\geq\frac{\log\left(n_{1}n_{2}\right)}{qn_{1}n_{2}}$
when $m\gg\log^{2}\left(n_{1}n_{2}\right)$. Then there exists some
constant $c_{11}>0$ such that if $m>c_{11}\mu_{1}^{2}c_{\mathrm{s}}^{2}r^{2}\log^{2}\left(n_{1}n_{2}\right)$,
then 
\[
\left\Vert \boldsymbol{F}_{i}\right\Vert _{\mathcal{A},\infty}\leq\frac{1}{4}\left\Vert \boldsymbol{F}_{i-1}\right\Vert _{\mathcal{A},\infty}\leq\frac{1}{4^{i}}\left\Vert \boldsymbol{F}_{0}\right\Vert _{\mathcal{A},\infty}\leq\frac{\tilde{c}_{9}\mu_{1}c_{\mathrm{s}}r}{4^{i}n_{1}n_{2}},
\]
where the last inequality follows from (\ref{eq:F0_Ainf}). As a result,
one can deduce 
\begin{align*}
 & \left\Vert \mathcal{A}_{\Omega^{\text{clean}}}\left(\boldsymbol{W}\right)\right\Vert _{\infty}=\left\Vert \sum_{i=1}^{j_{0}}\mathcal{A}_{\Omega^{\text{clean}}}\left(\frac{1}{q}\mathcal{A}_{\Omega_{i}^{\text{clean}}}+\mathcal{A}^{\perp}\right)\boldsymbol{F}_{i-1}\right\Vert _{\infty}\\
 & \quad=\left\Vert \sum_{i=1}^{j_{0}}\frac{1}{q}\mathcal{A}_{\Omega_{i}^{\text{clean}}}\boldsymbol{F}_{i-1}\right\Vert _{\infty}\\
 & \quad\leq\sum_{i=1}^{j_{0}}\frac{1}{q}\max_{(k,l)\in[n_{1}]\times[n_{2}]}\frac{\left|\left\langle \boldsymbol{A}_{(k,l)},\boldsymbol{F}_{i-1}\right\rangle \right|}{\sqrt{\omega_{k,l}}}=\sum_{i=1}^{j_{0}}\frac{1}{q}\left\Vert \boldsymbol{F}_{i-1}\right\Vert _{\mathcal{A},\infty}\\
 & \quad\leq\sum_{i=1}^{j_{0}}\frac{5\log\left(n_{1}n_{2}\right)}{\rho}\frac{\tilde{c}_{9}\mu_{1}c_{\mathrm{s}}r}{4^{i-1}n_{1}n_{2}}\\
 & \quad\leq\frac{20\log\left(n_{1}n_{2}\right)\tilde{c}_{9}\mu_{1}c_{\mathrm{s}}r}{3m}\leq\frac{1}{4\sqrt{m\log\left(n_{1}n_{2}\right)}},
\end{align*}
where the last inequality is obtained by setting $m>c_{12}\mu_{1}^{2}c_{\text{s}}^{2}r^{2}\log^{3}\left(n_{1}n_{2}\right)$
for some constant $c_{12}>0$.

To sum up, we have verified that $\boldsymbol{W}$ satisfies the four
conditions required in (\ref{eq:DualProperties-Robust}), and is hence
a valid dual certificate. This concludes the proof.

\section{Concluding Remarks\label{sec:Conclusions-and-Future}}

We present an efficient algorithm to estimate a spectrally sparse
signal from its partial time-domain samples that does not require
prior knowledge on the model order, which poses spectral compressed
sensing as a low-rank Hankel structured matrix completion problem.
Under mild incoherence conditions, our algorithm enables recovery
of the multi-dimensional unknown frequencies with infinite precision,
which remedies the basis mismatch issue that arises in conventional
CS paradigms. We have shown both theoretically and numerically that
our algorithm is stable against bounded noise and a constant proportion
of arbitrary corruptions, and can be extended numerically to tasks
such as super resolution. To the best of our knowledge, our result
on Hankel matrix completion is also the first theoretical guarantee
that is close to the information-theoretical limit (up to some logarithmic
factor).

Our results are based on uniform random observation models. In particular,
this paper considers directly taking a random subset of the time domain
samples, it is also possible to take a random set of linear mixtures
of the time domain samples, as in the renowned CS setting \cite{CandRomTao06}.
This again can be translated into taking linear measurements of the
low-rank $K$-fold Hankel matrix, given as $\boldsymbol{y}=\mathcal{B}(\boldsymbol{X}_{\text{e}})$.
Unfortunately, due to the Hankel structures, it is not clear whether
$\mathcal{B}$ exhibits approximate isometry property. Nonetheless,
the technique developed in this paper can be extended without difficulty
to analyze linear measurements, in a similar flavor of a golfing scheme
developed for CS in \cite{CandesPlan2011RIPless}.

It remains to be seen whether it is possible to obtain performance
guarantees of the proposed EMaC algorithm similar to that in \cite{CandesFernandez2012SR}
for super resolution. It is also of great interest to develop efficient
numerical methods to solve the EMaC algorithm in order to accommodate
large datasets.

\section{Acknowledgement}

This work was supported in part by the startup grant of The Ohio State
University to Y. Chi. The authors thank Mr. Yuanxin Li for preparing
Fig. 4 and Fig. 5.

\appendices{ }

\section{Bernstein Inequality}

Our analysis relies heavily on the Bernstein inequality. To simplify
presentation, we state below a user-friendly version of Bernstein
inequality, which is an immediate consequence of \cite[Theorem 1.6]{tropp2012user}.
\begin{lemma}\label{lemma:Bernstein}Consider $m$ independent random
matrices $\boldsymbol{M}_{l}$ ($1\leq l\leq m$) of dimension $d_{1}\times d_{2}$,
each satisfying $\mathbb{E}\left[\boldsymbol{M}_{l}\right]=0$ and
$\left\Vert \boldsymbol{M}_{l}\right\Vert \leq B$. Define 
\[
\sigma^{2}:=\max\left\{ \left\Vert \sum_{l=1}^{m}\mathbb{E}\left[\boldsymbol{M}_{l}\boldsymbol{M}_{l}^{*}\right]\right\Vert ,\left\Vert \sum_{l=1}^{m}\mathbb{E}\left[\boldsymbol{M}_{l}^{*}\boldsymbol{M}_{l}\right]\right\Vert \right\} .
\]
Then there exists a universal constant $c_{0}>0$ such that for any
integer $a\geq2$, 
\begin{equation}
\left\Vert \sum_{l=1}^{m}\boldsymbol{M}_{l}\right\Vert \leq c_{0}\left(\sqrt{a\sigma^{2}\log\left(d_{1}+d_{2}\right)}+aB\log\left(d_{1}+d_{2}\right)\right)\label{eq:Bernstein}
\end{equation}
with probability at least $1-(d_{1}+d_{2})^{-a}$.\end{lemma}

\section{Proof of Lemma \ref{lemma-Dual-Certificate}\label{sec:Proof-of-Lemma-Dual-Certificate}}

Consider any valid perturbation $\boldsymbol{H}$ obeying $\mathcal{P}_{\Omega}\left(\boldsymbol{X}+\boldsymbol{H}\right)=\mathcal{P}_{\Omega}\left(\boldsymbol{X}\right)$,
and denote by $\boldsymbol{H}_{\text{e}}$ the enhanced form of $\boldsymbol{H}$.
We note that the constraint requires $\mathcal{A}'_{\Omega}\left(\boldsymbol{H}_{\text{e}}\right)=0$
(or $\mathcal{A}{}_{\Omega}\left(\boldsymbol{H}_{\text{e}}\right)=0$)
and $\mathcal{A}^{\perp}\left(\boldsymbol{H}_{\text{e}}\right)=0$.
In addition, set $\boldsymbol{Z}_{0}=\mathcal{P}_{T^{\perp}}\left(\boldsymbol{B}\right)$
for any $\boldsymbol{B}$ that satisfies $\left\langle \boldsymbol{B},\mathcal{P}_{T^{\perp}}\left(\boldsymbol{H}_{\text{e}}\right)\right\rangle =\left\Vert \mathcal{P}_{T^{\perp}}\left(\boldsymbol{H}_{\text{e}}\right)\right\Vert _{*}$
and $\left\Vert \boldsymbol{B}\right\Vert \leq1$. Therefore, $\boldsymbol{Z}_{0}\in T^{\perp}$
and $\left\Vert \boldsymbol{Z}_{0}\right\Vert \leq1$, and hence $\boldsymbol{U}\boldsymbol{V}^{*}+\boldsymbol{Z}_{0}$
is a sub-gradient of the nuclear norm at $\boldsymbol{X}_{\text{e}}$.
We will establish this lemma by considering two scenarios separately.

(1) Consider first the case in which $\boldsymbol{H}_{\text{e}}$
satisfies 
\begin{equation}
\left\Vert \mathcal{P}_{T}\left(\boldsymbol{H}_{\text{e}}\right)\right\Vert _{\text{F}}\leq\frac{n_{1}^{2}n_{2}^{2}}{2}\left\Vert \mathcal{P}_{T^{\perp}}\left(\boldsymbol{H}_{\text{e}}\right)\right\Vert _{\text{F}}.\label{eq:SmallHeCase}
\end{equation}

Since $\boldsymbol{U}\boldsymbol{V}^{*}+\boldsymbol{Z}_{0}$ is a
sub-gradient of the nuclear norm at $\boldsymbol{X}_{\text{e}}$,
it follows that 
\begin{align}
 & \left\Vert \boldsymbol{X}_{\text{e}}+\boldsymbol{H}_{\text{e}}\right\Vert _{*}\nonumber \\
 & \quad\geq\left\Vert \boldsymbol{X}_{\text{e}}\right\Vert _{*}+\left\langle \boldsymbol{U}\boldsymbol{V}^{*}+\boldsymbol{Z}_{0},\boldsymbol{H}_{\text{e}}\right\rangle \nonumber \\
 & \quad=\left\Vert \boldsymbol{X}_{\text{e}}\right\Vert _{*}+\left\langle \boldsymbol{W},\boldsymbol{H}_{\text{e}}\right\rangle +\left\langle \boldsymbol{Z}_{0},\boldsymbol{H}_{\text{e}}\right\rangle -\left\langle \boldsymbol{W}-\boldsymbol{U}\boldsymbol{V}^{*},\boldsymbol{H}_{\text{e}}\right\rangle \nonumber \\
 & \quad=\left\Vert \boldsymbol{X}_{\text{e}}\right\Vert _{*}+\left\langle \left(\mathcal{A}'_{\Omega}+\mathcal{A}^{\perp}\right)\left(\boldsymbol{W}\right),\boldsymbol{H}_{\text{e}}\right\rangle \nonumber \\
 & \quad\quad\quad\quad\quad\quad+\left\langle \boldsymbol{Z}_{0},\boldsymbol{H}_{\text{e}}\right\rangle -\left\langle \boldsymbol{W}-\boldsymbol{U}\boldsymbol{V}^{*},\boldsymbol{H}_{\text{e}}\right\rangle \label{eq:DecompositionWHe1}\\
 & \quad\geq\left\Vert \boldsymbol{X}_{\text{e}}\right\Vert _{*}+\left\Vert \mathcal{P}_{T^{\perp}}\left(\boldsymbol{H}_{\text{e}}\right)\right\Vert _{*}-\left\langle \boldsymbol{W}-\boldsymbol{U}\boldsymbol{V}^{*},\boldsymbol{H}_{\text{e}}\right\rangle \label{eq:DecompositionWHe}
\end{align}
where \eqref{eq:DecompositionWHe1} holds from \eqref{eq:UV_W_Contained_in_AOmega},
and \eqref{eq:DecompositionWHe} follows from the property of $\boldsymbol{Z}_{0}$
and the fact that $\left(\mathcal{A}'_{\Omega}+\mathcal{A}^{\perp}\right)\left(\boldsymbol{H}_{\text{e}}\right)=0$.
The last term of (\ref{eq:DecompositionWHe}) can be bounded as 
\begin{align*}
 & \left\langle \boldsymbol{W}-\boldsymbol{U}\boldsymbol{V}^{*},\boldsymbol{H}_{\text{e}}\right\rangle \\
 & \quad=\left\langle \mathcal{P}_{T}\left(\boldsymbol{W}-\boldsymbol{U}\boldsymbol{V}^{*}\right),\boldsymbol{H}_{\text{e}}\right\rangle +\left\langle \mathcal{P}_{T^{\perp}}\left(\boldsymbol{W}-\boldsymbol{U}\boldsymbol{V}^{*}\right),\boldsymbol{H}_{\text{e}}\right\rangle \\
 & \quad\leq\left\Vert \mathcal{P}_{T}\left(\boldsymbol{W}-\boldsymbol{U}\boldsymbol{V}^{*}\right)\right\Vert _{\text{F}}\left\Vert \mathcal{P}_{T}\left(\boldsymbol{H}_{\text{e}}\right)\right\Vert _{\text{F}}\\
 & \quad\quad\quad\quad\quad\quad\quad+\left\Vert \mathcal{P}_{T^{\perp}}\left(\boldsymbol{W}\right)\right\Vert \left\Vert \mathcal{P}_{T^{\perp}}\left(\boldsymbol{H}_{\text{e}}\right)\right\Vert _{*}\\
 & \quad\leq\frac{1}{2n_{1}^{2}n_{2}^{2}}\left\Vert \mathcal{P}_{T}\left(\boldsymbol{H}_{\text{e}}\right)\right\Vert _{\text{F}}+\frac{1}{2}\left\Vert \mathcal{P}_{T^{\perp}}\left(\boldsymbol{H}_{\text{e}}\right)\right\Vert _{*},
\end{align*}
where the last inequality follows from the assumptions \eqref{eq:W_component_T}
and \eqref{eq:NormWTPerp}. Plugging this into (\ref{eq:DecompositionWHe})
yields 
\begin{align}
 & \left\Vert \boldsymbol{X}_{\text{e}}+\boldsymbol{H}_{\text{e}}\right\Vert _{*}\nonumber \\
 & \quad\geq\left\Vert \boldsymbol{X}_{\text{e}}\right\Vert _{*}-\frac{1}{2n_{1}^{2}n_{2}^{2}}\left\Vert \mathcal{P}_{T}\left(\boldsymbol{H}_{\text{e}}\right)\right\Vert _{\text{F}}+\frac{1}{2}\left\Vert \mathcal{P}_{T^{\perp}}\left(\boldsymbol{H}_{\text{e}}\right)\right\Vert _{*}\\
 & \quad\geq\left\Vert \boldsymbol{X}_{\text{e}}\right\Vert _{*}-\frac{1}{4}\left\Vert \mathcal{P}_{T^{\perp}}\left(\boldsymbol{H}_{\text{e}}\right)\right\Vert _{\text{F}}+\frac{1}{2}\left\Vert \mathcal{P}_{T^{\perp}}\left(\boldsymbol{H}_{\text{e}}\right)\right\Vert _{\text{F}}\label{eq:DualXeH}\\
 & \quad\geq\left\Vert \boldsymbol{X}_{\text{e}}\right\Vert _{*}+\frac{1}{4}\left\Vert \mathcal{P}_{T^{\perp}}\left(\boldsymbol{H}_{\text{e}}\right)\right\Vert _{\text{F}}\nonumber 
\end{align}
where \eqref{eq:DualXeH} follows from the inequality $\left\Vert \boldsymbol{M}\right\Vert _{*}\geq\left\Vert \boldsymbol{M}\right\Vert _{\text{F}}$
and (\ref{eq:SmallHeCase}). Therefore, $\boldsymbol{X}_{\text{e}}$
is the minimizer of EMaC.

We still need to prove the uniqueness of the minimizer. The inequality
(\ref{eq:DualXeH}) implies that $\left\Vert \boldsymbol{X}_{\text{e}}+\boldsymbol{H}_{\text{e}}\right\Vert _{*}=\left\Vert \boldsymbol{X}_{\text{e}}\right\Vert _{*}$
holds only when $\left\Vert \mathcal{P}_{T^{\perp}}\left(\boldsymbol{H}_{\text{e}}\right)\right\Vert _{\text{F}}=0$.
If $\left\Vert \mathcal{P}_{T^{\perp}}\left(\boldsymbol{H}_{\text{e}}\right)\right\Vert _{\text{F}}=0$,
then $\left\Vert \mathcal{P}_{T}\left(\boldsymbol{H}_{\text{e}}\right)\right\Vert _{\text{F}}\leq\frac{n_{1}^{2}n_{2}^{2}}{2}\left\Vert \mathcal{P}_{T^{\perp}}\left(\boldsymbol{H}_{\text{e}}\right)\right\Vert _{\text{F}}=0$,
and hence $\mathcal{P}_{T^{\perp}}\left(\boldsymbol{H}_{\text{e}}\right)=\mathcal{P}_{T}\left(\boldsymbol{H}_{\text{e}}\right)=0$,
which only occurs when $\boldsymbol{H}_{\text{e}}=0$. Hence, $\boldsymbol{X}_{\text{e}}$
is the unique minimizer in this situation.

(2) On the other hand, consider the complement scenario where the
following holds 
\begin{equation}
\left\Vert \mathcal{P}_{T}\left(\boldsymbol{H}_{\text{e}}\right)\right\Vert _{\text{F}}\geq\frac{n_{1}^{2}n_{2}^{2}}{2}\left\Vert \mathcal{P}_{T^{\perp}}\left(\boldsymbol{H}_{\text{e}}\right)\right\Vert _{\text{F}}.\label{eq:LargeHeCase-1}
\end{equation}
We would first like to bound $\left\Vert \left(\frac{n_{1}n_{2}}{m}\mathcal{A}_{\Omega}+\mathcal{A}^{\perp}\right)\mathcal{P}_{T}\left(\boldsymbol{H}_{\text{e}}\right)\right\Vert _{\text{F}}$
and $\left\Vert \left(\frac{n_{1}n_{2}}{m}\mathcal{A}_{\Omega}+\mathcal{A}^{\perp}\right)\mathcal{P}_{T^{\perp}}\left(\boldsymbol{H}_{\text{e}}\right)\right\Vert _{\text{F}}$.
The former term can be lower bounded by 
\begin{align}
 & \small\left\Vert \left(\frac{n_{1}n_{2}}{m}\mathcal{A}_{\Omega}+\mathcal{A}^{\perp}\right)\mathcal{P}_{T}\left(\boldsymbol{H}_{\text{e}}\right)\right\Vert _{\text{F}}^{2}\nonumber \\
 & \text{ }\text{ }=\small\left\langle \left(\frac{n_{1}n_{2}}{m}\mathcal{A}_{\Omega}+\mathcal{A}^{\perp}\right)\mathcal{P}_{T}\left(\boldsymbol{H}_{\text{e}}\right),\left(\frac{n_{1}n_{2}}{m}\mathcal{A}_{\Omega}+\mathcal{A}^{\perp}\right)\mathcal{P}_{T}\left(\boldsymbol{H}_{\text{e}}\right)\right\rangle \nonumber \\
 & \text{ }\text{ }=\small\left\langle \frac{n_{1}n_{2}}{m}\mathcal{A}_{\Omega}\mathcal{P}_{T}\left(\boldsymbol{H}_{\text{e}}\right),\frac{n_{1}n_{2}}{m}\mathcal{A}_{\Omega}\mathcal{P}_{T}\left(\boldsymbol{H}_{\text{e}}\right)\right\rangle \nonumber \\
 & \quad\quad\quad\quad\quad\small+\left\langle \mathcal{A}^{\perp}\mathcal{P}_{T}\left(\boldsymbol{H}_{\text{e}}\right),\mathcal{A}^{\perp}\mathcal{P}_{T}\left(\boldsymbol{H}_{\text{e}}\right)\right\rangle \nonumber \\
 & \text{ }\text{ }\geq\small\left\langle \mathcal{P}_{T}\left(\boldsymbol{H}_{\text{e}}\right),\frac{n_{1}n_{2}}{m}\mathcal{A}_{\Omega}\mathcal{P}_{T}\left(\boldsymbol{H}_{\text{e}}\right)\right\rangle +\left\langle \mathcal{P}_{T}\left(\boldsymbol{H}_{\text{e}}\right),\mathcal{A}^{\perp}\mathcal{P}_{T}\left(\boldsymbol{H}_{\text{e}}\right)\right\rangle \nonumber \\
 & \text{ }\text{ }=\small\left\langle \mathcal{P}_{T}\left(\boldsymbol{H}_{\text{e}}\right),\mathcal{P}_{T}\left(\frac{n_{1}n_{2}}{m}\mathcal{A}_{\Omega}+\mathcal{A}^{\perp}\right)\mathcal{P}_{T}\left(\boldsymbol{H}_{\text{e}}\right)\right\rangle \nonumber \\
 & \text{ }\text{ }=\small\left\langle \mathcal{P}_{T}\left(\boldsymbol{H}_{\text{e}}\right),\mathcal{P}_{T}\left(\boldsymbol{H}_{\text{e}}\right)\right\rangle \nonumber \\
 & \quad\quad\quad\small+\left\langle \mathcal{P}_{T}\left(\boldsymbol{H}_{\text{e}}\right),\left(\frac{n_{1}n_{2}}{m}\mathcal{P}_{T}\mathcal{A}_{\Omega}\mathcal{P}_{T}-\mathcal{P}_{T}\mathcal{A}\mathcal{P}_{T}\right)\mathcal{P}_{T}\left(\boldsymbol{H}_{\text{e}}\right)\right\rangle \nonumber \\
 & \text{ }\text{ }\geq\small\left\Vert \mathcal{P}_{T}\left(\boldsymbol{H}_{\text{e}}\right)\right\Vert _{\text{F}}^{2}-\left\Vert \mathcal{P}_{T}\mathcal{A}\mathcal{P}_{T}-\frac{n_{1}n_{2}}{m}\mathcal{P}_{T}\mathcal{A}_{\Omega}\mathcal{P}_{T}\right\Vert \left\Vert \mathcal{P}_{T}\left(\boldsymbol{H}_{\text{e}}\right)\right\Vert _{\text{F}}^{2}\nonumber \\
 & \text{ }\text{ }\geq\small\left(1-\left\Vert \mathcal{P}_{T}\mathcal{A}\mathcal{P}_{T}-\frac{n_{1}n_{2}}{m}\mathcal{P}_{T}\mathcal{A}_{\Omega}\mathcal{P}_{T}\right\Vert \right)\left\Vert \mathcal{P}_{T}\left(\boldsymbol{H}_{\text{e}}\right)\right\Vert _{\text{F}}^{2}\nonumber \\
 & \text{ }\text{ }\geq\text{ }\frac{1}{2}\left\Vert \mathcal{P}_{T}\left(\boldsymbol{H}_{\text{e}}\right)\right\Vert _{\text{F}}^{2}.\label{eq:LastStepLowerBoundPtPw}
\end{align}

On the other hand, since the operator norm of any projection operator
is bounded above by $1$, one can verify that 
\begin{align*}
\left\Vert \frac{n_{1}n_{2}}{m}\mathcal{A}_{\Omega}+\mathcal{A}^{\perp}\right\Vert  & \leq\frac{n_{1}n_{2}}{m}\left(\left\Vert \mathcal{A}_{a_{1}}+\mathcal{A}^{\perp}\right\Vert +\sum_{i=2}^{m}\left\Vert \mathcal{A}_{a_{i}}\right\Vert \right)\\
 & \leq n_{1}n_{2},
\end{align*}
where $a_{i}$ ($1\leq i\leq m$) are $m$ uniform random indices
that form $\Omega$. This implies the following bound: 
\begin{align*}
\left\Vert \left(\frac{n_{1}n_{2}}{m}\mathcal{A}_{\Omega}+\mathcal{A}^{\perp}\right)\mathcal{P}_{T^{\perp}}\left(\boldsymbol{H}_{\text{e}}\right)\right\Vert _{\text{F}} & \leq n_{1}n_{2}\left\Vert \mathcal{P}_{T^{\perp}}\left(\boldsymbol{H}_{\text{e}}\right)\right\Vert _{\text{F}}\\
 & \leq\frac{2}{n_{1}n_{2}}\left\Vert \mathcal{P}_{T}\left(\boldsymbol{H}_{\text{e}}\right)\right\Vert _{\text{F}},
\end{align*}
where the last inequality arises from our assumption. Combining this
with the above two bounds yields 
\begin{align*}
 & 0=\left\Vert \left(\frac{n_{1}n_{2}}{m}\mathcal{A}_{\Omega}+\mathcal{A}^{\perp}\right)\left(\boldsymbol{H}_{\text{e}}\right)\right\Vert _{\text{F}}\\
 & \quad\geq\left\Vert \left(\frac{n_{1}n_{2}}{m}\mathcal{A}_{\Omega}+\mathcal{A}^{\perp}\right)\mathcal{P}_{T}\left(\boldsymbol{H}_{\text{e}}\right)\right\Vert _{\text{F}}\\
 & \quad\quad\quad\quad-\left\Vert \left(\frac{n_{1}n_{2}}{m}\mathcal{A}_{\Omega}+\mathcal{A}^{\perp}\right)\mathcal{P}_{T^{\perp}}\left(\boldsymbol{H}_{\text{e}}\right)\right\Vert _{\text{F}}\\
 & \quad\geq\sqrt{\frac{1}{2}}\left\Vert \mathcal{P}_{T}\left(\boldsymbol{H}_{\text{e}}\right)\right\Vert _{\text{F}}-\frac{2}{n_{1}n_{2}}\left\Vert \mathcal{P}_{T}\left(\boldsymbol{H}_{\text{e}}\right)\right\Vert _{\text{F}}\\
 & \quad\geq\frac{1}{2}\left\Vert \mathcal{P}_{T}\left(\boldsymbol{H}_{\text{e}}\right)\right\Vert _{\text{F}}\geq\frac{n_{1}^{2}n_{2}^{2}}{4}\left\Vert \mathcal{P}_{T^{\perp}}\left(\boldsymbol{H}_{\text{e}}\right)\right\Vert _{\text{F}}\geq0,
\end{align*}
which immediately indicates $\mathcal{P}_{T^{\perp}}\left(\boldsymbol{H}_{\text{e}}\right)=0$
and $\mathcal{P}_{T}\left(\boldsymbol{H}_{\text{e}}\right)=0$. Hence,
(\ref{eq:LargeHeCase-1}) can only hold when $\boldsymbol{H}_{\text{e}}=0$.

\section{Proof of Lemma \ref{lemma-IncoherenceT_W}\label{sec:Proof-of-Lemma-lemma-IncoherenceT_W}}

Since $\boldsymbol{U}$ (resp. $\boldsymbol{V}$) and $\boldsymbol{E}_{\text{L}}$
(resp. $\boldsymbol{E}_{\text{R}}$) determine the same column (resp.
row) space, we can write 
\begin{align*}
\boldsymbol{U}\boldsymbol{U}^{*} & =\boldsymbol{E}_{\text{L}}\left(\boldsymbol{E}_{\text{L}}^{*}\boldsymbol{E}_{\text{L}}\right)^{-1}\boldsymbol{E}_{\text{L}}^{*},\\
\boldsymbol{V}\boldsymbol{V}^{*} & =\boldsymbol{E}_{\text{R}}^{*}\left(\boldsymbol{E}_{\text{R}}\boldsymbol{E}_{\text{R}}^{*}\right)^{-1}\boldsymbol{E}_{\text{R}},
\end{align*}
and thus 
\begin{align*}
\left\Vert \mathcal{P}_{U}\left(\boldsymbol{A}_{(k,l)}\right)\right\Vert _{\text{F}}^{2} & \leq\left\Vert \boldsymbol{E}_{\text{L}}\left(\boldsymbol{E}_{\text{L}}^{*}\boldsymbol{E}_{\text{L}}\right)^{-1}\boldsymbol{E}_{\text{L}}^{*}\boldsymbol{A}_{(k,l)}\right\Vert _{\text{F}}^{2}\\
 & \leq\frac{1}{\sigma_{\min}\left(\boldsymbol{E}_{\text{L}}^{*}\boldsymbol{E}_{\text{L}}\right)}\left\Vert \boldsymbol{E}_{\text{L}}^{*}\boldsymbol{A}_{(k,l)}\right\Vert _{\text{F}}^{2},\\
\text{and}\quad\left\Vert \mathcal{P}_{V}\left(\boldsymbol{A}_{(k,l)}\right)\right\Vert _{\text{F}}^{2} & \leq\left\Vert \boldsymbol{A}_{(k,l)}\boldsymbol{E}_{\text{R}}^{*}\left(\boldsymbol{E}_{\text{R}}\boldsymbol{E}_{\text{R}}^{*}\right)^{-1}\boldsymbol{E}_{\text{R}}\right\Vert _{\text{F}}^{2}\\
 & \leq\frac{1}{\sigma_{\min}\left(\boldsymbol{E}_{\text{R}}\boldsymbol{E}_{\text{R}}^{*}\right)}\left\Vert \boldsymbol{A}_{(k,l)}\boldsymbol{E}_{\text{R}}^{*}\right\Vert _{\text{F}}^{2}.
\end{align*}
Note that $\sqrt{\omega_{k,l}}\boldsymbol{E}_{\text{L}}^{*}\boldsymbol{A}_{(k,l)}$
consists of $\omega_{k,l}$ columns of $\boldsymbol{E}_{\text{L}}^{*}$
(and hence it contains $r\omega_{k,l}$ nonzero entries in total).
Owing to the fact that each entry of $\boldsymbol{E}_{\text{L}}^{*}$
has magnitude $\frac{1}{\sqrt{k_{1}k_{2}}}$, one can derive 
\[
\left\Vert \boldsymbol{E}_{\text{L}}^{*}\boldsymbol{A}_{(k,l)}\right\Vert _{\text{F}}^{2}=\frac{1}{\omega_{k,l}}\cdot r\omega_{k,l}\cdot\frac{1}{k_{1}k_{2}}=\frac{r}{k_{1}k_{2}}\leq\frac{rc_{\text{s}}}{n_{1}n_{2}}.
\]
A similar argument yields $\left\Vert \boldsymbol{A}_{(k,l)}\boldsymbol{E}_{\text{R}}^{*}\right\Vert _{\text{F}}^{2}\leq\frac{c_{\text{s}}r}{n_{1}n_{2}}$.
Combining $\sigma_{\min}\left(\boldsymbol{E}_{\text{L}}^{*}\boldsymbol{E}_{\text{L}}\right)\geq\frac{1}{\mu_{1}}$
and $\sigma_{\min}\left(\boldsymbol{E}_{\text{R}}\boldsymbol{E}_{\text{R}}^{*}\right)\geq\frac{1}{\mu_{1}}$,
\eqref{eq:IncoherenceUV_W} follows by plugging these facts into the
above equations. %Therefore, for every $(k,l)\in[n_{1}]\times[n_{2}]$,
%\begin{equation}
%\left\Vert \mathcal{P}_{U}\left(\boldsymbol{A}_{(k,l)}\right)\right\Vert _{\text{F}}^{2}\leq\frac{\mu_{1}c_{\text{s}}r}{n_{1}n_{2}},\quad\left\Vert \mathcal{P}_{V}\left(\boldsymbol{A}_{(k,l)}\right)\right\Vert _{\text{F}}^{2}\leq\frac{\mu_{1}c_{\text{s}}r}{n_{1}n_{2}},\quad\text{and}\quad\left\Vert \mathcal{P}_{T}\left(\boldsymbol{A}_{(k,l)}\right)\right\Vert _{\text{F}}^{2}\leq\frac{2\mu_{1}c_{\text{s}}r}{n_{1}n_{2}}.
%\end{equation}

To show \eqref{eq:UBAbPtAa}, since $\left|\left\langle \boldsymbol{A}_{\boldsymbol{b}},\mathcal{P}_{T}\left(\boldsymbol{A}_{\boldsymbol{a}}\right)\right\rangle \right|=\left|\left\langle \mathcal{P}_{T}\left(\boldsymbol{A}_{\boldsymbol{b}}\right),\boldsymbol{A}_{\boldsymbol{a}}\right\rangle \right|$,
we only need to examine the situation where $\omega_{\boldsymbol{b}}<\omega_{\boldsymbol{a}}$.
Observe that 
\begin{align*}
\left|\left\langle \boldsymbol{A}_{\boldsymbol{b}},\mathcal{P}_{T}\boldsymbol{A}_{\boldsymbol{a}}\right\rangle \right| & \leq\left|\left\langle \boldsymbol{A}_{\boldsymbol{b}},\boldsymbol{U}\boldsymbol{U}^{*}\boldsymbol{A}_{\boldsymbol{a}}\right\rangle \right|+\left|\left\langle \boldsymbol{A}_{\boldsymbol{b}},\boldsymbol{A}_{\boldsymbol{a}}\boldsymbol{V}\boldsymbol{V}^{*}\right\rangle \right|\\
 & \quad\quad\quad\quad+\left|\left\langle \boldsymbol{A}_{\boldsymbol{b}},\boldsymbol{U}\boldsymbol{U}^{*}\boldsymbol{A}_{\boldsymbol{a}}\boldsymbol{V}\boldsymbol{V}^{*}\right\rangle \right|.
\end{align*}
Owing to the multi-fold Hankel structure of $\boldsymbol{A}_{\boldsymbol{a}}$,
the matrix $\boldsymbol{U}\boldsymbol{U}^{*}\sqrt{\omega_{\boldsymbol{a}}}\boldsymbol{A}_{\boldsymbol{a}}$
consists of $\omega_{\boldsymbol{a}}$ columns of $\boldsymbol{U}\boldsymbol{U}^{*}$.
Since there are only $\omega_{\boldsymbol{b}}$ nonzero entries in
$\boldsymbol{A}_{\boldsymbol{b}}$ each of magnitude $\frac{1}{\sqrt{\omega_{\boldsymbol{b}}}}$,
we can derive 
\begin{align*}
\left|\left\langle \boldsymbol{A}_{\boldsymbol{b}},\boldsymbol{U}\boldsymbol{U}^{*}\boldsymbol{A}_{\boldsymbol{a}}\right\rangle \right| & \leq\left\Vert \boldsymbol{A}_{\boldsymbol{b}}\right\Vert _{1}\left\Vert \boldsymbol{U}\boldsymbol{U}^{*}\boldsymbol{A}_{\boldsymbol{a}}\right\Vert _{\infty}\\
 & =\omega_{\boldsymbol{b}}\cdot\frac{1}{\sqrt{\omega_{\boldsymbol{b}}}}\cdot\max_{\alpha,\beta}\left|\left(\boldsymbol{U}\boldsymbol{U}^{*}\boldsymbol{A}_{\boldsymbol{a}}\right)_{\alpha,\beta}\right|\\
 & \leq\sqrt{\frac{\omega_{\boldsymbol{b}}}{\omega_{\boldsymbol{a}}}}\max_{\alpha,\beta}\left|\left(\boldsymbol{U}\boldsymbol{U}^{*}\right)_{\alpha,\beta}\right|.
\end{align*}
%Denote by $\boldsymbol{M}_{*k}$ and $\boldsymbol{M}_{k*}$ the $k$th column and $k$th row of a matrix $\boldsymbol{M}$, respectively, then it can be observed that 
Each entry of $\boldsymbol{U}\boldsymbol{U}^{*}$ is bounded in magnitude
by 
\begin{align}
\left|\left(\boldsymbol{U}\boldsymbol{U}^{*}\right)_{k,l}\right| & =\left|\boldsymbol{e}_{k}^{\top}\boldsymbol{E}_{\text{L}}\left(\boldsymbol{E}_{\text{L}}^{*}\boldsymbol{E}_{\text{L}}\right)^{-1}\boldsymbol{E}_{\text{L}}^{*}\boldsymbol{e}_{l}\right|\nonumber \\
 & \leq\left\Vert \boldsymbol{e}_{k}^{\top}\boldsymbol{E}_{\text{L}}\right\Vert _{\text{F}}\left\Vert \left(\boldsymbol{E}_{\text{L}}^{*}\boldsymbol{E}_{\text{L}}\right)^{-1}\right\Vert \left\Vert \boldsymbol{E}_{\text{L}}^{*}\boldsymbol{e}_{l}\right\Vert _{\text{F}}\nonumber \\
 & \leq\frac{r}{k_{1}k_{2}}\frac{1}{\sigma_{\min}\left(\boldsymbol{E}_{\text{L}}^{*}\boldsymbol{E}_{\text{L}}\right)}\leq\frac{\mu_{1}c_{\text{s}}r}{n_{1}n_{2}},\label{eq:UmagnitudeBound}
\end{align}
which immediately implies that 
\begin{align}
\left|\left\langle \boldsymbol{A}_{\boldsymbol{b}},\boldsymbol{U}\boldsymbol{U}^{*}\boldsymbol{A}_{\boldsymbol{a}}\right\rangle \right| & \leq\sqrt{\frac{\omega_{\boldsymbol{b}}}{\omega_{\boldsymbol{a}}}}\frac{\mu_{1}c_{\text{s}}r}{n_{1}n_{2}}.\label{eq:UUAaAbbound}
\end{align}
Similarly, one can derive 
\begin{equation}
\left|\left\langle \boldsymbol{A}_{\boldsymbol{b}},\boldsymbol{A}_{\boldsymbol{a}}\boldsymbol{V}\boldsymbol{V}^{*}\right\rangle \right|\leq\sqrt{\frac{\omega_{\boldsymbol{b}}}{\omega_{\boldsymbol{a}}}}\frac{\mu_{1}c_{\text{s}}r}{n_{1}n_{2}}.\label{eq:VVAaAbbound}
\end{equation}

We still need to bound the magnitude of $\left\langle \boldsymbol{U}\boldsymbol{U}^{*}\boldsymbol{A}_{\boldsymbol{a}}\boldsymbol{V}\boldsymbol{V}^{*},\boldsymbol{A}_{\boldsymbol{b}}\right\rangle $.
One can observe that for the $k$th row of $\boldsymbol{U}\boldsymbol{U}^{*}$:
\begin{align*}
\left\Vert \boldsymbol{e}_{k}^{\top}\boldsymbol{U}\boldsymbol{U}^{*}\right\Vert _{\text{F}} & \leq\left\Vert \boldsymbol{e}_{k}^{\top}\boldsymbol{E}_{\text{L}}\left(\boldsymbol{E}_{\text{L}}^{*}\boldsymbol{E}_{\text{L}}\right)^{-1}\boldsymbol{E}_{\text{L}}^{*}\right\Vert _{\text{F}}\\
 & \leq\left\Vert \boldsymbol{e}_{k}^{\top}\boldsymbol{E}_{\text{L}}\right\Vert _{\text{F}}\left\Vert \left(\boldsymbol{E}_{\text{L}}^{*}\boldsymbol{E}_{\text{L}}\right)^{-1}\boldsymbol{E}_{\text{L}}^{*}\right\Vert \\
 & \leq\sqrt{\frac{\mu_{1}c_{\text{s}}r}{n_{1}n_{2}}}.
\end{align*}
Similarly, for the $l$th column of $\boldsymbol{V}\boldsymbol{V}^{*}$,
one has $\left\Vert \boldsymbol{V}\boldsymbol{V}^{*}\boldsymbol{e}_{l}\right\Vert _{\text{F}}\leq\sqrt{\frac{\mu_{1}c_{\text{s}}r}{n_{1}n_{2}}}$.
The magnitude of the entries of $\boldsymbol{U}\boldsymbol{U}^{*}\boldsymbol{A}_{\boldsymbol{a}}\boldsymbol{V}\boldsymbol{V}^{*}$
can now be bounded by 
\begin{align*}
\left|\left(\boldsymbol{U}\boldsymbol{U}^{*}\boldsymbol{A}_{\boldsymbol{a}}\boldsymbol{V}\boldsymbol{V}^{*}\right)_{k,l}\right| & \leq\left\Vert \boldsymbol{A}_{\boldsymbol{a}}\right\Vert \left\Vert \boldsymbol{e}_{k}^{\top}\boldsymbol{U}\boldsymbol{U}^{*}\right\Vert _{\text{F}}\left\Vert \boldsymbol{V}\boldsymbol{V}^{*}\boldsymbol{e}_{l}\right\Vert _{\text{F}}\\
 & \leq\frac{1}{\sqrt{\omega_{\boldsymbol{a}}}}\frac{\mu_{1}c_{\text{s}}r}{n_{1}n_{2}},
\end{align*}
where we used $\left\Vert \boldsymbol{A}_{\boldsymbol{a}}\right\Vert =1/\sqrt{\omega_{\boldsymbol{a}}}$.
Since $\boldsymbol{A}_{\boldsymbol{b}}$ has only $\omega_{\boldsymbol{b}}$
nonzero entries each has magnitude $\frac{1}{\sqrt{\omega_{\boldsymbol{b}}}}$,
one can verify that 
\begin{align}
\left|\left\langle \boldsymbol{U}\boldsymbol{U}^{*}\boldsymbol{A}_{\boldsymbol{a}}\boldsymbol{V}\boldsymbol{V}^{*},\boldsymbol{A}_{\boldsymbol{b}}\right\rangle \right| & \leq\left(\max_{k,l}\left|\left(\boldsymbol{U}\boldsymbol{U}^{*}\boldsymbol{A}_{\boldsymbol{a}}\boldsymbol{V}\boldsymbol{V}^{*}\right)_{k,l}\right|\right)\cdot\frac{\omega_{\boldsymbol{b}}}{\sqrt{\omega_{\boldsymbol{b}}}}\nonumber \\
 & =\sqrt{\frac{\omega_{\boldsymbol{b}}}{\omega_{\boldsymbol{a}}}}\frac{\mu_{1}c_{\text{s}}r}{n_{1}n_{2}}.\label{eq:UUAaVVbBound}
\end{align}
The above bounds (\ref{eq:UUAaAbbound}), (\ref{eq:VVAaAbbound})
and (\ref{eq:UUAaVVbBound}) taken together lead to \eqref{eq:UBAbPtAa}.
%\begin{align}
%\left|\left\langle \boldsymbol{A}_{\boldsymbol{b}},\mathcal{P}_{T}\boldsymbol{A}_{\boldsymbol{a}}\right\rangle \right| & \leq\left|\left\langle \boldsymbol{U}\boldsymbol{U}^{*}\boldsymbol{A}_{\boldsymbol{a}},\boldsymbol{A}_{\boldsymbol{b}}\right\rangle \right|+\left|\left\langle \boldsymbol{A}_{\boldsymbol{a}}\boldsymbol{V}\boldsymbol{V}^{*},\boldsymbol{A}_{\boldsymbol{b}}\right\rangle \right|+\left|\left\langle \boldsymbol{U}\boldsymbol{U}^{*}\boldsymbol{A}_{\boldsymbol{a}}\boldsymbol{V}\boldsymbol{V}^{*},\boldsymbol{A}_{\boldsymbol{b}}\right\rangle \right|\nonumber \\
% & \leq\sqrt{\frac{\omega_{\boldsymbol{b}}}{\omega_{\boldsymbol{a}}}}\frac{3\mu_{1}c_{\text{s}}r}{n_{1}n_{2}}.\label{eq:AbPtAaBoundKbKa}
%\end{align}

\section{Proof of Lemma \ref{lemma-Invertibility-PtWPt}\label{sec:Proof-of-Lemma-lemma-Invertibility-PtWPt}}

Define a family of operators 
\[
\mathcal{Z}_{(k,l)}:=\frac{n_{1}n_{2}}{m}\mathcal{P}_{T}\mathcal{A}_{(k,l)}\mathcal{P}_{T}-\frac{1}{m}\mathcal{P}_{T}\mathcal{A}\mathcal{P}_{T}.
\]
for any $(k,l)\in[n_{1}]\times[n_{2}]$. For any matrix $\boldsymbol{M}$,
we can compute 
\begin{align}
\mathcal{P}_{T}\mathcal{A}_{(k,l)}\mathcal{P}_{T}\left(\boldsymbol{M}\right) & =\mathcal{P}_{T}\left(\left\langle \boldsymbol{A}_{(k,l)},\mathcal{P}_{T}\boldsymbol{M}\right\rangle \boldsymbol{A}_{(k,l)}\right)\nonumber \\
 & =\mathcal{P}_{T}\left(\boldsymbol{A}_{(k,l)}\right)\left\langle \mathcal{P}_{T}\left(\boldsymbol{A}_{(k,l)}\right),\boldsymbol{M}\right\rangle ,\label{eq:ExpansionPtAPt}
\end{align}
and hence 
\begin{align*}
 & \left(\mathcal{P}_{T}\mathcal{A}_{(k,l)}\mathcal{P}_{T}\right)^{2}(\boldsymbol{M})\\
 & \quad=\left[\mathcal{P}_{T}\mathcal{A}_{(k,l)}\mathcal{P}_{T}\left(\boldsymbol{A}_{(k,l)}\right)\right]\left\langle \mathcal{P}_{T}\left(\boldsymbol{A}_{(k,l)}\right),\boldsymbol{M}\right\rangle \\
 & \quad=\left\langle \boldsymbol{A}_{(k,l)},\mathcal{P}_{T}\left(\boldsymbol{A}_{(k,l)}\right)\right\rangle \mathcal{P}_{T}\left(\boldsymbol{A}_{(k,l)}\right)\left\langle \mathcal{P}_{T}\left(\boldsymbol{A}_{(k,l)}\right),\boldsymbol{M}\right\rangle \\
 & \quad=\left\Vert \mathcal{P}_{T}\left(\boldsymbol{A}_{(k,l)}\right)\right\Vert _{\text{F}}^{2}\mathcal{P}_{T}\mathcal{A}_{(k,l)}\mathcal{P}_{T}\left(\boldsymbol{M}\right)\\
 & \quad\leq\frac{2\mu_{1}c_{\text{s}}r}{n_{1}n_{2}}\mathcal{P}_{T}\mathcal{A}_{(k,l)}\mathcal{P}_{T}\left(\boldsymbol{M}\right),
\end{align*}
where the last inequality follows from \eqref{eq:IncoherenceT_W}.
This further gives 
\begin{equation}
\left\Vert \mathcal{P}_{T}\mathcal{A}_{(k,l)}\mathcal{P}_{T}\right\Vert \leq\frac{2\mu_{1}c_{\text{s}}r}{n_{1}n_{2}}.\label{eq:ExpansionPtAPt_square}
\end{equation}
%Comparing (\ref{eq:ExpansionPtAPt}) and $ $(\ref{eq:ExpansionPtAPt_square}) gives 
%\begin{equation}
%\left(\mathcal{P}_{T}\mathcal{A}_{(k,l)}\mathcal{P}_{T}\right)^{2}=\left\langle \boldsymbol{A}_{(k,l)},\mathcal{P}_{T}\left(\boldsymbol{A}_{(k,l)}\right)\right\rangle \mathcal{P}_{T}\mathcal{A}_{(k,l)}\mathcal{P}_{T}\leq\frac{2\mu_{1}c_{\text{s}}r}{n_{1}n_{2}}\mathcal{P}_{T}\mathcal{A}_{(k,l)}\mathcal{P}_{T},\label{eq:BoundPtAPt_square}
%\end{equation}
%which follows from \eqref{eq:IncoherenceT_W}.
%\[
%\left\langle \boldsymbol{A}_{(k,l)},\mathcal{P}_{T}\left(\boldsymbol{A}_{(k,l)}\right)\right\rangle =\left\Vert \mathcal{P}_{T}\left(\boldsymbol{A}_{(k,l)}\right)\right\Vert _{\text{F}}^{2}\leq\frac{2\mu_{1}c_{\text{s}}r}{n_{1}n_{2}}.
%\]

Let $\boldsymbol{a}_{i}$ ($1\leq i\leq m$) be $m$ independent indices
uniformly drawn from $[n_{1}]\times[n_{2}]$, then we have $\mathbb{E}\left[\mathcal{Z}_{\boldsymbol{a}_{i}}\right]=0$
and 
\[
\left\Vert \mathcal{Z}_{\boldsymbol{a}_{i}}\right\Vert \leq2\max_{(k,l)\in[n_{1}]\times[n_{2}]}\frac{n_{1}n_{2}}{m}\left\Vert \mathcal{P}_{T}\boldsymbol{A}_{(k,l)}\mathcal{P}_{T}\right\Vert \leq\frac{4\mu_{1}c_{\text{s}}r}{m}.
\]
following from \eqref{eq:ExpansionPtAPt_square}. Further, 
\begin{align*}
\mathbb{E}\left[\mathcal{Z}_{\boldsymbol{a}_{i}}^{2}\right] & =\mathbb{E}\left(\frac{n_{1}n_{2}}{m}\mathcal{P}_{T}\mathcal{A}_{\boldsymbol{a}_{i}}\mathcal{P}_{T}\right)^{2}-\left(\mathbb{E}\left[\frac{n_{1}n_{2}}{m}\mathcal{P}_{T}\mathcal{A}_{\boldsymbol{a}_{i}}\mathcal{P}_{T}\right]\right)^{2}\\
 & =\frac{n_{1}^{2}n_{2}^{2}}{m^{2}}\mathbb{E}\left(\mathcal{P}_{T}\mathcal{A}_{\boldsymbol{a}_{i}}\mathcal{P}_{T}\right)^{2}-\frac{1}{m^{2}}\left(\mathcal{P}_{T}\mathcal{A}\mathcal{P}_{T}\right)^{2},
\end{align*}
We can then bound the operator norm as 
\begin{align}
\sum_{i=1}^{m}\left\Vert \mathbb{E}\left[\mathcal{Z}_{\boldsymbol{a}_{i}}^{2}\right]\right\Vert  & \leq\sum_{i=1}^{m}\frac{n_{1}^{2}n_{2}^{2}}{m^{2}}\left\Vert \mathbb{E}\left(\mathcal{P}_{T}\mathcal{A}_{\boldsymbol{a}_{i}}\mathcal{P}_{T}\right)^{2}\right\Vert \nonumber \\
 & \quad\quad\quad+\frac{1}{m}\left\Vert \left(\mathcal{P}_{T}\mathcal{A}\mathcal{P}_{T}\right)^{2}\right\Vert \nonumber \\
 & \leq\frac{n_{1}^{2}n_{2}^{2}}{m}\frac{2\mu_{1}c_{\text{s}}r}{n_{1}n_{2}}\left\Vert \mathbb{E}\left[\mathcal{P}_{T}\mathcal{A}_{\boldsymbol{a}_{i}}\mathcal{P}_{T}\right]\right\Vert +\frac{1}{m}\label{bound111}\\
 & =\frac{2\mu_{1}c_{\text{s}}rn_{1}n_{2}}{m}\frac{1}{n_{1}n_{2}}\left\Vert \mathcal{P}_{T}\mathcal{A}\mathcal{P}_{T}\right\Vert +\frac{1}{m^{2}}\nonumber \\
 & \leq\frac{4\mu_{1}c_{\text{s}}r}{m},\label{eq:DefnV_PAP}
\end{align}
where \eqref{bound111} uses \eqref{eq:ExpansionPtAPt_square}. %Besides, the first equality of (\ref{eq:BoundPtAPt_square}) gives
%$\left\Vert \mathcal{P}_{T}\mathcal{A}_{(k,l)}\mathcal{P}_{T}\right\Vert ^{2}\leq\left\Vert \mathcal{P}_{T}\boldsymbol{A}_{(k,l)}\right\Vert _{\text{F}}^{2}\left\Vert \mathcal{P}_{T}\mathcal{A}_{(k,l)}\mathcal{P}_{T}\right\Vert $
%and hence $\left\Vert \mathcal{P}_{T}\mathcal{A}_{(k,l)}\mathcal{P}_{T}\right\Vert \leq\left\Vert \mathcal{P}_{T}\boldsymbol{A}_{(k,l)}\right\Vert _{\text{F}}^{2}$,
%which immediately yields 
%This together with (\ref{eq:DefnV_PAP}) gives 
%\[
%\frac{2mV_{0}}{\left\Vert \mathcal{Z}_{\boldsymbol{a}_{i}}\right\Vert }\geq2.
%\]
Applying Lemma~\ref{lemma:Bernstein} yields that there exists some
constant $0<\epsilon\leq\frac{1}{2}$ such that % the Operator Bernstein Inequality \cite[Theorem 6]{Gross2011recovering}yields that for any $t\leq2$, 
%we have 
\[
\left\Vert \sum_{i=1}^{m}\mathcal{Z}_{\boldsymbol{a}_{i}}\right\Vert \leq\epsilon
\]
with probability exceeding $1-\left(n_{1}n_{2}\right)^{-4}$, provided
that $m>c_{1}\mu_{1}c_{\text{s}}r\log\left(n_{1}n_{2}\right)$ for
some universal constant $c_{1}>0$.

%\[
%\mathbb{P}\left(\left\Vert \sum_{i=1}^{m}\mathcal{Z}_{\boldsymbol{a}_{i}}\right\Vert >t\right)\leq2n_{1}n_{2}\exp\left(-\frac{t^{2}}{16\frac{\mu_{1}c_{\text{s}}r}{m}}\right).
%\]
%Finally, one can observe that $\sum_{i=1}^{m}\mathcal{Z}_{\boldsymbol{a}_{i}}$
%is equivalent to $\frac{n_{1}n_{2}}{m}\mathcal{P}_{T}\mathcal{A}_{\Omega}\mathcal{P}_{T}-\mathcal{P}_{T}\mathcal{A}\mathcal{P}_{T}$
%in distribution, which completes the proof.

\section{Proof of Lemma \ref{lemma:OpNorm_Anorm}\label{sec:proof-lemma:OpNorm_Anorm}}

Suppose that $\mathcal{A}_{\Omega}=\sum_{i=1}^{m}\mathcal{A}_{\boldsymbol{a}_{i}}$,
where $\boldsymbol{a}_{i}$, $1\leq i\leq m$, are $m$ independent
indices drawn uniformly at random from $[n_{1}]\times[n_{2}]$. Define
\[
\boldsymbol{S}_{(k,l)}:=\frac{n_{1}n_{2}}{m}\mathcal{A}_{(k,l)}\left(\boldsymbol{M}\right)-\frac{1}{m}\mathcal{A}\left(\boldsymbol{M}\right),\quad\left(k,l\right)\in\left[n_{1}\right]\times\left[n_{2}\right],
\]
which obeys $\mathbb{E}\left[\boldsymbol{S}_{\boldsymbol{a}_{i}}\right]={\bf 0}$
and 
\[
\left(\frac{n_{1}n_{2}}{m}\mathcal{A}_{\Omega}-\mathcal{A}\right)\left(\boldsymbol{M}\right):=\sum_{i=1}^{m}\boldsymbol{S}_{\boldsymbol{a}_{i}}.
\]

In order to apply Lemma \ref{lemma:Bernstein}, one needs to bound
$\left\Vert \mathbb{E}\left[\sum_{i=1}^{m}\boldsymbol{S}_{\boldsymbol{a}_{i}}\boldsymbol{S}_{\boldsymbol{a}_{i}}^{*}\right]\right\Vert $
and $\left\Vert \boldsymbol{S}_{\boldsymbol{a}_{i}}\right\Vert $,
which we tackle separately in the sequel. Observe that 
\begin{align*}
\boldsymbol{0}\preceq\boldsymbol{S}_{\left(k,l\right)}\boldsymbol{S}_{\left(k,l\right)}^{*} & =\left(\frac{n_{1}n_{2}}{m}\mathcal{A}_{(k,l)}\left(\boldsymbol{M}\right)-\frac{1}{m}\mathcal{A}\left(\boldsymbol{M}\right)\right)\cdot\\
 & \quad\quad\left(\frac{n_{1}n_{2}}{m}\mathcal{A}_{(k,l)}\left(\boldsymbol{M}\right)-\frac{1}{m}\mathcal{A}\left(\boldsymbol{M}\right)\right)^{*}\\
 & \preceq\left(\frac{n_{1}n_{2}}{m}\right)^{2}\mathcal{A}_{(k,l)}\left(\boldsymbol{M}\right)\left(\mathcal{A}_{(k,l)}\left(\boldsymbol{M}\right)\right)^{*}\\
 & =\left(\frac{n_{1}n_{2}}{m}\right)^{2}\left|\left\langle \boldsymbol{A}_{(k,l)},\boldsymbol{M}\right\rangle \right|^{2}\boldsymbol{A}_{(k,l)}\cdot\boldsymbol{A}_{(k,l)}^{\top}\\
 & \preceq\left(\frac{n_{1}n_{2}}{m}\right)^{2}\frac{\left|\left\langle \boldsymbol{A}_{(k,l)},\boldsymbol{M}\right\rangle \right|^{2}}{\omega_{k,l}}\boldsymbol{I},
\end{align*}
where the first inequality follows since $\frac{1}{m}\sum_{k,l}\mathcal{A}_{(k,l)}\left(\boldsymbol{M}\right)=\frac{1}{m}\mathcal{A}\left(\boldsymbol{M}\right)$,
and the last inequality arises from the fact that all non-zero entries
of $\boldsymbol{A}_{(k,l)}\cdot\boldsymbol{A}_{(k,l)}^{\top}$ lie
on its diagonal and are bounded in magnitude by $\frac{1}{\omega_{k,l}}$.
This immediately suggests 
\begin{align}
 & \left\Vert \mathbb{E}\left[\sum_{i=1}^{m}\boldsymbol{S}_{\boldsymbol{a}_{i}}\boldsymbol{S}_{\boldsymbol{a}_{i}}^{*}\right]\right\Vert =\frac{m}{n_{1}n_{2}}\left\Vert \sum_{\left(k,l\right)\in[n_{1}]\times[n_{2}]}\boldsymbol{S}_{\left(k,l\right)}\boldsymbol{S}_{\left(k,l\right)}^{*}\right\Vert \nonumber \\
 & \quad\leq\frac{m}{n_{1}n_{2}}\left\Vert \left(\frac{n_{1}n_{2}}{m}\right)^{2}\left(\sum_{\left(k,l\right)\in[n_{1}]\times[n_{2}]}\frac{\left|\left\langle \boldsymbol{A}_{(k,l)},\boldsymbol{M}\right\rangle \right|^{2}}{\omega_{k,l}}\right)\boldsymbol{I}\right\Vert \nonumber \\
 & \quad=\frac{n_{1}n_{2}}{m}\left\Vert \boldsymbol{M}\right\Vert _{\mathcal{A},2}^{2},\label{eq:Var_S_A2}
\end{align}
where the last equality follows from the definition of $\left\Vert \boldsymbol{M}\right\Vert _{\mathcal{A},2}$.
Following the same argument, one can derive the same bound for $\left\Vert \mathbb{E}\left[\sum_{i=1}^{m}\boldsymbol{S}_{\boldsymbol{a}_{i}}^{*}\boldsymbol{S}_{\boldsymbol{a}_{i}}\right]\right\Vert $
as well.

On the other hand, the operator norm of each $\boldsymbol{S}_{(k,l)}$
can be bounded as follows 
\begin{align}
\left\Vert \boldsymbol{S}_{(k,l)}\right\Vert  & \leq\left\Vert \frac{n_{1}n_{2}}{m}\mathcal{A}_{(k,l)}\left(\boldsymbol{M}\right)\right\Vert +\left\Vert \frac{1}{m}\mathcal{A}\left(\boldsymbol{M}\right)\right\Vert \nonumber \\
 & \leq2\max_{(k,l)\in\left[n_{1}\right]\times\left[n_{2}\right]}\left\Vert \frac{n_{1}n_{2}}{m}\mathcal{A}_{(k,l)}\left(\boldsymbol{M}\right)\right\Vert \nonumber \\
 & =\frac{2n_{1}n_{2}}{m}\max_{(k,l)\in\left[n_{1}\right]\times\left[n_{2}\right]}\left\Vert \left\langle \boldsymbol{A}_{(k,l)},\boldsymbol{M}\right\rangle \boldsymbol{A}_{(k,l)}\right\Vert \label{eq:A_identity}\\
 & =\frac{2n_{1}n_{2}}{m}\max_{(k,l)\in\left[n_{1}\right]\times\left[n_{2}\right]}\left|\frac{\left\langle \boldsymbol{A}_{(k,l)},\boldsymbol{M}\right\rangle }{\sqrt{\omega_{k,l}}}\right|\nonumber \\
 & =\frac{2n_{1}n_{2}}{m}\left\Vert \boldsymbol{M}\right\Vert _{\mathcal{A},\infty},\nonumber 
\end{align}
where (\ref{eq:A_identity}) holds since $\left\Vert \boldsymbol{A}_{(k,l)}\right\Vert =\frac{1}{\sqrt{\omega_{k,l}}}$
and the last equality follows by applying the definition of $\left\Vert \cdot\right\Vert _{\mathcal{A},\infty}$.

Finally, we combine the above two bounds together with Bernstein inequality
(Lemma \ref{lemma:Bernstein}) to obtain 
\begin{align*}
\left\Vert \left(\frac{n_{1}n_{2}}{m}\mathcal{A}_{\Omega}-\mathcal{A}\right)\left(\boldsymbol{M}\right)\right\Vert \leq & c_{2}\sqrt{\frac{n_{1}n_{2}\log\left(n_{1}n_{2}\right)}{m}}\left\Vert \boldsymbol{M}\right\Vert _{\mathcal{A},2}\\
 & \quad+c_{2}\frac{2n_{1}n_{2}\log\left(n_{1}n_{2}\right)}{m}\left\Vert \boldsymbol{M}\right\Vert _{\mathcal{A},\infty}
\end{align*}
with high probability, where $c_{2}>0$ is some absolute constant.

\section{Proof of Lemma \ref{lemma:NormA2_bound}}

\label{sec:Proof-of-Lemma:NormA2_bound}

Write $\mathcal{A}_{\Omega}=\sum_{i=1}^{m}\mathcal{A}_{\boldsymbol{a}_{i}}$,
where $\boldsymbol{a}_{i}$ ($1\leq i\leq m$) are $m$ independent
indices uniformly drawn from $[n_{1}]\times[n_{2}]$. By the definition
of $\left\Vert \boldsymbol{M}\right\Vert _{\mathcal{A},2}$, we need
to examine the components 
\[
\frac{1}{\sqrt{\omega_{k,l}}}\left\langle \boldsymbol{A}_{\left(k,l\right)},\left(\frac{n_{1}n_{2}}{m}\mathcal{P}_{T}\mathcal{A}_{\Omega}-\mathcal{P}_{T}\mathcal{A}\right)\left(\boldsymbol{M}\right)\right\rangle 
\]
for all $(k,l)\in[n_{1}]\times[n_{2}]$.

Define a set of variables $z_{(\alpha,\beta)}$'s to be 
\begin{equation}
\small z_{(\alpha,\beta)}^{\left(k,l\right)}:=\frac{1}{\sqrt{\omega_{k,l}}}\left\langle \boldsymbol{A}_{\left(k,l\right)},\frac{n_{1}n_{2}}{m}\mathcal{P}_{T}\mathcal{A}_{(\alpha,\beta)}\left(\boldsymbol{M}\right)-\frac{1}{m}\mathcal{P}_{T}\mathcal{A}\left(\boldsymbol{M}\right)\right\rangle ,\label{eq:zalpha}
\end{equation}
thus resulting in 
\[
\frac{1}{\sqrt{\omega_{k,l}}}\left\langle \boldsymbol{A}_{\left(k,l\right)},\left(\frac{n_{1}n_{2}}{m}\mathcal{P}_{T}\mathcal{A}_{\Omega}-\mathcal{P}_{T}\mathcal{A}\right)\left(\boldsymbol{M}\right)\right\rangle :=\sum_{i=1}^{m}z_{\boldsymbol{a}_{i}}^{\left(k,l\right)}.
\]
The definition of $\left\Vert \boldsymbol{M}\right\Vert _{\mathcal{A},2}$
allows us to express 
\begin{equation}
\left\Vert \left(\frac{n_{1}n_{2}}{m}\mathcal{P}_{T}\mathcal{A}_{\Omega}-\mathcal{P}_{T}\mathcal{A}\right)\left(\boldsymbol{M}\right)\right\Vert _{\mathcal{A},2}=\left\Vert \sum_{i=1}^{m}\boldsymbol{z}_{\boldsymbol{a}_{i}}\right\Vert _{2},
\end{equation}
where $\boldsymbol{z}_{(\alpha,\beta)}$'s are defined to be $n_{1}n_{2}$-dimensional
vectors 
\[
\boldsymbol{z}_{\left(\alpha,\beta\right)}:=\left[z_{\left(\alpha,\beta\right)}^{\left(k,l\right)}\right]_{\left(k,l\right)\in\left[n_{1}\right]\times\left[n_{2}\right]},\quad\left(\alpha,\beta\right)\in\left[n_{1}\right]\times\left[n_{2}\right].
\]

For any random vector $\boldsymbol{v}\in\mathcal{V}$, one can easily
bound $\left\Vert \boldsymbol{v}-\mathbb{E}\boldsymbol{v}\right\Vert _{2}\leq2\sup_{\tilde{\boldsymbol{v}}\in\mathcal{V}}\left\Vert \tilde{\boldsymbol{v}}\right\Vert _{2}$.
Observing that $\mathbb{E}\left[\boldsymbol{z}_{(\alpha,\beta)}\right]={\bf 0}$,
we can bound 
\begin{align}
 & \left\Vert \boldsymbol{z}_{\left(\alpha,\beta\right)}\right\Vert _{2}\leq2\sqrt{\sum_{k,l}\frac{1}{\omega_{k,l}}\left|\left\langle \boldsymbol{A}_{\left(k,l\right)},\frac{2n_{1}n_{2}}{m}\mathcal{P}_{T}\mathcal{A}_{(\alpha,\beta)}\left(\boldsymbol{M}\right)\right\rangle \right|^{2}}\nonumber \\
 & \quad\small=\frac{2n_{1}n_{2}}{m}\sqrt{\sum_{k,l}\frac{1}{\omega_{k,l}}\left|\left\langle \boldsymbol{A}_{\left(k,l\right)},\mathcal{P}_{T}\left(\boldsymbol{A}_{(\alpha,\beta)}\right)\left\langle \boldsymbol{A}_{(\alpha,\beta)},\boldsymbol{M}\right\rangle \right\rangle \right|^{2}}\nonumber \\
 & \quad\small=\frac{2n_{1}n_{2}}{m}\frac{\left|\left\langle \boldsymbol{A}_{(\alpha,\beta)},\boldsymbol{M}\right\rangle \right|}{\sqrt{\omega_{\alpha,\beta}}}\sqrt{\sum_{k,l}\frac{\omega_{\alpha,\beta}\left|\left\langle \boldsymbol{A}_{\left(k,l\right)},\mathcal{P}_{T}\left(\boldsymbol{A}_{(\alpha,\beta)}\right)\right\rangle \right|^{2}}{\omega_{k,l}}}\nonumber \\
 & \quad\small\leq\frac{2n_{1}n_{2}}{m}\frac{\left|\left\langle \boldsymbol{A}_{(\alpha,\beta)},\boldsymbol{M}\right\rangle \right|}{\sqrt{\omega_{\alpha,\beta}}}\sqrt{\frac{\mu_{5}r}{n_{1}n_{2}}}\nonumber \\
 & \quad\small=2\sqrt{\frac{n_{1}n_{2}}{m}\cdot\frac{\mu_{5}r}{m}}\frac{\left|\left\langle \boldsymbol{A}_{(\alpha,\beta)},\boldsymbol{M}\right\rangle \right|}{\sqrt{\omega_{\alpha,\beta}}},\label{eq:Boundz}
\end{align}
where \eqref{eq:Boundz} follows from the definition of $\mu_{5}$
in \eqref{eq:HypothesisMu5}. Now it follows that %The first step is then to evaluate the $\ell_{2}$ norm of each summand
%$\boldsymbol{z}_{\boldsymbol{a}_{i}}$. For any random vector $\boldsymbol{v}\in\mathcal{V}$,
%one can easily bound
%\[
%\left\Vert \boldsymbol{v}-\mathbb{E}\boldsymbol{v}\right\Vert _{2}\leq2\sup_{\tilde{\boldsymbol{v}}\in\mathcal{V}}\left\Vert \tilde{\boldsymbol{v}}\right\Vert _{2}.
%\]
%Observing that $\mathbb{E}\left[\boldsymbol{z}_{\boldsymbol{a}_{i}}\right]={\bf 0}$ gives
\begin{align}
\left\Vert \boldsymbol{z}_{\boldsymbol{a}_{i}}\right\Vert _{2} & \leq\max_{\alpha,\beta}\left\Vert \boldsymbol{z}_{\left(\alpha,\beta\right)}\right\Vert _{2}\nonumber \\
 & \leq\max_{\alpha,\beta}2\sqrt{\frac{n_{1}n_{2}}{m}\cdot\frac{\mu_{5}r}{m}}\frac{\left|\left\langle \boldsymbol{A}_{(\alpha,\beta)},\boldsymbol{M}\right\rangle \right|}{\sqrt{\omega_{\alpha,\beta}}}\nonumber \\
 & \leq2\sqrt{\frac{n_{1}n_{2}}{m}\cdot\frac{\mu_{5}r}{m}}\left\Vert \boldsymbol{M}\right\Vert _{\mathcal{A},\infty},\label{eq:BoundMu5}
\end{align}
where \eqref{eq:BoundMu5} follows from \eqref{eq:DefnAinf}. On the
other hand, 
\begin{align*}
\left|\mathbb{E}\left[\sum_{i=1}^{m}\boldsymbol{z}_{\boldsymbol{a}_{i}}^{*}\boldsymbol{z}_{\boldsymbol{a}_{i}}\right]\right| & =\frac{m}{n_{1}n_{2}}\sum_{\alpha,\beta}\|\boldsymbol{z}_{\left(\alpha,\beta\right)}\|_{2}^{2}\\
 & \leq\frac{m}{n_{1}n_{2}}\sum_{\alpha,\beta}4\frac{n_{1}n_{2}}{m}\cdot\frac{\mu_{5}r}{m}\frac{\left|\left\langle \boldsymbol{A}_{(\alpha,\beta)},\boldsymbol{M}\right\rangle \right|^{2}}{\omega_{\alpha,\beta}}\\
 & =\frac{4\mu_{5}r}{m}\left\Vert \boldsymbol{M}\right\Vert _{\mathcal{A},2}^{2},
\end{align*}
which again follows from \eqref{eq:DefnA2}. Since $\boldsymbol{z}_{\boldsymbol{a}_{i}}$'s
are vectors, we immediately obtain $\left\Vert \mathbb{E}\left[\sum_{i=1}^{m}\boldsymbol{z}_{\boldsymbol{a}_{i}}\boldsymbol{z}_{\boldsymbol{a}_{i}}^{*}\right]\right\Vert =\left|\mathbb{E}\left[\sum_{i=1}^{m}\boldsymbol{z}_{\boldsymbol{a}_{i}}^{*}\boldsymbol{z}_{\boldsymbol{a}_{i}}\right]\right|$.
Applying Lemma \ref{lemma:Bernstein} then suggests that 
\begin{align*}
 & \left\Vert \left(\frac{n_{1}n_{2}}{m}\mathcal{P}_{T}\mathcal{A}_{\Omega}-\mathcal{P}_{T}\mathcal{A}\right)\left(\boldsymbol{M}\right)\right\Vert _{\mathcal{A},2}\leq c_{3}\sqrt{\frac{\mu_{5}r\log\left(n_{1}n_{2}\right)}{m}}\left\Vert \boldsymbol{M}\right\Vert _{\mathcal{A},2}\\
 & \quad\small\quad\quad\quad\quad\quad\quad+c_{3}\sqrt{\frac{n_{1}n_{2}}{m}\cdot\frac{\mu_{5}r}{m}}\log\left(n_{1}n_{2}\right)\left\Vert \boldsymbol{M}\right\Vert _{\mathcal{A},\infty}
\end{align*}
with high probability for some numerical constant $c_{3}>0$, which
completes the proof.

% \leq c_{3}\left(\sqrt{\left|\mathbb{E}\left[\sum_{i=1}^{m}\boldsymbol{z}_{\boldsymbol{a}_{i}}^{\top}\boldsymbol{z}_{\boldsymbol{a}_{i}}\right]\right|\log\left(n_{1}n_{2}\right)}+\max_{(\alpha,\beta)\in\left[n_{1}\right]\times\left[n_{2}\right]}\left\Vert \boldsymbol{z}_{\left(\alpha,\beta\right)}\right\Vert _{2}\log\left(n_{1}n_{2}\right)\right)\\

\section{Proof of Lemma \ref{lemma:NormAinf_bound}\label{sec:Proof-of-Lemma:NormAinf_bound}}

From Appendix \ref{sec:Proof-of-Lemma:NormA2_bound}, it is straightforward
that 
\begin{equation}
\left\Vert \left(\frac{n_{1}n_{2}}{m}\mathcal{P}_{T}\mathcal{A}_{\Omega}-\mathcal{P}_{T}\mathcal{A}\right)\left(\boldsymbol{M}\right)\right\Vert _{\mathcal{A},\infty}=\max_{k,l}\left|\sum_{i=1}^{m}z_{\boldsymbol{a}_{i}}^{\left(k,l\right)}\right|,
\end{equation}
where $z_{\boldsymbol{a}_{i}}^{\left(k,l\right)}$'s are defined as
\eqref{eq:zalpha}. Using similar techniques as \eqref{eq:Boundz},
we can obtain 
\begin{align*}
\left|z_{(\alpha,\beta)}^{(k,l)}\right| & \leq2\max_{k,l}\frac{\left|\left\langle \boldsymbol{A}_{\left(k,l\right)},\frac{n_{1}n_{2}}{m}\mathcal{P}_{T}\left(\boldsymbol{A}_{(\alpha,\beta)}\right)\left\langle \boldsymbol{A}_{(\alpha,\beta)},\boldsymbol{M}\right\rangle \right\rangle \right|}{\sqrt{\omega_{k,l}}}\\
 & \leq2\max_{k,l}\left(\frac{1}{\sqrt{\omega_{k,l}}}\sqrt{\frac{\omega_{k,l}}{\omega_{\alpha,\beta}}}\frac{3\mu_{1}c_{s}r}{n_{1}n_{2}}\right)\frac{n_{1}n_{2}}{m}\left|\left\langle \boldsymbol{A}_{(\alpha,\beta)},\boldsymbol{M}\right\rangle \right|\\
 & =\frac{6\mu_{1}c_{s}r}{m}\frac{1}{\sqrt{\omega_{\alpha,\beta}}}\left|\left\langle \boldsymbol{A}_{(\alpha,\beta)},\boldsymbol{M}\right\rangle \right|,
\end{align*}
where we have made use of the fact \eqref{eq:UBAbPtAa}. As a result,
one has 
\[
\left|z_{(\alpha,\beta)}^{(k,l)}\right|\leq\frac{6\mu_{1}c_{s}r}{m}\|\boldsymbol{M}\|_{\mathcal{A},\infty}
\]
and 
\begin{align*}
\mathbb{E}\left[\sum_{i=1}^{m}|z_{\boldsymbol{a}_{i}}^{\left(k,l\right)}|^{2}\right] & =\frac{m}{n_{1}n_{2}}\sum_{\alpha,\beta}\left|z_{(\alpha,\beta)}^{(k,l)}\right|^{2}\\
 & \leq\frac{m}{n_{1}n_{2}}\left(\frac{6\mu_{1}c_{s}r}{m}\right)^{2}\sum_{\alpha,\beta}\frac{1}{\omega_{\alpha,\beta}}\left|\left\langle \boldsymbol{A}_{(\alpha,\beta)},\boldsymbol{M}\right\rangle \right|^{2}\\
 & =\frac{36\mu_{1}^{2}c_{s}^{2}r^{2}}{mn_{1}n_{2}}\|\boldsymbol{M}\|_{\mathcal{A},2}^{2}.
\end{align*}

The Bernstein inequality in Lemma \ref{lemma:Bernstein} taken collectively
with the union bound yields that 
\begin{align*}
 & \left\Vert \left(\frac{n_{1}n_{2}}{m}\mathcal{P}_{T}\mathcal{A}_{\Omega}-\mathcal{P}_{T}\mathcal{A}\right)\left(\boldsymbol{M}\right)\right\Vert _{\mathcal{A},\infty}\\
 & \quad\leq c_{4}\sqrt{\frac{\mu_{1}c_{\mathrm{s}}r\log\left(n_{1}n_{2}\right)}{m}}\cdot\sqrt{\frac{\mu_{1}c_{\mathrm{s}}r}{n_{1}n_{2}}}\left\Vert \boldsymbol{M}\right\Vert _{\mathcal{A},2}\\
 & \quad\quad\quad\quad+c_{4}\frac{\mu_{1}c_{\mathrm{s}}r\log\left(n_{1}n_{2}\right)}{m}\left\Vert \boldsymbol{M}\right\Vert _{\mathcal{A},\infty}
\end{align*}
with high probability for some constant $c_{4}>0$, completing the
proof.

\section{Proof of Lemma \ref{lemma:mu6}\label{sec:Proof-of-Lemma:mu6}}

To bound $\left\Vert \boldsymbol{U}\boldsymbol{V}^{*}\right\Vert _{\mathcal{A},\infty}$,
observe that there exists a unitary matrix $\boldsymbol{B}$ such
that 
\[
\boldsymbol{U}\boldsymbol{V}^{*}=\boldsymbol{E}_{\text{L}}\left(\boldsymbol{E}_{\text{L}}^{*}\boldsymbol{E}_{\text{L}}\right)^{-\frac{1}{2}}\boldsymbol{B}\left(\boldsymbol{E}_{\text{R}}\boldsymbol{E}_{\text{R}}^{*}\right)^{-\frac{1}{2}}\boldsymbol{E}_{\text{R}}.
\]
For any $\left(k,l\right)\in[n_{1}]\times[n_{2}]$, we can then bound
\begin{align*}
 & \left|\left(\boldsymbol{U}\boldsymbol{V}^{*}\right)_{k,l}\right|=\left|\boldsymbol{e}_{k}^{\top}\boldsymbol{E}_{\text{L}}\left(\boldsymbol{E}_{\text{L}}^{*}\boldsymbol{E}_{\text{L}}\right)^{-\frac{1}{2}}\boldsymbol{B}\left(\boldsymbol{E}_{\text{R}}\boldsymbol{E}_{\text{R}}^{*}\right)^{-\frac{1}{2}}\boldsymbol{E}_{\text{R}}\boldsymbol{e}_{l}\right|\\
 & \quad\leq\left\Vert \boldsymbol{e}_{k}^{\top}\boldsymbol{E}_{\text{L}}\right\Vert _{\text{F}}\left\Vert \left(\boldsymbol{E}_{\text{L}}^{*}\boldsymbol{E}_{\text{L}}\right)^{-\frac{1}{2}}\right\Vert \left\Vert \boldsymbol{B}\right\Vert \left\Vert \left(\boldsymbol{E}_{\text{R}}^{*}\boldsymbol{E}_{\text{R}}\right)^{-\frac{1}{2}}\right\Vert \left\Vert \boldsymbol{E}_{\text{R}}\boldsymbol{e}_{l}\right\Vert _{\text{F}}\\
 & \quad\leq\sqrt{\frac{r}{k_{1}k_{2}}}\mu_{1}\sqrt{\frac{r}{\left(n_{1}-k_{1}+1\right)\left(n_{2}-k_{2}+1\right)}}\\
 & \quad\leq\frac{\mu_{1}c_{\text{s}}r}{n_{1}n_{2}}.
\end{align*}
Since $\boldsymbol{A}_{(k,l)}$ has only $\omega_{k,l}$ nonzero entries
each of magnitude $\frac{1}{\sqrt{\omega_{k,l}}}$, this leads to
\begin{align*}
\left\Vert \boldsymbol{U}\boldsymbol{V}^{*}\right\Vert _{\mathcal{A},\infty} & =\frac{1}{\omega_{k,l}}\left|\sum_{(\alpha,\beta)\in\Omega_{\text{e}}\left(k,l\right)}\left(\boldsymbol{U}\boldsymbol{V}^{*}\right)_{\alpha,\beta}\right|\\
 & \leq\max_{k,l}\left|\left(\boldsymbol{U}\boldsymbol{V}^{*}\right)_{k,l}\right|\leq\frac{\mu_{1}c_{\text{s}}r}{n_{1}n_{2}}.
\end{align*}

The rest is to bound $\left\Vert \boldsymbol{U}\boldsymbol{V}^{*}\right\Vert _{\mathcal{A},2}$
and $\left\Vert \mathcal{P}_{T}\left(\sqrt{\omega_{k,l}}\boldsymbol{A}_{\left(k,l\right)}\right)\right\Vert _{\mathcal{A},2}$.
Observe that the $i$th row of $\boldsymbol{U}\boldsymbol{V}^{*}$
obeys 
\begin{align}
\left\Vert \boldsymbol{e}_{i}^{\top}\boldsymbol{U}\boldsymbol{V}^{*}\right\Vert _{\mathrm{F}}^{2} & =\left\Vert \boldsymbol{e}_{i}^{\top}\boldsymbol{U}\right\Vert _{\mathrm{F}}^{2}=\left\Vert \boldsymbol{e}_{i}^{\top}\boldsymbol{E}_{\text{L}}\left(\boldsymbol{E}_{\text{L}}^{*}\boldsymbol{E}_{\text{L}}\right)^{-\frac{1}{2}}\right\Vert _{\mathrm{F}}^{2}\nonumber \\
 & \leq\left\Vert \boldsymbol{e}_{i}^{\top}\boldsymbol{E}_{\text{L}}\right\Vert _{\mathrm{F}}^{2}\left\Vert \left(\boldsymbol{E}_{\text{L}}^{*}\boldsymbol{E}_{\text{L}}\right)^{-1}\right\Vert \nonumber \\
 & \leq\mu_{1}\left\Vert \boldsymbol{e}_{i}^{\top}\boldsymbol{E}_{\text{L}}\right\Vert _{\mathrm{F}}^{2}\leq\frac{\mu_{1}c_{\mathrm{s}}r}{n_{1}n_{2}}.\label{eq:UV_row_energy}
\end{align}
That said, the total energy allocated to any row of $\boldsymbol{U}\boldsymbol{V}^{*}$
cannot exceed $\frac{\mu_{1}c_{\mathrm{s}}r}{n_{1}n_{2}}$.

Moreover, the matrix $\mathcal{P}_{T}\left(\sqrt{\omega_{\alpha,\beta}}\boldsymbol{A}_{(\alpha,\beta)}\right)$
enjoys similar properties as well, which we briefly reason as follows.
First, the matrix $\boldsymbol{U}\boldsymbol{U}^{*}\left(\sqrt{\omega_{\alpha,\beta}}\boldsymbol{A}_{(\alpha,\beta)}\right)$
obeys 
\begin{align*}
\left\Vert \boldsymbol{e}_{i}^{\top}\boldsymbol{U}\boldsymbol{U}^{*}\left(\sqrt{\omega_{\alpha,\beta}}\boldsymbol{A}_{(\alpha,\beta)}\right)\right\Vert _{\mathrm{F}}^{2} & \leq\left\Vert \boldsymbol{e}_{i}^{\top}\boldsymbol{U}\right\Vert _{\mathrm{F}}^{2}\left\Vert \boldsymbol{U}^{*}\right\Vert ^{2}\left\Vert \sqrt{\omega_{\alpha,\beta}}\boldsymbol{A}_{(\alpha,\beta)}\right\Vert ^{2}\\
 & \leq\frac{\mu_{1}c_{\mathrm{s}}r}{n_{1}n_{2}},
\end{align*}
since the operator norm of $\boldsymbol{U}$ and $\sqrt{\omega_{\alpha,\beta}}\boldsymbol{A}_{(\alpha,\beta)}$
are both bounded by 1. The same bound for $\sqrt{\omega_{\alpha,\beta}}\boldsymbol{A}_{(\alpha,\beta)}\boldsymbol{V}\boldsymbol{V}^{*}$
can be demonstrated via the same argument as for $\boldsymbol{U}\boldsymbol{U}^{*}\left(\sqrt{\omega_{\alpha,\beta}}\boldsymbol{A}_{(\alpha,\beta)}\right)$.
Additionally, for $\boldsymbol{U}\boldsymbol{U}^{*}\left(\sqrt{\omega_{\alpha,\beta}}\boldsymbol{A}_{(\alpha,\beta)}\right)\boldsymbol{V}\boldsymbol{V}^{*}$
one has 
\begin{align*}
 & \left\Vert \boldsymbol{e}_{i}^{\top}\boldsymbol{U}\boldsymbol{U}^{*}\left(\sqrt{\omega_{\alpha,\beta}}\boldsymbol{A}_{(\alpha,\beta)}\right)\boldsymbol{V}\boldsymbol{V}^{*}\right\Vert _{\mathrm{F}}^{2}\\
 & \quad\leq\left\Vert \boldsymbol{e}_{i}^{\top}\boldsymbol{U}\right\Vert _{\mathrm{F}}^{2}\left\Vert \boldsymbol{U}^{*}\right\Vert ^{2}\left\Vert \boldsymbol{V}\boldsymbol{V}^{*}\right\Vert ^{2}\left\Vert \sqrt{\omega_{\alpha,\beta}}\boldsymbol{A}_{(\alpha,\beta)}\right\Vert ^{2}\\
 & \quad\leq\frac{\mu_{1}c_{\mathrm{s}}r}{n_{1}n_{2}}.
\end{align*}
By definition of $\mathcal{P}_{T}$, 
\begin{align*}
 & \left\Vert \boldsymbol{e}_{i}^{\top}\mathcal{P}_{T}\left(\sqrt{\omega_{\alpha,\beta}}\boldsymbol{A}_{(\alpha,\beta)}\right)\right\Vert _{\mathrm{F}}^{2}\leq3\left\Vert \boldsymbol{e}_{i}^{\top}\boldsymbol{U}\boldsymbol{U}^{*}\left(\sqrt{\omega_{\alpha,\beta}}\boldsymbol{A}_{(\alpha,\beta)}\right)\right\Vert _{\mathrm{F}}^{2}\\
 & \quad\quad\quad\quad\quad\quad\quad\quad\quad+3\left\Vert \boldsymbol{e}_{i}^{\top}\left(\sqrt{\omega_{\alpha,\beta}}\boldsymbol{A}_{(\alpha,\beta)}\right)\boldsymbol{V}\boldsymbol{V}^{*}\right\Vert _{\mathrm{F}}^{2}\\
 & \quad\quad\quad\quad\quad\quad\quad\quad\quad+3\left\Vert \boldsymbol{e}_{i}^{\top}\boldsymbol{U}\boldsymbol{U}^{*}\left(\sqrt{\omega_{\alpha,\beta}}\boldsymbol{A}_{(\alpha,\beta)}\right)\boldsymbol{V}\boldsymbol{V}^{*}\right\Vert _{\mathrm{F}}^{2}\\
 & \quad\leq\frac{9\mu_{1}c_{\mathrm{s}}r}{n_{1}n_{2}}.
\end{align*}

Now our task boils down to bounding $\left\Vert \boldsymbol{M}\right\Vert _{\mathcal{A},2}$
for some matrix $\boldsymbol{M}$ satisfying some energy constraints
per row, which subsumes $\left\Vert \boldsymbol{U}\boldsymbol{V}^{*}\right\Vert _{\mathcal{A},2}$
and $\left\Vert \mathcal{P}_{T}\left(\sqrt{\omega_{k,l}}\boldsymbol{A}_{\left(k,l\right)}\right)\right\Vert _{\mathcal{A},2}$
as special cases. We can then conclude the proof by applying the following
lemma.

\begin{lemma}Denote by the set $\mathcal{M}$ of feasible matrices
satisfying 
\begin{align}
\max_{i}\left\Vert \boldsymbol{e}_{i}^{\top}\boldsymbol{M}\right\Vert _{\mathrm{F}}^{2} & \leq\frac{9\mu_{1}c_{\mathrm{s}}r}{n_{1}n_{2}}.\label{eq:row_sum}
\end{align}
Then there exists some universal constant $c_{3}>0$ such that 
\begin{equation}
\max_{\boldsymbol{M}\in\mathcal{M}}\left\Vert \boldsymbol{M}\right\Vert _{\mathcal{A},2}^{2}\leq c_{3}\frac{\mu_{1}c_{\mathrm{s}}r}{n_{1}n_{2}}\log^{2}\left(n_{1}n_{2}\right).\label{eq:fM_bound}
\end{equation}
\end{lemma}

\begin{proof}For ease of presentation, we split any matrix $\boldsymbol{M}$
into 4 parts, which are defined as follows 
\begin{itemize}
\item $\boldsymbol{M}^{\left(1\right)}$: the matrix containing all upper
triangular components of all upper triangular blocks of $\boldsymbol{M}$; 
\item $\boldsymbol{M}^{\left(2\right)}$: the matrix containing all lower
triangular components of all upper triangular blocks of $\boldsymbol{M}$; 
\item $\boldsymbol{M}^{\left(3\right)}$: the matrix containing all upper
triangular components of all lower triangular blocks of $\boldsymbol{M}$; 
\item $\boldsymbol{M}^{\left(4\right)}$: the matrix containing all lower
triangular components of all lower triangular blocks of $\boldsymbol{M}$. 
\end{itemize}
Here, we use the term ``upper triangular'' and ``lower triangular''
in short for ``left upper triangular'' and ``right lower triangular'',
which are more natural for Hankel matrices. Instead of maximizing
$\left\Vert \boldsymbol{M}\right\Vert _{\mathcal{A},2}$ directly,
we will handle $\max_{\boldsymbol{M}\in\mathcal{M}}\|\boldsymbol{M}^{(l)}\|_{\mathcal{A},2}^{2}$
for each $1\leq l\leq4$ separately, owing to the fact that 
\begin{equation}
\max_{\boldsymbol{M}\in\mathcal{M}}\left\Vert \boldsymbol{M}\right\Vert _{\mathcal{A},2}^{2}\leq4\max_{\boldsymbol{M}:\text{ }\boldsymbol{M}^{(l)}\in\mathcal{M}}\left\Vert \boldsymbol{M}^{(l)}\right\Vert _{\mathcal{A},2}^{2}.
\end{equation}
In the sequel, we only demonstrate how to control $\|\boldsymbol{M}^{(1)}\|_{\mathcal{A},2}$.
Similar bounds can be derived for $\|\boldsymbol{M}^{(l)}\|_{\mathcal{A},2}$
($2\leq l\leq4$) via very similar argument.

To facilitate analysis, we divide the entire index set into several
subsets $\mathcal{W}_{i,j}$ such that for all $1\leq i\leq\left\lceil \log\left(n_{1}\right)\right\rceil $
and $1\leq j\leq\left\lceil \log\left(n_{2}\right)\right\rceil $,
\begin{equation}
\mathcal{W}_{i,j}:=\bigcup\left\{ \Omega_{\mathrm{e}}\left(k,l\right)\mid\left(k,l\right)\in\left[2^{i-1},2^{i}\right]\times\left[2^{j-1},2^{j}\right]\right\} .
\end{equation}
Consequently, for each $\Omega_{\mathrm{e}}\left(k,l\right)\subseteq\mathcal{W}_{i,j}$,
one has 
\[
2^{i-1}\cdot2^{j-1}\leq\omega_{k,l}\leq2^{i+j}.
\]
This allows us to derive for each $\mathcal{W}_{i,j}$ that 
\begin{align}
 & \sum_{\left(k,l\right)\in\mathcal{W}_{i,j}}\frac{1}{\omega_{k,l}^{2}}\left|\sum\nolimits _{\left(\alpha,\beta\right)\in\Omega_{\mathrm{e}}\left(k,l\right)}\boldsymbol{M}_{\alpha,\beta}^{(1)}\right|^{2}\nonumber \\
 & \quad\leq\sum_{\left(k,l\right)\in\mathcal{W}_{i,j}}\frac{1}{\omega_{k,l}}\sum\nolimits _{\left(\alpha,\beta\right)\in\Omega_{\mathrm{e}}\left(k,l\right)}\left|\boldsymbol{M}_{\alpha,\beta}^{(1)}\right|^{2}\label{eq:Arithmetic}\\
 & \quad\leq\frac{1}{2^{i+j-2}}\sum_{\left(k,l\right)\in\mathcal{W}_{i,j}}\sum\nolimits _{\left(\alpha,\beta\right)\in\Omega_{\mathrm{e}}\left(k,l\right)}\left|\boldsymbol{M}_{\alpha,\beta}^{(1)}\right|^{2},\label{eq:BoundEnergy}
\end{align}
where (\ref{eq:Arithmetic}) follows from the RMS-AM (root-mean square
v.s. arithmetic mean) inequality.

Observe that the indices contained in $\mathcal{W}_{i,j}$ reside
within no more than $2^{i}\cdot2^{j}$ rows. By assumption (\ref{eq:row_sum}),
the total energy allocated to $\mathcal{W}_{i,j}$ must be bounded
above by 
\begin{align*}
\sum_{\left(k,l\right)\in\mathcal{W}_{i,j}}\sum\nolimits _{\left(\alpha,\beta\right)\in\Omega_{\mathrm{e}}\left(k,l\right)}\left|\boldsymbol{M}_{\alpha,\beta}^{(1)}\right|^{2} & \leq2^{i}\cdot2^{j}\max_{i}\left\Vert \boldsymbol{e}_{i}^{\top}\boldsymbol{M}\right\Vert _{\text{F}}^{2}\\
 & \leq2^{i+j}\cdot\frac{9\mu_{1}c_{\mathrm{s}}r}{n_{1}n_{2}}.
\end{align*}
Substituting it into (\ref{eq:BoundEnergy}) immediately leads to
\begin{align}
\sum_{\left(k,l\right)\in\mathcal{W}_{i,j}}\frac{1}{\omega_{k,l}^{2}}\left|\sum\nolimits _{\left(\alpha,\beta\right)\in\Omega_{\mathrm{e}}\left(k,l\right)}\boldsymbol{M}_{\alpha,\beta}^{(1)}\right|^{2} & \leq\frac{36\mu_{1}c_{\mathrm{s}}r}{n_{1}n_{2}}.
\end{align}
By definition, 
\[
\small\left\Vert \boldsymbol{M}\right\Vert _{\mathcal{A},2}^{2}=\sum_{\begin{array}{c}
\footnotesize1\leq i\leq\left\lceil \log n_{1}\right\rceil \\
\footnotesize1\leq j\leq\left\lceil \log n_{2}\right\rceil 
\end{array}}\sum_{\left(k,l\right)\in\mathcal{W}_{i,j}}\frac{\left|\sum_{\left(\alpha,\beta\right)\in\Omega_{\mathrm{e}}\left(k,l\right)}\boldsymbol{M}_{\alpha,\beta}\right|^{2}}{\omega_{k,l}^{2}}.
\]
Combining the above bounds over all $\mathcal{W}_{i,j}$ then gives
\begin{align*}
\left\Vert \boldsymbol{M}^{(1)}\right\Vert _{\mathcal{A},2}^{2} & \leq\frac{36\mu_{1}c_{\mathrm{s}}r\left\lceil \log\left(n_{1}\right)\right\rceil \cdot\left\lceil \log\left(n_{2}\right)\right\rceil }{n_{1}n_{2}}
\end{align*}
as claimed.\end{proof}

\section{Proof of Lemma \ref{lemma-Dual-Robust}\label{sec:Proof-of-Lemma-Dual-Robust}}

Suppose there is a non-zero perturbation $(\boldsymbol{H},\boldsymbol{T})$
such that $(\boldsymbol{X}+\boldsymbol{H},\boldsymbol{S}+\boldsymbol{T})$
is the optimizer of Robust-EMaC. One can easily verify that $\mathcal{P}_{\Omega^{\perp}}\left(\boldsymbol{S}+\boldsymbol{T}\right)=0$,
otherwise we can always set $\boldsymbol{S}+\boldsymbol{T}$ as $\mathcal{P}_{\Omega}\left(\boldsymbol{S}+\boldsymbol{T}\right)$
to yield a better estimate. This together with the fact that $\mathcal{P}_{\Omega^{\perp}}\left(\boldsymbol{S}\right)=0$
implies that $\mathcal{P}_{\Omega}\left(\boldsymbol{T}\right)=\boldsymbol{T}$.
Observe that the constraints of Robust-EMaC indicate 
\begin{align*}
\mathcal{P}_{\Omega}\left(\boldsymbol{X}+\boldsymbol{S}\right) & =\mathcal{P}_{\Omega}\left(\boldsymbol{X}+\boldsymbol{H}+\boldsymbol{S}+\boldsymbol{T}\right),\\
\Rightarrow\quad & \mathcal{P}_{\Omega}\left(\boldsymbol{H}+\boldsymbol{T}\right)=0,
\end{align*}
which is equivalent to requiring $\mathcal{A}'_{\Omega}\left(\boldsymbol{H}_{\text{e}}\right)=-\mathcal{A}'_{\Omega}\left(\boldsymbol{T}_{\text{e}}\right)=-\boldsymbol{T}_{\text{e}}$
and $\mathcal{A}^{\perp}\left(\boldsymbol{H}_{\text{e}}\right)=0$.

Recall that $\boldsymbol{H}_{\text{e}}$ and $\boldsymbol{S}_{\text{e}}$
are the enhanced forms of $\boldsymbol{H}$ and $\boldsymbol{S}$,
respectively. Set $\boldsymbol{W}_{0}\in T^{\perp}$ to be a matrix
satisfying $\left\langle \boldsymbol{W}_{0},\mathcal{P}_{T^{\perp}}\left(\boldsymbol{H}_{\text{e}}\right)\right\rangle =\left\Vert \mathcal{P}_{T^{\perp}}\left(\boldsymbol{H}_{\text{e}}\right)\right\Vert _{*}$
and $\left\Vert \boldsymbol{W}_{0}\right\Vert \leq1$, then $\boldsymbol{U}\boldsymbol{V}^{*}+\boldsymbol{W}_{0}$
is a sub-gradient of the nuclear norm at $\boldsymbol{X}_{\text{e}}$.
This gives 
\begin{align}
\left\Vert \boldsymbol{X}_{\text{e}}+\boldsymbol{H}_{\text{e}}\right\Vert _{*} & \geq\left\Vert \boldsymbol{X}_{\text{e}}\right\Vert _{*}+\left\langle \boldsymbol{U}\boldsymbol{V}^{*}+\boldsymbol{W}_{0},\boldsymbol{H}_{\text{e}}\right\rangle \nonumber \\
 & =\left\Vert \boldsymbol{X}_{\text{e}}\right\Vert _{*}+\left\langle \boldsymbol{U}\boldsymbol{V}^{*},\boldsymbol{H}_{\text{e}}\right\rangle +\left\Vert \mathcal{P}_{T^{\perp}}\left(\boldsymbol{H}_{\text{e}}\right)\right\Vert _{*}.\label{eq:subgradient_robust}
\end{align}
Owing to the fact that $\text{support}\left(\boldsymbol{S}\right)\subseteq\Omega^{\text{dirty}}$,
one has $\boldsymbol{S}_{\text{e}}=\mathcal{A}'_{\Omega^{\text{dirty}}}\left(\boldsymbol{S}_{\text{e}}\right)$.
Combining this and the fact that $\text{support}\left(\boldsymbol{S}_{\text{e}}+\boldsymbol{T}_{\text{e}}\right)\subseteq\Omega$
yields 
\[
\|\boldsymbol{S}_{\text{e}}+\boldsymbol{T}_{\text{e}}\|_{1}=\left\Vert \mathcal{A}'_{\Omega^{\text{clean}}}(\boldsymbol{T}_{\text{e}})\right\Vert _{1}+\left\Vert \boldsymbol{S}_{\text{e}}+\mathcal{A}'_{\Omega^{\text{dirty}}}\left(\boldsymbol{T}_{\text{e}}\right)\right\Vert _{1},
\]
which further gives 
\begin{align}
 & \left\Vert \boldsymbol{S}_{\text{e}}+\boldsymbol{T}_{\text{e}}\right\Vert _{1}-\left\Vert \boldsymbol{S}_{\text{e}}\right\Vert _{1}\nonumber \\
 & \quad=\left\Vert \mathcal{A}'_{\Omega^{\text{clean}}}\left(\boldsymbol{T}_{\text{e}}\right)\right\Vert _{1}+\left\Vert \boldsymbol{S}_{\text{e}}+\mathcal{A}'_{\Omega^{\text{dirty}}}\left(\boldsymbol{T}_{\text{e}}\right)\right\Vert _{1}-\left\Vert \boldsymbol{S}_{\text{e}}\right\Vert _{1}\nonumber \\
 & \quad\geq\left\Vert \mathcal{A}'_{\Omega^{\text{clean}}}\left(\boldsymbol{T}_{\text{e}}\right)\right\Vert _{1}+\left\langle \mbox{sgn}\left(\boldsymbol{S}_{\text{e}}\right),\mathcal{A}'_{\Omega^{\text{dirty}}}\left(\boldsymbol{T}_{\text{e}}\right)\right\rangle \label{bound112}\\
 & \quad=\left\Vert \mathcal{A}'_{\Omega^{\text{clean}}}\left(\boldsymbol{T}_{\text{e}}\right)\right\Vert _{1}-\left\langle \mbox{sgn}\left(\boldsymbol{S}_{\text{e}}\right),\mathcal{A}'_{\Omega^{\text{dirty}}}\left(\boldsymbol{H}_{\text{e}}\right)\right\rangle \label{bound113}\\
 & \quad=\left\Vert \mathcal{A}'_{\Omega^{\text{clean}}}\left(\boldsymbol{T}_{\text{e}}\right)\right\Vert _{1}-\left\langle \mathcal{A}'_{\Omega^{\text{dirty}}}\left(\mbox{sgn}\left(\boldsymbol{S}_{\text{e}}\right)\right),\boldsymbol{H}_{\text{e}}\right\rangle \nonumber \\
 & \quad=\left\Vert \mathcal{A}'_{\Omega^{\text{clean}}}\left(\boldsymbol{H}_{\text{e}}\right)\right\Vert _{1}-\left\langle \mbox{sgn}\left(\boldsymbol{S}_{\text{e}}\right),\boldsymbol{H}_{\text{e}}\right\rangle .\label{eq:subgradientS_robust}
\end{align}
Here, \eqref{bound112} follows from the fact that $\text{sgn}(\boldsymbol{S}_{\text{e}})$
is the sub-gradient of $\left\Vert \cdot\right\Vert _{1}$ at $\boldsymbol{S}_{\text{e}}$,
and \eqref{bound113} arises from the identity $\mathcal{P}_{\Omega^{\text{dirty}}}\left(\boldsymbol{H}+\boldsymbol{T}\right)=0$
and hence $\mathcal{A}'_{\Omega^{\text{dirty}}}\left(\boldsymbol{H}_{\text{e}}\right)=-\mathcal{A}'_{\Omega^{\text{dirty}}}\left(\boldsymbol{T}_{\text{e}}\right)$.
The inequalities \eqref{eq:subgradient_robust} and \eqref{eq:subgradientS_robust}
taken collectively lead to 
\begin{align}
 & \left\Vert \boldsymbol{X}_{\text{e}}+\boldsymbol{H}_{\text{e}}\right\Vert _{*}+\lambda\left\Vert \boldsymbol{S}_{\text{e}}+\boldsymbol{T}_{\text{e}}\right\Vert _{1}-\left(\left\Vert \boldsymbol{X}_{\text{e}}\right\Vert _{*}+\lambda\left\Vert \boldsymbol{S}_{\text{e}}\right\Vert _{1}\right)\nonumber \\
 & \text{ }\text{ }\geq\text{ }\left\langle \boldsymbol{U}\boldsymbol{V}^{*},\boldsymbol{H}_{\text{e}}\right\rangle +\left\Vert \mathcal{P}_{T^{\perp}}\left(\boldsymbol{H}_{\text{e}}\right)\right\Vert _{*}+\lambda\left\Vert \mathcal{A}'_{\Omega^{\text{clean}}}\left(\boldsymbol{H}_{\text{e}}\right)\right\Vert _{1}\nonumber \\
 & \quad\quad\quad\quad\quad-\lambda\left\langle \mbox{sgn}\left(\boldsymbol{S}_{\text{e}}\right),\boldsymbol{H}_{\text{e}}\right\rangle \nonumber \\
 & \quad\geq-\left\langle \lambda\mbox{sgn}\left(\boldsymbol{S}_{\text{e}}\right)-\boldsymbol{U}\boldsymbol{V}^{*},\boldsymbol{H}_{\text{e}}\right\rangle +\left\Vert \mathcal{P}_{T^{\perp}}\left(\boldsymbol{H}_{\text{e}}\right)\right\Vert _{*}\nonumber \\
 & \text{ }\text{ }\quad\quad\quad\quad+\lambda\left\Vert \mathcal{A}'_{\Omega^{\text{clean}}}\left(\boldsymbol{H}_{\text{e}}\right)\right\Vert _{1}.\label{eq:subgradient_total_robust}
\end{align}

It remains to show that the right-hand side of \eqref{eq:subgradient_total_robust}
cannot be negative. For a dual matrix $\boldsymbol{W}$ satisfying
Conditions (\ref{eq:DualProperties-Robust}), one can derive 
\begin{align}
 & \left\langle \lambda\mbox{sgn}\left(\boldsymbol{S}_{\text{e}}\right)-\boldsymbol{U}\boldsymbol{V}^{*},\boldsymbol{H}_{\text{e}}\right\rangle \nonumber \\
 & \text{ }\text{ }=\left\langle \boldsymbol{W}+\lambda\mbox{sgn}\left(\boldsymbol{S}_{\text{e}}\right)-\boldsymbol{U}\boldsymbol{V}^{*},\boldsymbol{H}_{\text{e}}\right\rangle -\left\langle \boldsymbol{W},\boldsymbol{H}_{\text{e}}\right\rangle \nonumber \\
 & \text{ }\text{ }=\left\langle \mathcal{P}_{T}\left(\boldsymbol{W}+\lambda\mbox{sgn}\left(\boldsymbol{S}_{\text{e}}\right)-\boldsymbol{U}\boldsymbol{V}^{*}\right),\mathcal{P}_{T}\left(\boldsymbol{H}_{\text{e}}\right)\right\rangle \nonumber \\
 & \quad\quad+\left\langle \mathcal{P}_{T^{\perp}}\left(\boldsymbol{W}+\lambda\mbox{sgn}\left(\boldsymbol{S}_{\text{e}}\right)-\boldsymbol{U}\boldsymbol{V}^{*}\right),\mathcal{P}_{T^{\perp}}\left(\boldsymbol{H}_{\text{e}}\right)\right\rangle \nonumber \\
 & \text{ }\text{ }\quad-\left\langle \mathcal{A}'_{\Omega^{\text{clean}}}\left(\boldsymbol{W}\right),\mathcal{A}'_{\Omega^{\text{clean}}}\left(\boldsymbol{H}_{\text{e}}\right)\right\rangle \nonumber \\
 & \quad\quad-\left\langle \mathcal{A}'_{\left(\Omega^{\text{clean}}\right)^{\perp}}\left(\boldsymbol{W}\right),\mathcal{A}'_{\left(\Omega^{\text{clean}}\right)^{\perp}}\left(\boldsymbol{H}_{\text{e}}\right)\right\rangle \nonumber \\
 & \text{ }\text{ }\leq\footnotesize\text{ }\frac{\lambda}{n_{1}^{2}n_{2}^{2}}\left\Vert \mathcal{P}_{T}\left(\boldsymbol{H}_{\text{e}}\right)\right\Vert _{\text{F}}+\frac{1}{4}\left\Vert \mathcal{P}_{T^{\perp}}\left(\boldsymbol{H}_{\text{e}}\right)\right\Vert _{*}+\frac{\lambda}{4}\left\Vert \mathcal{A}'_{\Omega^{\text{clean}}}\left(\boldsymbol{H}_{\text{e}}\right)\right\Vert _{\text{1}},\label{eq:InequalityUse4PropertiesOfW}
\end{align}
where the last inequality follows from the four properties of $\boldsymbol{W}$
in (\ref{eq:DualProperties-Robust}). Since $\left(\boldsymbol{X}+\boldsymbol{H},\boldsymbol{S}+\boldsymbol{T}\right)$
is assumed to be the optimizer, substituting (\ref{eq:InequalityUse4PropertiesOfW})
into \eqref{eq:subgradient_total_robust} then yields 
\begin{align}
0\text{ }\geq & \small\left\Vert \boldsymbol{X}_{\text{e}}+\boldsymbol{H}_{\text{e}}\right\Vert _{*}+\lambda\left\Vert \boldsymbol{S}_{\text{e}}+\boldsymbol{T}_{\text{e}}\right\Vert _{1}-\left(\left\Vert \boldsymbol{X}_{\text{e}}\right\Vert _{*}+\lambda\left\Vert \boldsymbol{S}_{\text{e}}\right\Vert _{1}\right)\label{eq:InequalityRobust1}\\
\geq & \small\text{ }\frac{3}{4}\left\Vert \mathcal{P}_{T^{\perp}}\left(\boldsymbol{H}_{\text{e}}\right)\right\Vert _{*}+\frac{3}{4}\lambda\left\Vert \mathcal{A}'_{\Omega^{\text{clean}}}\left(\boldsymbol{H}_{\text{e}}\right)\right\Vert _{1}-\frac{\lambda}{n_{1}^{2}n_{2}^{2}}\left\Vert \mathcal{P}_{T}\left(\boldsymbol{H}_{\text{e}}\right)\right\Vert _{\text{F}}\nonumber \\
\geq & \small\text{ }\frac{3}{4}\left\Vert \mathcal{P}_{T^{\perp}}\left(\boldsymbol{H}_{\text{e}}\right)\right\Vert _{*}+\frac{3}{4}\lambda\left\Vert \mathcal{A}'_{\Omega^{\text{clean}}}\left(\boldsymbol{H}_{\text{e}}\right)\right\Vert _{\text{F}}-\frac{\lambda}{n_{1}^{2}n_{2}^{2}}\left\Vert \mathcal{P}_{T}\left(\boldsymbol{H}_{\text{e}}\right)\right\Vert _{\text{F}},\label{eq:PTperpAPTineuqliaty}
\end{align}
where (\ref{eq:PTperpAPTineuqliaty}) arises due to the inequality
$\left\Vert \boldsymbol{M}\right\Vert _{\text{F}}\leq\left\Vert \boldsymbol{M}\right\Vert _{1}$.

The invertibility condition (\ref{eq:InvertibilityPtAcleanPt}) on
$\mathcal{P}_{T}\mathcal{A}_{\Omega^{\text{clean}}}\mathcal{P}_{T}$
is equivalent to 
\[
\left\Vert \mathcal{P}_{T}-\mathcal{P}_{T}\left(\frac{1}{\rho\left(1-\tau\right)}\mathcal{A}_{\Omega^{\text{clean}}}+\mathcal{A}^{\perp}\right)\mathcal{P}_{T}\right\Vert \leq\frac{1}{2},
\]
indicating that 
\begin{align*}
\frac{1}{2}\left\Vert \mathcal{P}_{T}\left(\boldsymbol{H}_{\text{e}}\right)\right\Vert _{\text{F}} & \leq\left\Vert \mathcal{P}_{T}\left(\frac{1}{\rho\left(1-\tau\right)}\mathcal{A}_{\Omega^{\text{clean}}}+\mathcal{A}^{\perp}\right)\mathcal{P}_{T}\left(\boldsymbol{H}_{\text{e}}\right)\right\Vert _{\text{F}}\\
 & \leq\frac{3}{2}\left\Vert \mathcal{P}_{T}\left(\boldsymbol{H}_{\text{e}}\right)\right\Vert _{\text{F}}.
\end{align*}
One can, therefore, bound $\left\Vert \mathcal{P}_{T}\left(\boldsymbol{H}_{\text{e}}\right)\right\Vert _{\text{F}}$
as follows 
\begin{align}
 & \left\Vert \mathcal{P}_{T}\left(\boldsymbol{H}_{\text{e}}\right)\right\Vert _{\text{F}}\leq2\left\Vert \mathcal{P}_{T}\left(\frac{1}{\rho\left(1-\tau\right)}\mathcal{A}_{\Omega^{\text{clean}}}+\mathcal{A}^{\perp}\right)\mathcal{P}_{T}\left(\boldsymbol{H}_{\text{e}}\right)\right\Vert _{\text{F}}\nonumber \\
 & \quad\leq\frac{2}{\rho\left(1-\tau\right)}\left\Vert \mathcal{P}_{T}\mathcal{A}_{\Omega^{\text{clean}}}\mathcal{P}_{T}\left(\boldsymbol{H}_{\text{e}}\right)\right\Vert _{\text{F}}+2\left\Vert \mathcal{P}_{T}\mathcal{A}^{\perp}\mathcal{P}_{T}\left(\boldsymbol{H}_{\text{e}}\right)\right\Vert _{\text{F}}\nonumber \\
 & \quad\leq\frac{2}{\rho\left(1-\tau\right)}\left(\left\Vert \mathcal{P}_{T}\mathcal{A}_{\Omega^{\text{clean}}}\left(\boldsymbol{H}_{\text{e}}\right)\right\Vert _{\text{F}}+\left\Vert \mathcal{P}_{T}\mathcal{A}_{\Omega^{\text{clean}}}\mathcal{P}_{T^{\perp}}\left(\boldsymbol{H}_{\text{e}}\right)\right\Vert _{\text{F}}\right)\nonumber \\
 & \quad\quad\quad+2\left\Vert \mathcal{P}_{T}\mathcal{A}^{\perp}\left(\boldsymbol{H}_{\text{e}}\right)\right\Vert _{\text{F}}+2\left\Vert \mathcal{P}_{T}\mathcal{A}^{\perp}\mathcal{P}_{T^{\perp}}\left(\boldsymbol{H}_{\text{e}}\right)\right\Vert _{\text{F}}\nonumber \\
 & \quad\leq\frac{2}{\rho\left(1-\tau\right)}\left(\left\Vert \mathcal{A}_{\Omega^{\text{clean}}}\left(\boldsymbol{H}_{\text{e}}\right)\right\Vert _{\text{F}}+\left\Vert \mathcal{A}_{\Omega^{\text{clean}}}\mathcal{P}_{T^{\perp}}\left(\boldsymbol{H}_{\text{e}}\right)\right\Vert _{\text{F}}\right)\nonumber \\
 & \quad\quad\quad+2\left\Vert \mathcal{P}_{T^{\perp}}\left(\boldsymbol{H}_{\text{e}}\right)\right\Vert _{\text{F}},\label{eq:PtHeExpansion}
\end{align}
where the last inequality exploit the facts that $\mathcal{A}^{\perp}\left(\boldsymbol{H}_{\text{e}}\right)=0$
and $\left\Vert \mathcal{P}_{T}\left(\boldsymbol{M}\right)\right\Vert _{\text{F}}\leq\left\Vert \boldsymbol{M}\right\Vert _{\text{F}}$.

Recall that $\mathcal{A}_{\Omega^{\text{clean}}}$ corresponds to
sampling with replacement. Condition (\ref{eq:ConnectionSamplingReplacement})
together with (\ref{eq:PtHeExpansion}) leads to 
\begin{align}
 & \left\Vert \mathcal{P}_{T}\left(\boldsymbol{H}_{\text{e}}\right)\right\Vert _{\text{F}}\nonumber \\
 & \leq\frac{20\log\left(n_{1}n_{2}\right)}{\rho\left(1-\tau\right)}\left(\left\Vert \mathcal{A}'_{\Omega^{\text{clean}}}\left(\boldsymbol{H}_{\text{e}}\right)\right\Vert _{\text{F}}+\left\Vert \mathcal{A}'_{\Omega^{\text{clean}}}\mathcal{P}_{T^{\perp}}\left(\boldsymbol{H}_{\text{e}}\right)\right\Vert _{\text{F}}\right)\nonumber \\
 & \quad\quad\quad\quad+2\left\Vert \mathcal{P}_{T^{\perp}}\left(\boldsymbol{H}_{\text{e}}\right)\right\Vert _{\text{F}}\nonumber \\
 & \leq\frac{20\log\left(n_{1}n_{2}\right)}{\rho\left(1-\tau\right)}\left\Vert \mathcal{A}'_{\Omega^{\text{clean}}}\left(\boldsymbol{H}_{\text{e}}\right)\right\Vert _{\text{F}}\nonumber \\
 & \quad\quad\quad\quad+\left(\frac{20\log\left(n_{1}n_{2}\right)}{\rho\left(1-\tau\right)}+2\right)\left\Vert \mathcal{P}_{T^{\perp}}\left(\boldsymbol{H}_{\text{e}}\right)\right\Vert _{\text{F}}\nonumber \\
 & \leq\frac{20\log\left(n_{1}n_{2}\right)}{\rho\left(1-\tau\right)}\left\Vert \mathcal{A}'_{\Omega^{\text{clean}}}\left(\boldsymbol{H}_{\text{e}}\right)\right\Vert _{\text{F}}\nonumber \\
 & \quad\quad\quad\quad+\left(\frac{20\log\left(n_{1}n_{2}\right)}{\rho\left(1-\tau\right)}+2\right)\left\Vert \mathcal{P}_{T^{\perp}}\left(\boldsymbol{H}_{\text{e}}\right)\right\Vert _{*},\label{eq:UBPtHe}
\end{align}
where the last inequality follows from the fact that $\left\Vert \boldsymbol{M}\right\Vert _{\text{F}}\leq\left\Vert \boldsymbol{M}\right\Vert _{*}$.
Substituting (\ref{eq:UBPtHe}) into (\ref{eq:PTperpAPTineuqliaty})
yields 
\begin{align}
\left(\frac{3}{4}-\frac{\lambda}{n_{1}^{2}n_{2}^{2}}\left(\frac{20\log\left(n_{1}n_{2}\right)}{\rho\left(1-\tau\right)}+2\right)\right)\left\Vert \mathcal{P}_{T^{\perp}}\left(\boldsymbol{H}_{\text{e}}\right)\right\Vert _{*}\nonumber \\
\quad+\lambda\left(\frac{3}{4}-\frac{20\log\left(n_{1}n_{2}\right)}{\rho\left(1-\tau\right)n_{1}^{2}n_{2}^{2}}\right)\left\Vert \mathcal{A}'_{\Omega^{\text{clean}}}\left(\boldsymbol{H}_{\text{e}}\right)\right\Vert _{\text{F}}\leq0.\label{eq:RobustPtperpAwclean}
\end{align}
Since $\lambda<1$ and $\rho n_{1}^{2}n_{2}^{2}\gg\log\left(n_{1}n_{2}\right)$,
both terms on the left-hand side of (\ref{eq:RobustPtperpAwclean})
are positive. This can only occur when 
\begin{equation}
\mathcal{P}_{T^{\perp}}\left(\boldsymbol{H}_{\text{e}}\right)=0\quad\text{and}\quad\mathcal{A}'_{\Omega^{\text{clean}}}\left(\boldsymbol{H}_{\text{e}}\right)=0.\label{eq:PropertiesOfHe_Robust}
\end{equation}

(1) Consider first the situation where 
\begin{equation}
\left\Vert \mathcal{P}_{T}\left(\boldsymbol{H}_{\text{e}}\right)\right\Vert _{\text{F}}\leq\frac{n_{1}^{2}n_{2}^{2}}{2}\left\Vert \mathcal{P}_{T^{\perp}}\left(\boldsymbol{H}_{\text{e}}\right)\right\Vert _{\text{F}}.\label{eq:ConditionPtPtperp_Robust}
\end{equation}
One can immediately see that 
\[
\left\Vert \mathcal{P}_{T}\left(\boldsymbol{H}_{\text{e}}\right)\right\Vert _{\text{F}}\leq\frac{n_{1}^{2}n_{2}^{2}}{2}\left\Vert \mathcal{P}_{T^{\perp}}\left(\boldsymbol{H}_{\text{e}}\right)\right\Vert _{\text{F}}=0
\]
which implies $\mathcal{P}_{T}\left(\boldsymbol{H}_{\text{e}}\right)=\mathcal{P}_{T^{\perp}}\left(\boldsymbol{H}_{\text{e}}\right)=0$,
and therefore $\boldsymbol{H}_{\text{e}}=0$. That said, Robust-EMaC
succeeds in finding $\boldsymbol{X}_{\text{e}}$ under Condition (\ref{eq:ConditionPtPtperp_Robust}).

(2) Consider instead the complement situation where 
\[
\left\Vert \mathcal{P}_{T}\left(\boldsymbol{H}_{\text{e}}\right)\right\Vert _{\text{F}}>\frac{n_{1}^{2}n_{2}^{2}}{2}\left\Vert \mathcal{P}_{T^{\perp}}\left(\boldsymbol{H}_{\text{e}}\right)\right\Vert _{\text{F}}.
\]
Note that $\mathcal{A}'_{\Omega^{\text{clean}}}(\boldsymbol{H}_{\text{e}})=\mathcal{A}^{\perp}(\boldsymbol{H}_{\text{e}})=0$
and $\left\Vert \mathcal{P}_{T}\mathcal{A}\mathcal{P}_{T}-\frac{1}{\rho\left(1-\tau\right)}\mathcal{P}_{T}\mathcal{A}_{\Omega^{\text{clean}}}\mathcal{P}_{T}\right\Vert \leq\frac{1}{2}$.
Using the same argument as in the proof of Lemma \ref{lemma-Dual-Certificate}
(see the second part of Appendix \ref{sec:Proof-of-Lemma-Dual-Certificate})
with $\Omega$ replaced by $\Omega^{\text{clean}}$, we can conclude
$\boldsymbol{H}_{\text{e}}=0.$

\section{Proof of Lemma \ref{lemma-Psparse}\label{sec:Proof-of-Lemma-Psparse}}

We first state the following useful inequality in the proof. For any
$\boldsymbol{b}\in[n_{1}]\times[n_{2}]$, one has 
\begin{align}
 & \footnotesize\sum_{\boldsymbol{a}\in[n_{1}]\times[n_{2}]}\left|\left\langle \mathcal{P}_{T}\boldsymbol{A}_{\boldsymbol{b}},\boldsymbol{A}_{\boldsymbol{a}}\right\rangle \right|^{2}\omega_{\boldsymbol{a}}\leq\sum_{\boldsymbol{a}\in[n_{1}]\times[n_{2}]}\left(\sqrt{\frac{\omega_{\boldsymbol{b}}}{\omega_{\boldsymbol{a}}}}\frac{3\mu_{1}c_{\text{s}}r}{n_{1}n_{2}}\right)^{2}\omega_{\boldsymbol{a}}\label{eq:eq116}\\
 & \quad\quad\quad\footnotesize=\omega_{\boldsymbol{b}}\sum_{\boldsymbol{a}\in[n_{1}]\times[n_{2}]}\left(\frac{3\mu_{1}c_{\text{s}}r}{n_{1}n_{2}}\right)^{2}=\omega_{\boldsymbol{b}}\frac{9\mu_{1}^{2}c_{\text{s}}^{2}r^{2}}{n_{1}n_{2}},\label{eq:Mu4BoundViaMu1}
\end{align}
where \eqref{eq:eq116} follows from \eqref{eq:UBAbPtAa}. %introduce one additional incoherence measure that will be useful in stating our approach. %Let $\boldsymbol{X}_{\mathrm{e}}$

By definition, $\Omega^{\text{dirty}}$ is the set of \emph{distinct}
locations that appear in $\Omega$ but not in $\Omega^{\text{clean}}$.
To simplify the analysis, we introduce an auxiliary multi-set $\tilde{\Omega}^{\text{dirty}}$
that contains $\rho sn_{1}n_{2}$ i.i.d. entries. Specifically, suppose
that $\Omega=\left\{ \boldsymbol{a}_{i}\mid1\leq i\leq\rho n_{1}n_{2}\right\} $,
$\Omega^{\text{clean}}=\left\{ \boldsymbol{a}_{i}\mid1\leq i\leq\rho\left(1-\tau\right)n_{1}n_{2}\right\} $
and $\tilde{\Omega}^{\text{dirty}}=\left\{ \boldsymbol{a}_{i}\mid\rho\left(1-\tau\right)n_{1}n_{2}<i\leq\rho n_{1}n_{2}\right\} $,
where $\boldsymbol{a}_{i}$'s are independently and uniformly selected
from $[n_{1}]\times[n_{2}]$.

In addition, we consider an equivalent model for $\text{sgn}\left(\boldsymbol{S}\right)$
as follows 
\begin{itemize}
\item Define $\boldsymbol{K}=\left(K_{\alpha,\beta}\right)_{1\leq\alpha\leq n_{1},1\leq\beta\leq n_{2}}$
to be a random $n_{1}\times n_{2}$ matrix such that all of its entries
are independent and have amplitude 1 (i.e. in the real case, all entries
are either $1$ or $-1$, and in the complex case, all entries have
amplitude 1 and arbitrary phase on the unit circle). We assume that
$\mathbb{E}\left[\boldsymbol{K}\right]=\boldsymbol{0}$. 
\item Set $\text{sgn}\left(\boldsymbol{S}\right)$ such that $\text{sgn}\left(\boldsymbol{S}_{\alpha,\beta}\right)=K_{\alpha,\beta}{\bf 1}_{\left\{ \left(\alpha,\beta\right)\in\Omega^{\text{dirty}}\right\} }$,
and hence 
\[
\text{sgn}\left(\boldsymbol{S}_{\mathrm{e}}\right)=\sum_{\left(\alpha,\beta\right)\in\Omega^{\text{dirty}}}K_{\alpha,\beta}\sqrt{\omega_{\alpha,\beta}}\boldsymbol{A}_{\alpha,\beta}.
\]

\end{itemize}
Recall that $\text{support}\left(\boldsymbol{S}\right)\subseteq\Omega^{\text{dirty}}$.
Rather than directly studying $\text{sgn}\left(\boldsymbol{S}_{\text{e}}\right)$,
we will first examine an auxiliary matrix 
\[
\tilde{\boldsymbol{S}}_{\text{e}}:=\sum_{i=\rho\left(1-s\right)n_{1}n_{2}+1}^{\rho n_{1}n_{2}}K_{\boldsymbol{a}_{i}}\sqrt{\omega_{\boldsymbol{a}_{i}}}\boldsymbol{A}_{\boldsymbol{a}_{i}},
\]
and then bound the difference between $\tilde{\boldsymbol{S}}_{\text{e}}$
and $\text{sgn}\left(\boldsymbol{S}_{\text{e}}\right)$.

For any given pair $(k,l)\in[n_{1}]\times[n_{2}]$, define a random
variable 
\begin{align*}
\mathcal{Z}_{\alpha,\beta}: & =\sqrt{\frac{\omega_{\alpha,\beta}}{\omega_{k,l}}}\left\langle \mathcal{P}_{T}\boldsymbol{A}_{(k,l)},K_{\alpha,\beta}\boldsymbol{A}_{\alpha,\beta}\right\rangle .
\end{align*}
Thus, conditioned on $\boldsymbol{K}$, $\mathcal{Z}_{\boldsymbol{a}_{i}}$'s
are conditionally independent and $\frac{1}{\sqrt{\omega_{k,l}}}\left\langle \boldsymbol{A}_{(k,l)},\mathcal{P}_{T}\left(\tilde{\boldsymbol{S}}_{\text{e}}\right)\right\rangle $
is equivalent to $\sum_{i=\rho\left(1-s\right)n_{1}n_{2}+1}^{\rho n_{1}n_{2}}\mathcal{Z}_{\boldsymbol{a}_{i}}$
in distribution. The conditional mean and variance of $\mathcal{Z}_{\boldsymbol{a}_{i}}$
are given as 
\[
\mathbb{E}\left[\mathcal{Z}_{\boldsymbol{a}_{i}}|\boldsymbol{K}\right]=\frac{1}{n_{1}n_{2}}\frac{1}{\sqrt{\omega_{k,l}}}\left\langle \mathcal{P}_{T}\boldsymbol{A}_{(k,l)},\boldsymbol{K}_{\text{e}}\right\rangle ,
\]
where $\boldsymbol{K}_{\text{e}}$ is the enhanced matrix of $\boldsymbol{K}$,
and 
\begin{align*}
 & {\bf Var}\left[\mathcal{Z}_{\boldsymbol{a}_{i}}|\boldsymbol{K}\right]\leq\mathbb{E}\left[\mathcal{Z}_{\boldsymbol{a}_{i}}\mathcal{Z}_{\boldsymbol{a}_{i}}^{*}|\boldsymbol{K}\right]\\
 & \quad=\frac{1}{n_{1}n_{2}}\frac{1}{\omega_{k,l}}\sum_{\boldsymbol{b}\in[n_{1}]\times[n_{2}]}\omega_{\boldsymbol{b}}\left|\left\langle \mathcal{P}_{T}\boldsymbol{A}_{(k,l)},\boldsymbol{A}_{\boldsymbol{b}}\right\rangle \right|^{2}\\
 & \quad\leq\frac{9\mu_{1}^{2}c_{\text{s}}^{2}r^{2}}{n_{1}^{2}n_{2}^{2}},
\end{align*}
where the last inequality follows from \eqref{eq:Mu4BoundViaMu1}.
Besides, from \eqref{eq:UBAbPtAa}, the magnitude of $\mathcal{Z}_{\alpha,\beta}$
can be bounded as follows 
\begin{equation}
\left|\mathcal{Z}_{\alpha,\beta}\right|\leq\frac{3\mu_{1}c_{\text{s}}r}{n_{1}n_{2}}.\label{eq:BoundZalpha_beta_robust}
\end{equation}
%This implies that
%\[
%\frac{2\rho sn_{1}n_{2}V}{\left|\mathcal{Z}_{\boldsymbol{a}_{i}}\right|}\geq\frac{\frac{2\rho s\max\left(\mu_{1}c_{\text{s}},\mu_{4}\right)r}{n_{1}n_{2}}}{\frac{3\mu_{1}c_{\text{s}}r}{n_{1}n_{2}}}\geq\frac{2}{3}\rho s.
%\]

Applying Lemma~\ref{lemma:Bernstein} then yields that with probability
exceeding $1-\left(n_{1}n_{2}\right)^{-4}$, 
\begin{align}
 & \frac{1}{\sqrt{\omega_{k,l}}}\left|\left\langle \boldsymbol{A}_{(k,l)},\mathcal{P}_{T}\left(\tilde{\boldsymbol{S}}_{\text{e}}\right)\right\rangle -\rho\tau\left\langle \mathcal{P}_{T}\boldsymbol{A}_{(k,l)},\boldsymbol{K}_{\text{e}}\right\rangle \right|\nonumber \\
 & \quad\leq c_{13}\mu_{1}c_{\text{s}}r\left(\sqrt{\frac{\rho\tau\log\left(n_{1}n_{2}\right)}{n_{1}n_{2}}}+\frac{\log\left(n_{1}n_{2}\right)}{n_{1}n_{2}}\right)\nonumber \\
 & \quad\leq2c_{13}\mu_{1}c_{\text{s}}r\sqrt{\frac{\rho\tau\log\left(n_{1}n_{2}\right)}{n_{1}n_{2}}}\label{eq:RobustDiff}
\end{align}
for some constant $c_{13}>0$ provided $\rho\tau n_{1}n_{2}\gg\log\left(n_{1}n_{2}\right)$.

The next step is to bound $\frac{\rho\tau}{\sqrt{\omega_{k,l}}}\left\langle \mathcal{P}_{T}\boldsymbol{A}_{(k,l)},\boldsymbol{K}_{\text{e}}\right\rangle $.
For convenience of analysis, we represent $\boldsymbol{K}_{\text{e}}$
as 
\begin{equation}
\boldsymbol{K}_{\text{e}}=\sum_{\boldsymbol{a}\in[n_{1}]\times[n_{2}]}z_{\boldsymbol{a}}\sqrt{\omega_{\boldsymbol{a}}}\boldsymbol{A}_{\boldsymbol{a}},\label{eq:Kenhanced}
\end{equation}
where $z_{\boldsymbol{a}}$'s are independent (not necessarily i.i.d.)
zero-mean random variables satisfying $\left|z_{\boldsymbol{a}}\right|=1$.
Let 
\[
\mathcal{Y}_{\boldsymbol{a}}:=\frac{1}{\sqrt{\omega_{k,l}}}\left\langle \mathcal{P}_{T}\boldsymbol{A}_{(k,l)},z_{\boldsymbol{a}}\sqrt{\omega_{\boldsymbol{a}}}\boldsymbol{A}_{\boldsymbol{a}}\right\rangle ,
\]
then $\mathbb{E}\left[\mathcal{Y}_{\boldsymbol{a}}\right]=0$, \eqref{eq:UBAbPtAa}
and \eqref{eq:Mu4BoundViaMu1} allow us to bound 
\[
\left|\mathcal{Y}_{\boldsymbol{a}}\right|=\frac{1}{\sqrt{\omega_{k,l}}}\left|\left\langle \mathcal{P}_{T}\boldsymbol{A}_{(k,l)},\sqrt{\omega_{\boldsymbol{a}}}\boldsymbol{A}_{\boldsymbol{a}}\right\rangle \right|\leq\frac{3\mu_{1}c_{\text{s}}r}{n_{1}n_{2}},
\]
and 
\begin{align*}
\sum_{\boldsymbol{a}\in[n_{1}]\times[n_{2}]}\mathbb{E}\left[\mathcal{Y}_{\boldsymbol{a}}\mathcal{Y}_{\boldsymbol{a}}^{*}\right] & =\frac{1}{\omega_{k,l}}\sum_{\boldsymbol{a}\in[n_{1}]\times[n_{2}]}\left|\left\langle \mathcal{P}_{T}\boldsymbol{A}_{(k,l)},\sqrt{\omega_{\boldsymbol{a}}}\boldsymbol{A}_{\boldsymbol{a}}\right\rangle \right|^{2}\\
 & \leq\frac{9\mu_{1}^{2}c_{\text{s}}^{2}r^{2}}{n_{1}n_{2}}.
\end{align*}
%Thus, one has $\frac{2\tilde{V}}{\left|\mathcal{Y}_{\boldsymbol{a}}\right|}\geq\frac{2}{3}.$
Applying Lemma~\ref{lemma:Bernstein} suggests that there exists
a constant $c_{14}>0$ such that 
\begin{align*}
\frac{1}{\sqrt{\omega_{k,l}}}\left|\left\langle \mathcal{P}_{T}\boldsymbol{A}_{(k,l)},\boldsymbol{K}_{\text{e}}\right\rangle \right| & =\left|\sum_{\boldsymbol{a}\in[n_{1}]\times[n_{2}]}\mathcal{Y}_{\boldsymbol{a}}\right|\\
 & \leq c_{14}\mu_{1}c_{\text{s}}r\sqrt{\frac{\log\left(n_{1}n_{2}\right)}{n_{1}n_{2}}}
\end{align*}
%Bernstein inequality \cite[Theorem 6]{Gross2011recovering} suggests that for any $t<\frac{2}{3}$, 
%\[
%\mathbb{P}\left\{ \left|\frac{1}{\sqrt{\omega_{k,l}}}\left\langle \mathcal{P}_{T}\boldsymbol{A}_{(k,l)},\boldsymbol{K}_{\text{e}}\right\rangle \right|\geq t\right\} \leq2n_{1}n_{2}\exp\left(-\frac{t^{2}}{4\frac{\max\left\{ \mu_{1}c_{\text{s}},\mu_{4}\right\} r}{n_{1}n_{2}}}\right).
%\]
%Consequently, there exists a constant $c_{24}>0$ such that 
%\[
%\frac{\rho^{2}s^{2}}{\omega_{k,l}}\left|\left\langle \mathcal{P}_{T}\boldsymbol{A}_{(k,l)},\boldsymbol{K}_{\text{e}}\right\rangle \right|^{2}\leq\frac{c_{24}\rho^{2}s^{2} \mu_{1}^2c_{\text{s}}^2 r^2 \log\left(n_{1}n_{2}\right)}{n_{1}n_{2}}
%\]
with high probability provided $n_{1}n_{2}\gg\log(n_{1}n_{2})$. This
together with (\ref{eq:RobustDiff}) suggests that 
\begin{align}
 & \frac{1}{\sqrt{\omega_{k,l}}}\left|\left\langle \boldsymbol{A}_{(k,l)},\mathcal{P}_{T}\left(\tilde{\boldsymbol{S}}_{\text{e}}\right)\right\rangle \right|\nonumber \\
 & \quad\leq\frac{1}{\sqrt{\omega_{k,l}}}\left|\left\langle \boldsymbol{A}_{(k,l)},\mathcal{P}_{T}\left(\tilde{\boldsymbol{S}}_{\text{e}}\right)\right\rangle -\rho\tau\left\langle \mathcal{P}_{T}\boldsymbol{A}_{(k,l)},\boldsymbol{K}_{\text{e}}\right\rangle \right|\nonumber \\
 & \quad\quad\quad\quad+\frac{\rho\tau}{\sqrt{\omega_{k,l}}}\left|\left\langle \mathcal{P}_{T}\boldsymbol{A}_{(k,l)},\boldsymbol{K}_{\text{e}}\right\rangle \right|\nonumber \\
 & \quad\leq c_{15}\mu_{1}c_{\text{s}}r\sqrt{\frac{\rho\tau\log\left(n_{1}n_{2}\right)}{n_{1}n_{2}}}\label{eq:equation120}
\end{align}
for some constant $c_{15}>0$ with high probability.

We still need to bound the deviation of $\tilde{\boldsymbol{S}}_{\text{e}}$
from $\text{sgn}\left(\boldsymbol{S}_{\text{e}}\right)$. Observe
that the difference between them arise from sampling with replacement,
i.e. there are a few entries in $\left\{ \boldsymbol{a}_{i}\mid\rho\left(1-\tau\right)n_{1}n_{2}<i\leq\rho n_{1}n_{2}\right\} $
that either fall within $\Omega^{\text{clean}}$ or have appeared
more than once. A simple Chernoff bound argument (e.g. \cite{Alon2008})
indicates the number of aforementioned conflicts is upper bounded
by $10\log\left(n_{1}n_{2}\right)$ with high probability. That said,
one can find a collection of entry locations $\left\{ \boldsymbol{b}_{1},\cdots,\boldsymbol{b}_{N}\right\} $
such that 
\begin{equation}
\tilde{\boldsymbol{S}}_{\text{e}}-\text{sgn}\left(\boldsymbol{S}_{\text{e}}\right)=\sum_{i=1}^{N}K_{\boldsymbol{b}_{i}}\sqrt{\omega_{\boldsymbol{b}_{i}}}\boldsymbol{A}_{\boldsymbol{b}_{i}},\label{eq:DiffSe_Se_tilde}
\end{equation}
where $N\leq10\log\left(n_{1}n_{2}\right)$ with high probability.
Therefore, we can bound 
\begin{align*}
 & \frac{1}{\sqrt{\omega_{k,l}}}\left|\left\langle \boldsymbol{A}_{(k,l)},\mathcal{P}_{T}\left(\tilde{\boldsymbol{S}}_{\text{e}}-\text{sgn}\left(\boldsymbol{S}_{\text{e}}\right)\right)\right\rangle \right|\\
 & \quad\leq\sum_{i=1}^{N}\frac{1}{\sqrt{\omega_{k,l}}}\left|\left\langle \boldsymbol{A}_{(k,l)},\mathcal{P}_{T}\left(\sqrt{\omega_{\boldsymbol{b}_{i}}}\boldsymbol{A}_{\boldsymbol{b}_{i}}\right)\right\rangle \right|\\
 & \quad\leq N\frac{3\mu_{1}c_{\text{s}}r}{n_{1}n_{2}}\leq\frac{30\mu_{1}c_{\text{s}}r\log(n_{1}n_{2})}{n_{1}n_{2}}.
\end{align*}
following \eqref{eq:UBAbPtAa}. %and hence, for some constant $c_{22}>0$, 
%\[
%\frac{1}{\omega_{k,l}}\left|\left\langle \boldsymbol{A}_{(k,l)},\mathcal{P}_{T}\left(\tilde{\boldsymbol{S}}_{\text{e}}-\text{sgn}\left(\boldsymbol{S}_{\text{e}}\right)\right)\right\rangle \right|^{2}\leq c_{22}\frac{\mu_{1}^{2}c_{\text{s}}^{2}r^{2}\log^{2}\left(n_{1}n_{2}\right)}{n_{1}^{2}n_{2}^{2}} %\leq\frac{c_{20}\rho s\max\left(\mu_{1}c_{\text{s}},\mu_{4}\right)r\log\left(n_{1}n_{2}\right)}{n_{1}n_{2}}
%\]
%holds with high probability, provided that $\rho sn_{1}n_{2}\geq c_{23}\mu_{1}c_{\text{s}}r\log\left(n_{1}n_{2}\right)$
%for some $c_{23}>0$.
Putting the above inequality and \eqref{eq:equation120} together
yields that for every $\left(k,l\right)\in\left[n_{1}\right]\times\left[n_{2}\right]$,
\begin{align*}
 & \frac{1}{\sqrt{\omega_{k,l}}}\left|\left\langle \boldsymbol{A}_{(k,l)},\mathcal{P}_{T}\left(\text{sgn}\left(\boldsymbol{S}_{\text{e}}\right)\right)\right\rangle \right|\\
 & \quad\leq\frac{1}{\sqrt{\omega_{k,l}}}\left|\left\langle \boldsymbol{A}_{(k,l)},\mathcal{P}_{T}\left(\tilde{\boldsymbol{S}}_{\text{e}}-\text{sgn}\left(\boldsymbol{S}_{\text{e}}\right)\right)\right\rangle \right|\\
 & \quad\quad\quad\quad+\frac{1}{\sqrt{\omega_{k,l}}}\left|\left\langle \boldsymbol{A}_{(k,l)},\mathcal{P}_{T}\left(\tilde{\boldsymbol{S}}_{\text{e}}\right)\right\rangle \right|\\
 & \quad\leq c_{15}\mu_{1}c_{\text{s}}r\sqrt{\frac{\rho\tau\log\left(n_{1}n_{2}\right)}{n_{1}n_{2}}}+\frac{30\mu_{1}c_{\text{s}}r\log(n_{1}n_{2})}{n_{1}n_{2}}\\
 & \quad\leq c_{9}\mu_{1}c_{\text{s}}r\sqrt{\frac{\rho\tau\log\left(n_{1}n_{2}\right)}{n_{1}n_{2}}}
\end{align*}
for some constant $c_{9}>0$ provided $\rho\tau n_{1}n_{2}>\log\left(n_{1}n_{2}\right)$.
This completes the proof. %\end{IEEEproof}

\section{Proof of Lemma \ref{lemma-PTperp_sparse}\label{sec:Proof-of-Lemma-Pperp_sparse}}

Consider the model of $\text{sgn}(\boldsymbol{S})$, $\boldsymbol{K}$
and $\tilde{\boldsymbol{S}}_{\text{e}}$ as introduced in the proof
of Lemma \ref{lemma-Psparse} in Appendix \ref{sec:Proof-of-Lemma-Psparse}.
For any $\left(\alpha,\beta\right)\in[n_{1}]\times[n_{2}]$, define
\[
\tilde{\mathcal{Z}}_{\alpha,\beta}:=\mathcal{A}_{\alpha,\beta}\left(\boldsymbol{K}_{\text{e}}\right)=\sqrt{\omega_{\alpha,\beta}}K_{\alpha,\beta}\boldsymbol{A}_{\alpha,\beta}.
\]
With this notation, we can see that $\tilde{\mathcal{Z}}_{\boldsymbol{a}_{i}}$'s
are conditionally independent given $\boldsymbol{K}$, and satisfy
\begin{align*}
\mathbb{E}\left[\tilde{\mathcal{Z}}_{\boldsymbol{a}_{i}}|\boldsymbol{K}\right] & =\frac{1}{n_{1}n_{2}}\sum_{\left(\alpha,\beta\right)\in[n_{1}]\times[n_{2}]}\sqrt{\omega_{\alpha,\beta}}\boldsymbol{A}_{\alpha,\beta}K_{\alpha,\beta}\\
 & =\frac{1}{n_{1}n_{2}}\boldsymbol{K}_{\text{e}},
\end{align*}
\[
\left\Vert \tilde{\mathcal{Z}}_{\alpha,\beta}\right\Vert =\left\Vert \sqrt{\omega_{\alpha,\beta}}\boldsymbol{A}_{\alpha,\beta}\right\Vert =1,
\]
and 
\begin{align*}
\left\Vert \mathbb{E}\left[\tilde{\mathcal{Z}}_{\boldsymbol{a}_{i}}\tilde{\mathcal{Z}}_{\boldsymbol{a}_{i}}^{*}|\boldsymbol{K}\right]\right\Vert  & \leq\frac{1}{n_{1}n_{2}}\sum_{\left(\alpha,\beta\right)\in[n_{1}]\times[n_{2}]}\left\Vert \omega_{\alpha,\beta}\boldsymbol{A}_{\alpha,\beta}\boldsymbol{A}_{\alpha,\beta}^{*}\right\Vert \\
 & =1.
\end{align*}
%Therefore,
%\[
%\frac{2\rho sn_{1}n_{2}V}{\left\Vert \tilde{\mathcal{Z}}_{\alpha,\beta}\right\Vert }\geq2\rho sn_{1}n_{2}.
%\]

Since $\tilde{\boldsymbol{S}}_{\text{e}}=\sum_{i=\left(1-\tau\right)\rho n_{1}n_{2}+1}^{\rho n_{1}n_{2}}\tilde{\mathcal{Z}}_{\boldsymbol{a}_{i}}$,
applying Lemma~\ref{lemma:Bernstein} implies that conditioned on
$\boldsymbol{K}$, there exists a constant $c_{16}>0$ such that %Bernstein inequality \cite[Theorem 6]{Gross2011recovering}
%implies that for any $t<2\rho sn_{1}n_{2}$, one has, conditional
%on $\boldsymbol{K}$, that 
%\begin{align*}
%\mathbb{P}\left(\left\Vert \tilde{\boldsymbol{S}}_{\text{e}}-\rho s\boldsymbol{K}_{\text{e}}\right\Vert >t\right) & \leq2n_{1}n_{2}\exp\left(-\frac{t^{2}}{4\rho sn_{1}n_{2}V}\right).
%\end{align*}
%Therefore, 
\begin{align}
\left\Vert \tilde{\boldsymbol{S}}_{\text{e}}-\rho\tau\boldsymbol{K}_{\text{e}}\right\Vert <\sqrt{c_{16}\rho\tau n_{1}n_{2}\log\left(n_{1}n_{2}\right)}\label{eq:NormSe-minus-Ke}
\end{align}
with probability at least than $1-n_{1}^{-5}n_{2}^{-5}$.

The next step is to bound the operator norm of $\rho\tau\boldsymbol{K}_{\text{e}}$.
Recall the decomposition form of $\boldsymbol{K}_{\text{e}}$ in \eqref{eq:Kenhanced}.
%For convenience of analysis, we write 
%\[
%\boldsymbol{K}_{\text{e}}=\sum_{\boldsymbol{a}\in[n_{1}]\times[n_{2}]}z_{\boldsymbol{a}}\sqrt{\omega_{\boldsymbol{a}}}\boldsymbol{A}_{\boldsymbol{a}},
%\]
%for a collection of independent (not necessarily i.i.d.) random variables
%$z_{\boldsymbol{a}}$'s satisfying $\left|z_{\boldsymbol{a}}\right|=1$
%and $\mathbb{E}z_{\boldsymbol{a}}=0$. 
Let $\mathcal{Y}_{\boldsymbol{a}}:=z_{\boldsymbol{a}}\sqrt{\omega_{\boldsymbol{a}}}\boldsymbol{A}_{\boldsymbol{a}}$,
then we have $\mathbb{E}\left[\mathcal{Y}_{\boldsymbol{a}}\right]=0$,
$\left\Vert \mathcal{Y}_{\boldsymbol{a}}\right\Vert =1$, and 
\[
\left\Vert \sum_{\boldsymbol{a}\in[n_{1}]\times[n_{2}]}\mathbb{E}\mathcal{Y}_{\boldsymbol{a}}\mathcal{Y}_{\boldsymbol{a}}^{*}\right\Vert =\left\Vert \sum_{\boldsymbol{a}\in[n_{1}]\times[n_{2}]}\omega_{\boldsymbol{a}}\boldsymbol{A}_{\boldsymbol{a}}\boldsymbol{A}_{\boldsymbol{a}}^{*}\right\Vert \leq n_{1}n_{2}
\]
%which gives 
%\[
%\frac{2\tilde{V}}{\left\Vert \mathcal{Y}_{\boldsymbol{a}}\right\Vert }\geq n_{1}n_{2}.
%\]
Therefore, applying Lemma~\ref{lemma:Bernstein} yields that there
exists a constant $c_{17}>0$ such that %Bernstein inequality yields that for any $t\leq n_{1}n_{2}$,
%\[
%\mathbb{P}\left\{ \left\Vert \boldsymbol{K}_{\text{e}}\right\Vert >t\right\} =\mathbb{P}\left\{ \left\Vert \sum_{\boldsymbol{a}\in[n_{1}]\times[n_{2}]}\mathcal{Y}_{\boldsymbol{a}}\right\Vert >t\right\} \leq2n_{1}n_{2}\exp\left(-\frac{t^{2}}{4n_{1}n_{2}}\right)
%\]
%or, more simply, 
\[
\left\Vert \boldsymbol{K}_{\text{e}}\right\Vert \leq\sqrt{c_{17}n_{1}n_{2}\log\left(n_{1}n_{2}\right)}
\]
with high probability. This and (\ref{eq:NormSe-minus-Ke}), taken
collectively, yield 
\[
\left\Vert \tilde{\boldsymbol{S}}_{\text{e}}\right\Vert \leq\left\Vert \tilde{\boldsymbol{S}}_{\text{e}}-\rho\tau\boldsymbol{K}_{\text{e}}\right\Vert +\rho\tau\left\Vert \boldsymbol{K}_{\text{e}}\right\Vert <2\sqrt{c_{18}\rho\tau n_{1}n_{2}\log\left(n_{1}n_{2}\right)}
\]
with high probability, where $c_{18}=\max\{c_{16},c_{17}\}$. On the
other hand, (\ref{eq:DiffSe_Se_tilde}) implies that, 
\begin{align*}
\left\Vert \tilde{\boldsymbol{S}}_{\text{e}}-\text{sgn}\left(\boldsymbol{S}_{\text{e}}\right)\right\Vert  & \leq\sum_{i=1}^{N}\left\Vert \sqrt{\omega_{\boldsymbol{b}_{i}}}\boldsymbol{A}_{\boldsymbol{b}_{i}}\right\Vert =N\leq10\log\left(n_{1}n_{2}\right)\\
 & \leq\sqrt{c_{18}\rho\tau n_{1}n_{2}\log\left(n_{1}n_{2}\right)}
\end{align*}
with high probability, provided $\rho\tau n_{1}n_{2}>100\log\left(n_{1}n_{2}\right)/c_{18}$.
Consequently, for a sufficiently small constant $\tau$, 
\begin{align*}
 & \left\Vert \mathcal{P}_{T^{\perp}}\left(\lambda\mbox{sgn}\left(\boldsymbol{S}_{\text{e}}\right)\right)\right\Vert \leq\lambda\left\Vert \mbox{sgn}\left(\boldsymbol{S}_{\text{e}}\right)\right\Vert \\
 & \quad\leq\lambda\left\Vert \tilde{\boldsymbol{S}}_{\text{e}}-\mbox{sgn}\left(\boldsymbol{S}_{\text{e}}\right)\right\Vert +\lambda\left\Vert \tilde{\boldsymbol{S}}_{\text{e}}\right\Vert \\
 & \quad\leq3\lambda\sqrt{c_{18}\rho\tau n_{1}n_{2}\log\left(n_{1}n_{2}\right)}\\
 & \quad=3\sqrt{c_{18}\tau}\leq\frac{1}{8}
\end{align*}
with probability exceeding $1-n_{1}^{-5}n_{2}^{-5}$.

\section{Proof of Theorem \ref{theorem-EMaC-Noisy}\label{sec:Proof-of-theorem-EMaC-Noisy}}

We prove this theorem under the conditions of Lemma \ref{lemma-Dual-Certificate},
i.e. \eqref{eq:WellConditionPtAomegaPt}--\eqref{eq:NormWTPerp}.
%that is,
%\begin{equation}
%\left\Vert \mathcal{P}_{T}\mathcal{A}\mathcal{P}_{T}-\frac{n_{1}n_{2}}{m}\mathcal{P}_{T}\mathcal{A}_{\Omega}\mathcal{P}_{T}\right\Vert \leq\frac{1}{2},
%\end{equation}
%and there exists a matrix $\boldsymbol{W}$ that obeys
%\begin{equation}
%\mathcal{A}'_{\Omega^{\perp}}\left(\boldsymbol{U}\boldsymbol{V}^{*}+\boldsymbol{W}\right)=0,
%\end{equation}
%\begin{equation}
%\left\Vert \mathcal{P}_{T}\left(\boldsymbol{W}\right)\right\Vert _{\mathrm{F}}\leq\frac{1}{2n_{1}^{2}n_{2}^{2}},
%\end{equation}
%and
%\begin{equation}
%\left\Vert \mathcal{P}_{T^{\perp}}\left(\boldsymbol{W}\right)\right\Vert \leq\frac{1}{2}.
%\end{equation}
Note that these conditions are satisfied with high probability, as
we have shown in the proof of Theorem \ref{theorem-EMaC-noiseless}.

Denote by $\hat{\boldsymbol{X}}_{\text{e}}=\boldsymbol{X}_{\text{e}}+\boldsymbol{H}_{\text{e}}$
the solution to Noisy-EMaC. By writing $\boldsymbol{H}_{\text{e}}=\mathcal{A}_{\Omega}\left(\boldsymbol{H}_{\text{e}}\right)+\mathcal{A}_{\Omega^{\perp}}\left(\boldsymbol{H}_{\text{e}}\right)$,
one can obtain 
\begin{align}
\|\boldsymbol{X}_{\text{e}}\|_{*} & \geq\|\hat{\boldsymbol{X}}_{\text{e}}\|_{*}=\|\boldsymbol{X}_{\text{e}}+\boldsymbol{H}_{\text{e}}\|_{*}\nonumber \\
 & \geq\|\boldsymbol{X}_{\text{e}}+\mathcal{A}_{\Omega^{\perp}}(\boldsymbol{H}_{\text{e}})\|_{*}-\|\mathcal{A}_{\Omega}(\boldsymbol{H}_{\text{e}})\|_{*}.\label{cone_constraint}
\end{align}
The term $\left\Vert \mathcal{A}_{\Omega}\left(\boldsymbol{H}_{\text{e}}\right)\right\Vert _{\text{F}}$
can be bounded using the triangle inequality as 
\begin{equation}
\left\Vert \mathcal{A}_{\Omega}\left(\boldsymbol{H}_{\text{e}}\right)\right\Vert _{\text{F}}\leq\left\Vert \mathcal{A}_{\Omega}\left(\hat{\boldsymbol{X}}_{\text{e}}-\boldsymbol{X}_{\text{e}}^{\text{o}}\right)\right\Vert _{\text{F}}+\left\Vert \mathcal{A}_{\Omega}\left(\boldsymbol{X}_{\text{e}}-\boldsymbol{X}_{\text{e}}^{\text{o}}\right)\right\Vert _{\text{F}}.
\end{equation}
Since the constraint of Noisy-EMaC requires $\left\Vert \mathcal{P}_{\Omega}\left(\hat{\boldsymbol{X}}-\boldsymbol{X}^{\text{o}}\right)\right\Vert _{\text{F}}\leq\delta$
and $\left\Vert \mathcal{P}_{\Omega}\left(\boldsymbol{X}-\boldsymbol{X}^{\text{o}}\right)\right\Vert _{\text{F}}\leq\delta$,
the Hankel structure of the enhanced form allows us to bound $\left\Vert \mathcal{A}_{\Omega}\left(\hat{\boldsymbol{X}}_{\text{e}}-\boldsymbol{X}_{\text{e}}^{\text{o}}\right)\right\Vert _{\text{F}}\leq\sqrt{n_{1}n_{2}}\delta$
and $\left\Vert \mathcal{A}_{\Omega}\left(\boldsymbol{X}_{\text{e}}-\boldsymbol{X}_{\text{e}}^{\text{o}}\right)\right\Vert _{\text{F}}\leq\sqrt{n_{1}n_{2}}\delta$,
leading to 
\[
\left\Vert \mathcal{A}_{\Omega}\left(\boldsymbol{H}_{\text{e}}\right)\right\Vert _{\text{F}}\leq2\sqrt{n_{1}n_{2}}\delta.
\]

i) Suppose first that $\boldsymbol{H}_{\text{e}}$ satisfies 
\begin{equation}
\left\Vert \mathcal{P}_{T}\mathcal{A}_{\Omega^{\perp}}\left(\boldsymbol{H}_{\text{e}}\right)\right\Vert _{\text{F}}\leq\frac{n_{1}^{2}n_{2}^{2}}{2}\left\Vert \mathcal{P}_{T^{\perp}}\mathcal{A}_{\Omega^{\perp}}\left(\boldsymbol{H}_{\text{e}}\right)\right\Vert _{\text{F}}.\label{eq:AssumptionAperp}
\end{equation}
Applying the same analysis as for (\ref{eq:DualXeH}) allows us to
bound the perturbation $\mathcal{A}_{\Omega^{\perp}}(\boldsymbol{H}_{\text{e}})$
as follows 
\begin{align*}
\left\Vert \boldsymbol{X}_{\text{e}}+\mathcal{A}_{\Omega^{\perp}}(\boldsymbol{H}_{\text{e}})\right\Vert _{*} & \geq\left\Vert \boldsymbol{X}_{\text{e}}\right\Vert _{*}+\frac{1}{4}\left\Vert \mathcal{P}_{T^{\perp}}\mathcal{A}_{\Omega^{\perp}}(\boldsymbol{H}_{\text{e}})\right\Vert _{\text{F}}.
\end{align*}
Combining this with \eqref{cone_constraint}, we have 
\begin{align*}
\left\Vert \mathcal{P}_{T^{\perp}}\mathcal{A}_{\Omega^{\perp}}(\boldsymbol{H}_{\text{e}})\right\Vert _{\text{F}} & \leq4\|\mathcal{A}_{\Omega}(\boldsymbol{H}_{\text{e}})\|_{*}\\
 & \leq4\sqrt{n_{1}n_{2}}\|\mathcal{A}_{\Omega}(\boldsymbol{H}_{\text{e}})\|_{\text{F}}\leq8n_{1}n_{2}\delta.
\end{align*}
Furthermore, the inequality (\ref{eq:AssumptionAperp}) indicates
that 
\begin{equation}
\left\Vert \mathcal{P}_{T}\mathcal{A}_{\Omega^{\perp}}\left(\boldsymbol{H}_{\text{e}}\right)\right\Vert _{\text{F}}\leq4n_{1}^{3}n_{2}^{3}\left\Vert \mathcal{P}_{T^{\perp}}\mathcal{A}_{\Omega^{\perp}}\left(\boldsymbol{H}_{\text{e}}\right)\right\Vert _{\text{F}}.
\end{equation}
Therefore, combining all the above results give 
\begin{align*}
\small\|\boldsymbol{H}_{\text{e}}\|_{\text{F}} & \leq\|\mathcal{A}_{\Omega}(\boldsymbol{H}_{\text{e}})\|_{\text{F}}+\left\Vert \mathcal{P}_{T}\mathcal{A}_{\Omega^{\perp}}\left(\boldsymbol{H}_{\text{e}}\right)\right\Vert _{\text{F}}+\left\Vert \mathcal{P}_{T^{\perp}}\mathcal{A}_{\Omega^{\perp}}\left(\boldsymbol{H}_{\text{e}}\right)\right\Vert _{\text{F}}\\
 & \leq\left\{ 2\sqrt{n_{1}n_{2}}+8n_{1}n_{2}+4n_{1}^{3}n_{2}^{3}\right\} \delta\\
 & \leq5n_{1}^{3}n_{2}^{3}\delta
\end{align*}
for sufficiently large $n_{1}$ and $n_{2}$.

ii) On the other hand, consider the situation where 
\begin{equation}
\left\Vert \mathcal{P}_{T}\mathcal{A}_{\Omega^{\perp}}\left(\boldsymbol{H}_{\text{e}}\right)\right\Vert _{\text{F}}>\frac{n_{1}^{2}n_{2}^{2}}{2}\left\Vert \mathcal{P}_{T^{\perp}}\mathcal{A}_{\Omega^{\perp}}\left(\boldsymbol{H}_{\text{e}}\right)\right\Vert _{\text{F}}.\label{eq:AssumptionAperp-2}
\end{equation}
Employing similar argument as in Part (2) of Appendix \ref{sec:Proof-of-Lemma-Dual-Certificate}
yields that (\ref{eq:AssumptionAperp-2}) can only arise when $\mathcal{A}_{\Omega^{\perp}}\left(\boldsymbol{H}_{\text{e}}\right)=0$.
In this case, one has 
\begin{align*}
\small\|\boldsymbol{H}_{\text{e}}\|_{\text{F}} & \leq\|\mathcal{A}_{\Omega}(\boldsymbol{H}_{\text{e}})\|_{\text{F}}+\left\Vert \mathcal{A}_{\Omega^{\perp}}\left(\boldsymbol{H}_{\text{e}}\right)\right\Vert _{\text{F}}\\
 & =\|\mathcal{A}_{\Omega}(\boldsymbol{H}_{\text{e}})\|_{\text{F}}\leq2\sqrt{n_{1}n_{2}}\delta.
\end{align*}
concluding the proof.

%\section{Proof of Theorem~\ref{theorem-EMaC-Hankel}\label{proof-EMaC-Hankel}}
%In order to extend the results to structured Hankel matrix completion,
%from the proof of Theorem~\ref{theorem-EMaC-noiseless} it is sufficient
%to have the first two conditions in \eqref{eq:IncoherenceT_W} to
%hold for general Hankel matrices. The proof is done by recognizing
%these two conditions are equivalent to \eqref{eq:IncohrenceUU_Hankel}.
% 

 \bibliographystyle{IEEEtran}
\bibliography{bibfileSparseMatrixPencil}

\begin{IEEEbiographynophoto}{Yuxin Chen} (S'09) received the B.S. in Microelectronics with High Distinction from Tsinghua University in 2008, the M.S. in Electrical and Computer Engineering from the University of Texas at Austin in 2010, and the M.S. in Statistics from Stanford University in 2013. He is currently a Ph.D. candidate in the Department of Electrical Engineering at Stanford University. His research interests include information theory, compressed sensing, network science and high-dimensional statistics. \end{IEEEbiographynophoto} 

\begin{IEEEbiographynophoto}{Yuejie Chi} (S'09-M'12) received the Ph.D. degree in Electrical
Engineering from Princeton University in 2012, and the B.E. (Hon.)
degree in Electrical Engineering from Tsinghua University, Beijing,
China, in 2007. Since September 2012, she has been an assistant professor
with the department of Electrical and Computer Engineering and the
department of Biomedical Informatics at the Ohio State University.

She is the recipient of the IEEE Signal Processing Society Young Author
Best Paper Award in 2013 and the Best Paper Award at the IEEE International
Conference on Acoustics, Speech, and Signal Processing (ICASSP) in
2012. She received the Ralph E. Powe Junior Faculty Enhancement Award
from Oak Ridge Associated Universities in 2014, a Google Faculty Research
Award in 2013, the Roberto Padovani scholarship from Qualcomm Inc.
in 2010, and an Engineering Fellowship from Princeton University in
2007. She has held visiting positions at Colorado State University,
Stanford University and Duke University, and interned at Qualcomm
Inc. and Mitsubishi Electric Research Lab. Her research interests
include high-dimensional data analysis, statistical signal processing,
machine learning and their applications in communications, networks,
imaging and bioinformatics. \end{IEEEbiographynophoto}

\end{document}